\documentclass[12pt]{article}
\usepackage{graphicx}%
\usepackage{multirow}%
\usepackage{amsmath,amssymb,amsfonts}%
\numberwithin{equation}{section}
\usepackage{amsthm}%
\newtheorem{theorem}{Theorem}[section]
\newtheorem{lemma}{Lemma}[section]
\newtheorem{remark}{Remark}[section]
\newtheorem{corollary}{Corollary}[section]

\newtheorem{assumption}{Assumption}
\newtheorem{example}{Example}
\usepackage{enumitem}
\usepackage{mathrsfs}%
\usepackage[title]{appendix}%
\usepackage{xcolor}%
\usepackage{textcomp}%
\usepackage{manyfoot}%
\usepackage{booktabs}%
\usepackage{float}
\usepackage{listings}%
\usepackage{textcomp}
\usepackage{eufrak}
\usepackage[mathscr]{eucal}
\usepackage{tikz}
\usepackage{floatpag}
\usepackage{caption}
\usepackage{graphicx}%
\usepackage{natbib}
\usepackage{url} 
\usepackage{mathtools}
\usepackage{bbm}
\usepackage[normalem]{ulem}
\usepackage{threeparttable}
\usepackage{footnote}
\usepackage{tablefootnote}
\usepackage{dsfont}
\usepackage{subcaption}
\usepackage[ruled,linesnumbered]{algorithm2e}
\usetikzlibrary{arrows,chains,matrix,positioning,scopes}
\usepackage{hyperref}
\hypersetup{colorlinks=true,linkcolor=blue, citecolor=blue, linktocpage}
\mathtoolsset{showonlyrefs}
\allowdisplaybreaks
\newcommand{\blind}{1}

\makeatletter
\tikzset{join/.code=\tikzset{after node path={%
\ifx\tikzchainprevious\pgfutil@empty\else(\tikzchainprevious)%
edge[every join]#1(\tikzchaincurrent)\fi}}}
\makeatother

\tikzset{>=stealth',every on chain/.append style={join},
         every join/.style={}}
\tikzstyle{labeled}=[execute at begin node=$\scriptstyle,
   execute at end node=$]

\makeatletter
\newcommand*{\Rom}[1]{\expandafter\@slowromancap\romannumeral #1@}
\makeatother

\begin{document}

\def\spacingset#1{\renewcommand{\baselinestretch}%
{#1}\small\normalsize} \spacingset{1}

 \title{\bf Transfer learning for high-dimensional Factor-augmented sparse linear model}
  \date{}
\if1\blind
{
 
  \author{Bo Fu\\
    School of Mathematics and Statistics\\ Xi'an Jiaotong University\\
    and \\
    Dandan Jiang* \\School of Mathematics and Statistics\thanks{Corresponding author.}\\
    Xi'an Jiaotong University}
   
  \maketitle
} \fi

\if0\blind
{
  \bigskip
  \bigskip
  \bigskip

  \maketitle
} \fi

\bigskip
\begin{abstract}
We propose Trans-FARM, a transfer learning method for high-dimensional sparse linear models augmented with latent factors, addressing settings where covariates exhibit strong correlations and latent factor structures, which is motivated by the demand for stable and reliable estimations in economic and financial applications. Our framework jointly corrects for latent factor-induced effects and mitigates multicollinearity, reducing potential model misspecification relative to conventional linear sparse regression. When the target dataset is limited by sample size but multiple heterogeneous auxiliary datasets are accessible, we propose transfer-augmented estimators that leverage auxiliary information to improve estimation accuracy. We establish non-asymptotic $\ell_1$- and $\ell_2$-error bounds for the oracle estimator with a known informative source set and derive corresponding minimax lower bounds. This estimator is minimax optimal under suitable regularity and source-target similarity conditions. The minimax lower bounds also reveal a dimension-dependent term arising from observational non-identifiability within the model class, even when the informative domains share the same regression coefficients and the sample sizes are sufficiently large. To guard against negative transfer, we develop a data-driven source detection algorithm with provable consistency in selecting informative auxiliary datasets. In addition, we provide a procedure for constructing simultaneous confidence intervals for the regression coefficients of interest. Simulation studies and a real data analysis illustrate the benefits of the proposed approach in settings with latent factors and heterogeneous sources.
\end{abstract}

\vfill

\newpage
\spacingset{1.5} 

\addtocontents{toc}{\protect\setcounter{tocdepth}{0}}
\section{Introduction}\label{sec1}
In many modern statistical and machine learning applications, data are often collected from multiple related sources that share similar yet non-identical underlying structures. Directly combining these heterogeneous datasets may lead to biased estimation due to distributional discrepancies, whereas analyzing the target dataset alone can suffer from a limited sample size. Transfer learning offers a principled framework to leverage auxiliary information from source datasets to improve estimation performance on the target task. A growing body of work seeks to develop statistical foundations for transfer learning. Examples include nonparametric classification \citep{cai2021transfer}, transfer reinforcement learning \citep{chai2025deep}, Gaussian graphical models \citep{li2023transfer}, and high-dimensional parametric regression models such as linear models \citep{li2022transfer}, generalized linear models \citep{bastani2021predicting,tian2023transfer}, smoothed quantile regression \citep{zhang2025transfer}, multinomial regression \citep{yang2025debiased}, and adaptive Huber regression models \citep{yang2025communication}. In these parametric regression models, the similarity between the source and target datasets is typically quantified via the distance between their regression coefficients.

However, those existing transfer learning methods for parametric regression mentioned above rely critically on weak dependence on the variables that ensure estimation consistency. In many applications--particularly in finance, economics, and biomedicine--the covariates exhibit strong cross-sectional dependence and are often driven by latent common factors. Under such structures, in the high-dimensional linear regression model, the conditions (e.g., the irrepresentable condition) required for variable selection and estimation consistency of Lasso-type procedures are violated \citep{fan2020factor}. Moreover, when the latent factors contribute directly to the response, the linear regression model based only on the observed covariates becomes misspecified. As a result, transfer learning procedures built upon standard high-dimensional parametric regression may perform unreliably in the presence of latent factor effects or highly correlated covariates. This structural issue was considered in \cite{fan2024latent}, where they proposed the Factor Augmented (sparse linear) Regression Model (FARM), defined as follows:
\begin{equation}
\begin{aligned}\label{eq1.1}
    y=\boldsymbol{f}^\top\boldsymbol{\gamma}^*+\boldsymbol{u}^\top\boldsymbol{\beta}^*+\mathcal{E},\qquad \boldsymbol{x}=\boldsymbol{B}\boldsymbol{f}+\boldsymbol{u},
\end{aligned}
\end{equation}
where $y\in\mathbb{R}$ is a scalar response variable, $\boldsymbol{x}\in\mathbb{R}^{p}$ is the covariate, $\boldsymbol{f}\in\mathbb{R}^{r}$ is the latent factors, $\boldsymbol{B}\in\mathbb{R}^{p\times r}$ is the factor loading matrix, and $\boldsymbol{u}\in\mathbb{R}^{p}$ is an idiosyncratic component. The parameter $\boldsymbol{\beta}^*\in\mathbb{R}^{p}$ is the true regression parameter vector that is of interest, and $\boldsymbol{\gamma}^*\in\mathbb{R}^{r}$ quantifies the contribution of the latent factor $\boldsymbol{f}$. The FARM model can be viewed as a generalization of the sparse linear model, as it can be rewritten as:
\begin{align*}
    y=\boldsymbol{f}^\top\boldsymbol{\varphi}^*+\boldsymbol{x}^\top\boldsymbol{\beta}^*+\mathcal{E},\qquad\boldsymbol{x}=\boldsymbol{B}\boldsymbol{f}+\boldsymbol{u},
\end{align*}
where $\boldsymbol{\varphi}^*=\boldsymbol{\gamma}^*-\boldsymbol{B}^\top\boldsymbol{\beta}^*$ measures the extra contribution of $\boldsymbol{f}$ to $y$ beyond $\boldsymbol{x}$. It enhances sparse linear regression by incorporating informative directions spanned by the common factors $\boldsymbol{f}$. From another point of view, the FARM model can be viewed as a generalization of the factor regression model \citep{stock2002forecasting,bing2021prediction} as it utilizes additional information in $\boldsymbol{u}$. Therefore, the strong performance of FARM arises from its ability to jointly capture both the common factors $\boldsymbol{f}$ and idiosyncratic components $\boldsymbol{u}$ in explaining the response. 

Although FARM accommodates highly correlated covariates and latent factors by separating common factors and idiosyncratic components, in practice, the target dataset may have a limited sample size, and related auxiliary datasets often contain useful information that can enhance estimation accuracy. Since the regression coefficient $\boldsymbol{\beta}^*$ quantifies the contribution of the covariates $\boldsymbol{x}$ to the outcome $y$, this makes its estimation and inference a quantity of significant importance \citep{guo2022doubly}. Exploiting such auxiliary information, however, requires accounting for heterogeneity in both regression coefficients and factor structures across domains. Effective transfer therefore requires controlling factor estimation errors and source-target discrepancies while retaining the statistical gains from additional observations. Establishing minimax lower bounds further clarifies when these gains are optimal and whether latent factor structures impose limitations that increasing the sample size alone cannot overcome.
\begin{figure}[!ht]
    \centering
    \includegraphics[width=0.9\textwidth]{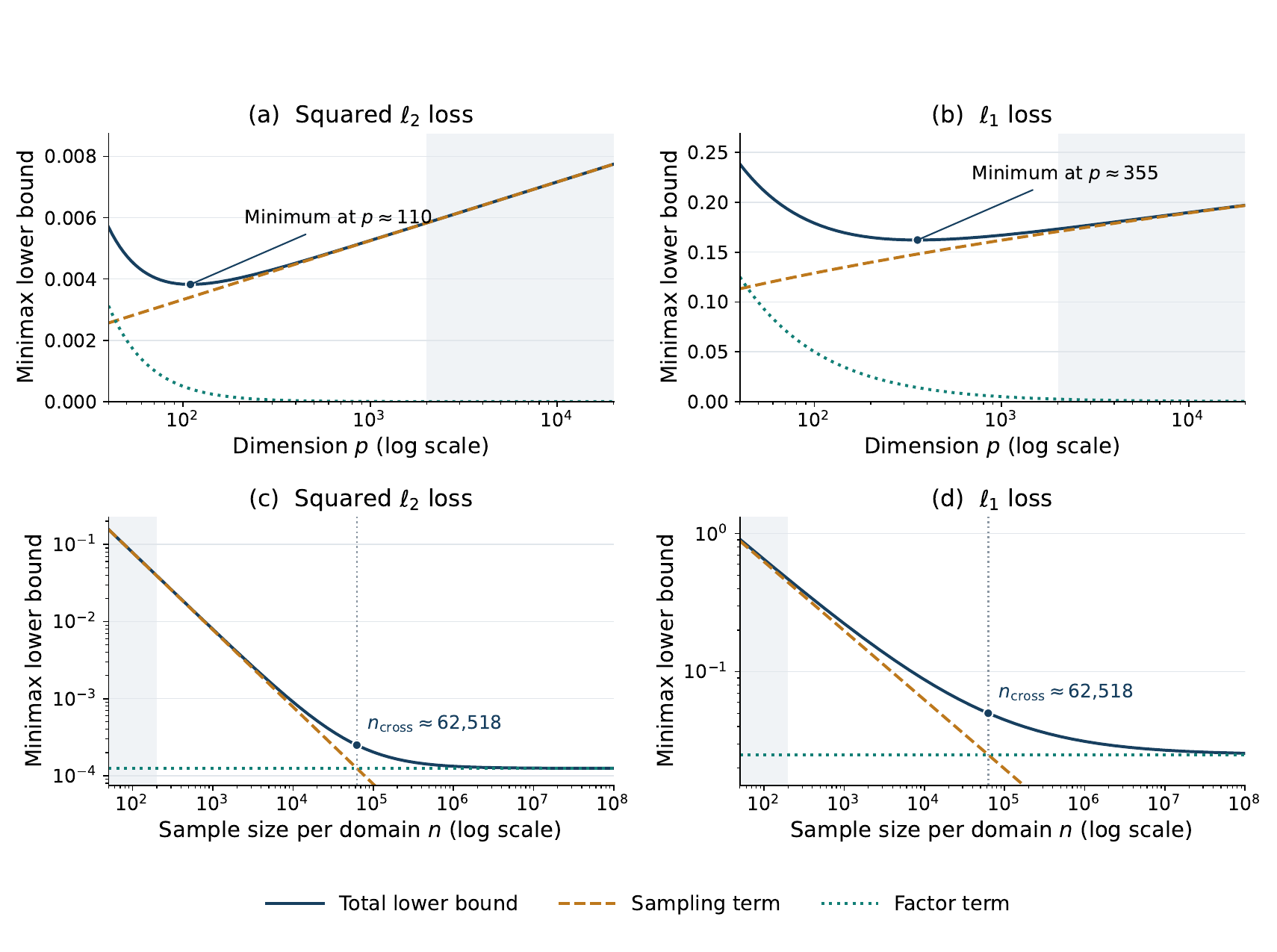}
    \captionsetup{font=footnotesize}
    \caption{Illustrations of the minimax lower bounds under homogeneous regression coefficients ($\eta=0$), with $s=5$, $b_{n,\min}=1$, and two informative sources. Panels (a)--(b) vary $p$ with $n^{(0)}=n^{(1)}=n^{(2)}=2000$, whereas panels (c)--(d) vary the common sample size $n$ with $p=200$. The left and right panels correspond to squared $\ell_2$ and $\ell_1$ losses, respectively. The common multiplicative constant is omitted. Shading marks $p\geq n^{(0)}$.}
    \label{fig:minimax-illustration}
\end{figure}

\begin{figure}[!ht]
    \centering
    \includegraphics[width=0.9\textwidth]{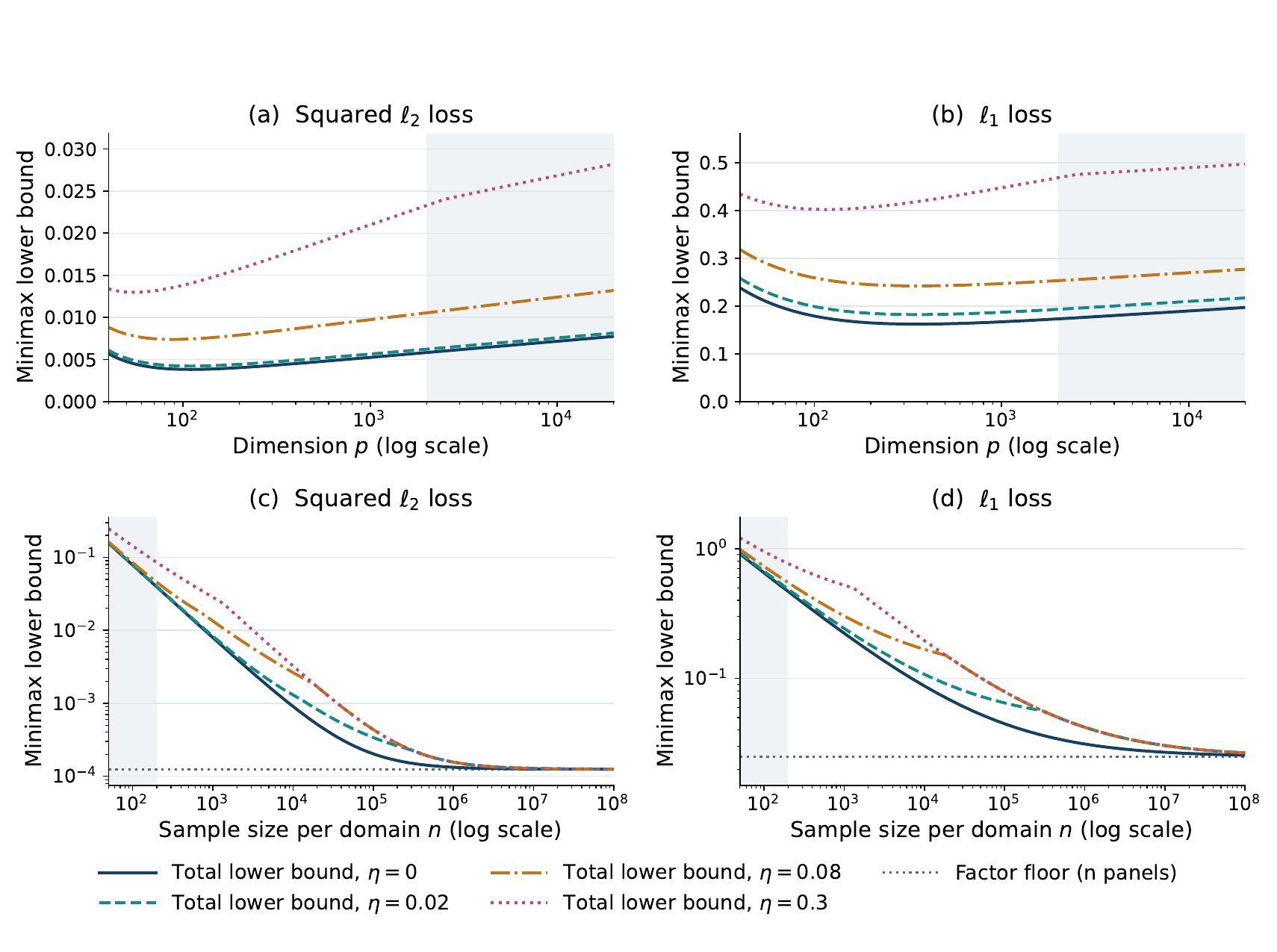}
    \captionsetup{font=footnotesize}
    \caption{Illustrations of the minimax lower bounds under heterogeneous regression coefficients, with $\eta\in\{0,0.02,0.08,0.3\}$, $s=5$, $b_{n,\min}=1$, and two informative sources. Panels (a)--(b) vary $p$ with $n^{(0)}=n^{(1)}=n^{(2)}=2000$, whereas panels (c)--(d) vary the common sample size $n$ with $p=200$. The left and right panels correspond to squared $\ell_2$ and $\ell_1$ losses, respectively. Shading marks $p\geq n^{(0)}$.}
    \label{fig:minimax-illustration2}
\end{figure}

We develop \textbf{Trans}fer learning procedure for \textbf{FARM} (Trans-FARM), a transfer learning procedure that allows information to be shared across multiple related datasets while accounting for the heterogeneity in both the factor structure and the idiosyncratic components. By removing the effects of strong correlations among covariates and the additional variations induced by latent factors, Trans-FARM yields more accurate estimation of regression coefficients. Compared with standard FARM, it improves estimation by borrowing similarity from informative sources; compared with transfer learning methods for linear models \citep{li2022transfer,tian2023transfer}, it achieves a convergence rate that matches the common transfer learning methods \citep{li2022transfer,tian2023transfer} when factors are absent, while outperforming them when latent factors are present due to their misspecification of models. Our analysis also quantifies the estimation gains from informative sources, establishes the corresponding statistical limits, and identifies conditions under which the proposed estimator attains the minimax rates. An interesting finding of our minimax analysis is that the estimation error may not vanish even when the target and informative sources share the same regression coefficients and the sample sizes become arbitrarily large. This is due to the nonidentifiability inherent in the factor-augmented model, where different parameter configurations can induce the same distribution of the observed data. In particular, when $b_{n,\min}\asymp1$, our minimax lower bounds (Theorem~\ref{thm-minimax rate}) contain additional terms of order $s/p^2$ and $s/p$ under squared $\ell_2$ and $\ell_1$ losses, respectively, where $s$ is the sparsity level, and $b_{n,\min}$ is the minimum of the domain-specific upper bounds on the factor-effect magnitude. Based on the minimax lower bounds established in Theorem~\ref{thm-minimax rate}, Figures~\ref{fig:minimax-illustration} and \ref{fig:minimax-illustration2} illustrate how the dimension-dependent and sample-size-dependent terms interact under different levels of heterogeneity. As the dimension increases, the lower bounds may exhibit nonmonotone behavior due to these competing effects. In contrast, increasing the sample sizes at a fixed dimension reduces the sampling error but leaves a nonvanishing factor-dependent term. These findings highlight the fundamental limits of estimation under the factor-augmented model and motivate us to investigate when transfer learning can achieve the minimax rates. We also address the practical problems of identifying informative sources and conducting simultaneous inference. Our contributions are summarized as follows:\vspace{-0.15cm}
\begin{itemize}
\item A Transfer Learning Framework for Data with Highly-Correlated Covariates and Latent Factor Structures. We develop a transfer learning framework that explicitly accounts for the data with latent factors and highly correlated covariates, which is particularly relevant in applications from economics, finance, and biomedicine. The framework systematically leverages auxiliary datasets to enhance the estimation of regression coefficients, providing a solid foundation for both accurate estimation and valid inference in complex data environments. We propose a source detection algorithm with proven consistency to avoid negative transfer. 
\item Theoretical Guarantees, Minimax Optimality, and Fundamental Limits. We establish the non-asymptotic $\ell_1$- and $\ell_2$-error bounds for our proposed Trans-FARM estimator, and derive minimax lower bounds, showing that the oracle estimator attains the optimal rates under suitable regularity and source-target similarity conditions. The lower bounds further reveal dimension-dependent terms (see Theorem \ref{thm-minimax rate}), arising from observational non-identifiability within the model class. These limits persist even when the informative domains share the same regression coefficients and the total sample size is arbitrarily large.
\item Valid Inference for Heterogeneous Effects. Moving beyond point estimation, we construct a comprehensive inferential tool for the high-dimensional transfer learning setting by constructing simultaneous confidence intervals for regression coefficients. A key innovation is a studentized approach that yields coefficient-specific interval lengths, thereby accurately capturing the heterogeneous influence of different features on the response.
\end{itemize}


\textbf{Notation:} For a vector $\mathbf{v}=(v_1,\cdots,v_p)^\top\in\mathbb{R}^{p}$, denote $\|\mathbf{v}\|=(\sum_{i=1}^{p}|v_i|^2)^{1/2}$, $\|\mathbf{v}\|_1=\sum_{i=1}^{p}|v_i|$, $\|\mathbf{v}\|_\infty=\max_{1\leq i\leq p}|v_i|$. For $\mathcal{S}\subseteq[p]$, let $\mathcal{S}^c$ be the complement of $\mathcal{S}$, i.e., $\mathcal{S}^c=[p]\backslash\mathcal{S}$, where we use the standard notation $[p]$ to denote $\{1,\ldots,p\}$ here. For a matrix $\mathbf{A}=(a_{ij})_{1\leq i\leq k_1,1\leq j\leq k_2}\in \mathbb{R}^{k_1\times k_2}$, denote $\|\mathbf{A}\|_{\max}=\max_{1\leq i\leq k_1,1\leq j\leq k_2}|a_{ij}|$, $\|\mathbf{A}\|_{\infty}=\max_{1\leq i\leq k_1}\sum_{j=1}^{k_2}|a_{ij}|$, $\|\mathbf{A}\|_{1}=\max_{1\leq j\leq k_2}\sum_{i=1}^{k_1}|a_{ij}|$, and let $\|\mathbf{A}\|_2$ be the spectral norm of $\mathbf{A}$. We use the standard notation $\boldsymbol{A}_{i,j}$ to represent the $(i,j)$-th element of $\boldsymbol{A}$, and use $\boldsymbol{A}_{i,-j}$ to denote the $i$-th row without the $(i,j)$-th element of $\boldsymbol{A}$ for a given matrix $\boldsymbol{A}$. We define the sub-Gaussian norm of a random variable $X$ as $\|X\|_{\psi_2}=\inf\{t>0:\mathbb{E}\exp(X^2/t^2)\leq 2\}$, and the sub-exponential norm of a sub-exponential random variable $Y$ as $\|Y\|_{\psi_1}=\inf\{t>0:\mathbb{E}\exp(|Y|/t)\leq 2\}$. The sub-Gaussian and the sub-exponential norms of a sub-exponential random vector $\boldsymbol{x}\in\mathbb{R}^{p}$ are defined as $\|\boldsymbol{x}\|_{\psi_2}=\sup_{\boldsymbol{v}\in\mathbb{R}^{p},\|\boldsymbol{v}\|_2=1}\|\boldsymbol{v}^\top\boldsymbol{x}\|_{\psi_2}$, and $\|\boldsymbol{x}\|_{\psi_1}=\sup_{\boldsymbol{v}\in\mathbb{R}^{p},\|\boldsymbol{v}\|_2=1}\|\boldsymbol{v}^\top\boldsymbol{x}\|_{\psi_1}$, respectively. For any two sequences $\{a_n\}$ and $\{b_n\}$, we use $a_n=O(b_n),a_n=o(b_n),$ and $a_n\asymp b_n$ to represent that there exist $c,C>0$ such that $|a_n|\leq C|b_n|$ for all $n$, $\lim_{n\to\infty}|a_n|/|b_n|=0$, and $0<c<|a_n/b_n|<C<\infty$, respectively. $a_n\lesssim b_n$ means $a_n=O(b_n)$. For a sequence of random variables $\{Z_n\}$, we use the notation $Z_n=O_P(1)$ if $\lim_{M\to\infty}\limsup_{n\to\infty}\mathbb{P}(|Z_n|>M)=0$, and $Z_n=o_P(1)$ if $\lim_{n\to\infty}\mathbb{P}(|Z_n|>\varepsilon)=0$ for any $\varepsilon>0$. For any two sequences of random variables $\{X_n\}$ and $\{Y_n\}$, we say $X_n=O_P(Y_n)$ if $X_n/Y_n=O_P(1)$, and $X_n=o_P(Y_n)$ if $X_n/Y_n=o_P(1)$.
\vspace{-0.2cm}

\section{Preliminaries and Model Formulation}\label{sec2}
Suppose that we have a target dataset $\{\boldsymbol{x}_i^{(0)},y_i^{(0)}\}_{i=1}^{n_0}$ and $K$ independent source datasets $\left\{\{\boldsymbol{x}_i^{(k)},y_i^{(k)}\}_{i=1}^{n_k}\right\}_{k=1}^{K}$. For $k\in\{0,1,\cdots,K\}$, we consider the following model (in matrix form):\vspace{-0.05cm}
\begin{equation}
\begin{aligned}\label{eq2.1}
    \boldsymbol{Y}^{(k)}&=\boldsymbol{F}_k\boldsymbol{\gamma}^{(k)}+\boldsymbol{U}_k\boldsymbol{w}^{(k)}+\boldsymbol{\mathcal{E}}^{(k)},\quad\boldsymbol{X}^{(k)}=\boldsymbol{F}_k{\boldsymbol{B}_k}^\top+\boldsymbol{U}_k,\quad k=0,\ldots,K,
\end{aligned}
\end{equation}
where $\boldsymbol{X}^{(k)}=({\boldsymbol{x}_1^{(k)\top}}\!,\ldots,{\boldsymbol{x}_{n_k}^{(k)\top}}\!)^\top \in \mathbb{R}^{n_k\times p}$ is the design matrix whose $i$-th row $\boldsymbol{x}_i^{(k)}\in\mathbb{R}^{p}$ denotes the observed covariate vector of the $i$-th sample in the $k$-th dataset, $\boldsymbol{F}_k=\!({\boldsymbol{f}_1^{(k)\top}}\!,\ldots,{\boldsymbol{f}_{n_k}^{(k)\top}}\!)^\top\in\mathbb{R}^{n_k\times r_k}$ denotes the unobserved factor matrix collecting the realizations of the latent factors in the $k$-th dataset, $\boldsymbol{U}_k=({\boldsymbol{u}_1^{(k)\top}},\ldots,{\boldsymbol{u}_{n_k}^{(k)\top}})^\top \in \mathbb{R}^{n_k\times p}$ denotes the idiosyncratic component, $\boldsymbol{Y}^{(k)}=(y_1^{(k)},\ldots,y_{n_k}^{(k)})^\top\in \mathbb{R}^{n_k\times 1}$ denotes the response vector, $\boldsymbol{\mathcal{E}}^{(k)}=(\mathcal{E}_1^{(k)},\ldots,\mathcal{E}_{n_k}^{(k)})^\top\in \mathbb{R}^{n_k\times 1}$ denotes the error vector, and $\boldsymbol{B}_k\in \mathbb{R}^{p\times r_k}$ denotes the factor loading matrix. The parameter $\boldsymbol{w}^{(k)}\in\mathbb{R}^{p}$ is the true regression parameter vector, and $\boldsymbol{\gamma}^{(k)}\in\mathbb{R}^{r_k}$ is the contribution of the latent factor $\boldsymbol{f}^{(k)}$ in the $k$-th dataset. Here, we regard $\{(\boldsymbol{x}_i^{(k)},y_i^{(k)},\boldsymbol{f}_i^{(k)},\boldsymbol{u}_i^{(k)},\mathcal{E}_i^{(k)})\}_{i=1}^{n_k}$ as independent and identically distributed (i.i.d.) realizations of $(\boldsymbol{x}^{(k)},y^{(k)},\boldsymbol{f}^{(k)},\boldsymbol{u}^{(k)},\mathcal{E}^{(k)})$ for $k=0,1,\ldots,K$. Moreover, we allow the number of factors in different design matrices $\boldsymbol{X}^{(k)}$ to be different, i.e., $r_k$ can be unequal to each other. Our goal is to estimate the target regression coefficient vector $\boldsymbol{w}^{(0)}$. The matrix form in \eqref{eq2.1} can also be written as
\begin{equation}\label{eq2.2}
\begin{aligned}
    \boldsymbol{Y}^{(k)}&=\boldsymbol{F}_k\boldsymbol{\varphi}^{(k)}+\boldsymbol{X}^{(k)}\boldsymbol{w}^{(k)}+\boldsymbol{\mathcal{E}}^{(k)},\quad \boldsymbol{X}^{(k)}=\boldsymbol{F}_k\boldsymbol{B}_k^\top+\boldsymbol{U}_k,\quad k=0,\ldots,K,
\end{aligned}
\end{equation}
where $\boldsymbol{\varphi}^{(k)}=\boldsymbol{\gamma}^{(k)}-\boldsymbol{B}_k^\top\boldsymbol{w}^{(k)}$ quantifies the additional effect of factors $\boldsymbol{F}_k$ to $\boldsymbol{Y}^{(k)}$. Since the latent factor $\boldsymbol{f}_i^{(k)}$, the idiosyncratic component $\boldsymbol{u}_i^{(k)}$, and the loading matrix $\boldsymbol{B}_k$ are unobserved and must be estimated, we first provide their estimators. Note that $\boldsymbol{f}_i^{(k)}$ is not identifiable because $\boldsymbol{x}_{i}^{(k)}=\boldsymbol{B}_{k}\boldsymbol{f}_i^{(k)}+\boldsymbol{u}_{i}^{(k)}=\boldsymbol{B}_{k}\boldsymbol{T}(\boldsymbol{T}^{-1}\boldsymbol{f}_i^{(k)})+\boldsymbol{u}_{i}^{(k)}$ for any non-singular matrix $\boldsymbol{T}\in\mathbb{R}^{r_k\times r_k}$. To ensure identifiability, we impose the following normalization:$$\mbox{cov}(\boldsymbol{f}^{(k)})=\boldsymbol{I}_{r_k}\text{ and }\boldsymbol{B}_{k}^\top\boldsymbol{B}_{k}\text{ is diagonal.}$$ We also assume that $\mathcal{E}^{(k)}$, $\boldsymbol{f}^{(k)}$ and $\boldsymbol{u}^{(k)}$ are mutually independent and zero-mean, and $\max_kr_k<\infty$. Then the estimators follows by solving the constrained least squares problem for $k=0,\ldots,K$:
\begin{equation}\label{eq2.3}
\begin{aligned}
    &(\hat{\boldsymbol{F}}_k,\hat{\boldsymbol{B}}_{k})=\mathop{\arg\min}_{\boldsymbol{F}\in\mathbb{R}^{n_k\times r_k},\boldsymbol{B}\in\mathbb{R}^{p\times r_k}}\|\boldsymbol{X}^{(k)}-\boldsymbol{F}\boldsymbol{B}^\top\|_{\mathbb{F}}^2\\
    &\text{subject to }\frac{1}{n_k}\boldsymbol{F}^\top\boldsymbol{F}=\boldsymbol{I}_{r_k}\text{, and }\boldsymbol{B}^\top\boldsymbol{B}\text{ is diagonal.}
\end{aligned}
\end{equation}

Simple calculation yields that $\hat{\boldsymbol{F}}_k/n_k^{1/2}$ is the eigenvectors corresponding to the top $r_k$ eigenvalues of $\boldsymbol{X}^{(k)}{\boldsymbol{X}^{(k)}}^\top$, $\hat{\boldsymbol{B}}_{k}^\top=(\hat{\boldsymbol{F}}_k^\top\hat{\boldsymbol{F}}_k)^{-1}\hat{\boldsymbol{F}}_k^\top\boldsymbol{X}^{(k)}=\hat{\boldsymbol{F}}_k^\top\boldsymbol{X}^{(k)}/n_k$, and the estimator for the remainder term is $\hat{\boldsymbol{U}}_k=\boldsymbol{X}^{(k)}-\hat{\boldsymbol{F}}_k\hat{\boldsymbol{B}}_{k}^\top=(\boldsymbol{I}_{n_k}-\hat{\boldsymbol{F}}_k\hat{\boldsymbol{F}}_k^\top/n_k)\boldsymbol{X}^{(k)}$. Then we have $\boldsymbol{X}^{(k)}=\hat{\boldsymbol{F}}_k\hat{\boldsymbol{B}}_{k}^\top+\hat{\boldsymbol{U}}_k$, and $\hat{\boldsymbol{F}}_k$, $\hat{\boldsymbol{B}}_{k}$, and $\hat{\boldsymbol{U}}_k$ are estimates for $\boldsymbol{F}_k$, $\boldsymbol{B}_k$, and $\boldsymbol{U}_k$ in model \eqref{eq2.2}.

To attain the consistency of these estimations, we first impose some regularity assumptions that are standard in the high-dimensional factor model \citep{bai2003inferential,fan2013large,fan2024latent}.

\begin{assumption}[Assumptions on $\boldsymbol{f}^{(k)}$, $\boldsymbol{u}^{(k)}$, and $\boldsymbol{B}_k$ (for the $k$-th dataset)]\label{assum1}
\leavevmode\par\vspace{0.3ex}%
    \begin{enumerate}[label=(\alph*)]
    \item Assume that $\boldsymbol{\nu}^{(k)}=(\boldsymbol{u}^{(k)\top},\boldsymbol{f}^{(k)\top})^\top$ is sub-Gaussian with $\|\boldsymbol{\nu}^{(k)}\|_{\psi_2}\leq v_0$.
    \item There exists a constant $C>1$ such that $p/C\leq\lambda_{\min}(\boldsymbol{B}_k^\top\boldsymbol{B}_k)\leq\lambda_{\max}(\boldsymbol{B}_k^\top\boldsymbol{B}_k)\leq pC$.
    \item There exists a constant $\Upsilon>0$ such that $\|\boldsymbol{B}_k\|_{\max}\leq\Upsilon$ and $\mathbb{E}|\boldsymbol{u}^{(k)\top}\boldsymbol{u}^{(k)}-\mbox{tr}(\boldsymbol{\Sigma}_{\boldsymbol{u}}^{(k)})|^4\leq\Upsilon p^2$.
    \item There exists a positive constant $\iota<1$ such that $\iota\leq\lambda_{\min}(\boldsymbol{\Sigma}_{\boldsymbol{u}}^{(k)})$, $\|\boldsymbol{\Sigma}_{\boldsymbol{u}}^{(k)}\|_1\leq1/\iota$ and $\min_{1\leq i,\ell\leq p}\mbox{var}(u_{i}^{(k)}u_{\ell}^{(k)})\geq\iota$.
\end{enumerate}
\end{assumption}

The consistency of factor estimations is summarized in the following lemma.

\begin{lemma}\label{lemma1}
    Define $\boldsymbol{H}_k=n_k^{-1}\boldsymbol{V}_k^{-1}\hat{\boldsymbol{F}}_k^\top\boldsymbol{F}_k\boldsymbol{B}_k^\top\boldsymbol{B}_k$, where $\boldsymbol{V}_k\in\mathbb{R}^{r_k\times r_k}$ is a diagonal matrix with its diagonal elements being the first $r_k$ largest eigenvalues of $n_k^{-1}\boldsymbol{X}^{(k)}{\boldsymbol{X}^{(k)\top}}$. If Assumption \ref{assum1} holds for any $k\in\mathcal{L}\subseteq\{0,\ldots,K\}$, and $K\leq p^{C_K}$, let $\zeta_{t,\mathcal{L}}=t+\log(2p|\mathcal{L}|)$ for $t>0$ and $n_{\min}=\min_{k\in\mathcal{L}}n_k$, the following bounds hold simultaneously over $k\in\mathcal{L}$ with probability at least $1-2|\mathcal{L}|e^{-cn_{\min}}-|\mathcal{L}|e^{-cp}-3e^{-t}$,\vspace{-0.05cm}
    \begin{itemize}
        \item $\|\hat{\boldsymbol{F}}_k-\boldsymbol{F}_k\boldsymbol{H}_k^\top\|_{\mathbb{F}}^{2}\leq C(n_k/p+1/n_k)$;
        \item For any nonempty $\mathcal{J}\subseteq[p]$, $\max_{j\in\mathcal{J}}\sum_{i=1}^{n_k}|\hat{u}_{ij}^{(k)}-u_{ij}^{(k)}|^2\leq C(t+\log(2|\mathcal{L}||\mathcal{J}|)+n_k/p)$;
        \item $\|\boldsymbol{H}_k^\top\boldsymbol{H}_k-\boldsymbol{I}_{r_k}\|_{\mathrm{F}}^2\leq C\big((\log(e|\mathcal{L}|)+t)/n_k+1/p\big)$;
        \item $\max_{j\in[p]}\left\{\|[\hat{\boldsymbol{B}}_k-\boldsymbol{B}_k\boldsymbol{H}_k^{-1}]_{j\cdot}\|_2^2\vee\|[\hat{\boldsymbol{B}}_k-\boldsymbol{B}_k\boldsymbol{H}_k^\top]_{j\cdot}\|_2^2\right\}\leq C\left(1/p+\zeta_{t,\mathcal{L}}/n_k\right)$
    \end{itemize}
    for some sufficiently large constant $C>0$ depending on the constants in Assumption \ref{assum1}. Moreover, on the same event, $cp\leq\lambda_{\min}(\boldsymbol{V}_k)\leq\lambda_{\max}(\boldsymbol{V}_k)\leq Cp,~\|\boldsymbol{H}_k\|_2+\|\boldsymbol{H}_k^{-1}\|_2\leq C,~\|\boldsymbol{F}_k\|_2\leq Cn_k^{1/2},~\|\boldsymbol{U}_k\|_2\leq C(p^{1/2}+n_k^{1/2}),~\|\boldsymbol{U}_k\boldsymbol{B}_k\|_2\leq C(n_kp)^{1/2},~\max_{j\in[p]}\|\boldsymbol{U}_{k,\cdot j}\|_2\leq C(n_k^{1/2}+\zeta_{t,\mathcal{L}}^{1/2})$.
\end{lemma}

\begin{remark}\label{remark 2.1}
    The true number of latent factors $r_k$ in each dataset is unknown in practice, and determining $r_k$ in a data-driven manner is a crucial task. A rich literature has proposed various approaches for estimating $r_k$ \citep{bai2002determining,lam2012factor,ahn2013eigenvalue,fan2022estimating}. Our theoretical results remain valid as long as $r_k$ is replaced by any consistent estimator $\hat{r}_k$, i.e., we only require $\mathbb{P}(\hat{r}_k=r_k,~\forall k) \to 1, ~\text{as } n_k\to\infty.$
\end{remark}
Therefore, we assume throughout this paper that the number of factors $r_k$ is known. 
For empirical implementation, we adopt the eigenvalue ratio method \citep{lam2012factor,ahn2013eigenvalue} to determine $r_k$. Specifically, let $\lambda_{i}(\boldsymbol{X}^{(k)}{\boldsymbol{X}^{(k)\top}})$ denote the eigenvalues of the Gram matrix $\boldsymbol{X}^{(k)}{\boldsymbol{X}^{(k)\top}}$. Then the estimator of the number of factors is given by $\hat{r}_k=\mathop{\arg\max}_{i\leq\mathcal{K}}\left\{\lambda_i(\boldsymbol{X}^{(k)}{\boldsymbol{X}^{(k)}}^\top)/\lambda_{i+1}(\boldsymbol{X}^{(k)}{\boldsymbol{X}^{(k)}}^\top)\right\},$
where $1\leq\mathcal{K}\leq \min\{n_k-1,p\}$ is a prescribed upper bound for $r_k$.

We now consider estimating the regression coefficient vector $\boldsymbol{w}^{(k)}$ in each dataset. Under the high-dimensional setting, we assume that $\boldsymbol{w}^{(k)}$ is sparse. After attaining the estimators for $\hat{\boldsymbol{F}}_k$, $\hat{\boldsymbol{B}}_{k}$, and $\hat{\boldsymbol{U}}_k$ via \eqref{eq2.3}, let $\lambda>0$ be a tuning parameter, the regularized estimators of $\boldsymbol{w}^{(k)}$ and $\boldsymbol{\gamma}^{(k)}$ for $k=0,\ldots,K$ in FARM model are given by 
\begin{equation}\label{eq2.4}
    \begin{aligned}
        (\hat{\boldsymbol{w}}^{(k)},\hat{\boldsymbol{\gamma}}^{(k)})=\mathop{\arg\min}\limits_{\boldsymbol{w}\in\mathbb{R}^p,\boldsymbol{\gamma}\in\mathbb{R}^{r_k}}\left\{\frac{1}{2n_k}\|\boldsymbol{Y}^{(k)}-\hat{\boldsymbol{U}}_k\boldsymbol{w}-\hat{\boldsymbol{F}}_k\boldsymbol{\gamma}\|_2^2+\lambda\|\boldsymbol{w}\|_1\right\}.
    \end{aligned}
\end{equation}
\vspace{-0.2cm}

\section{Factor Augmented Transfer Learning}\label{sec3}
\subsection{Oracle-Trans-FARM Estimation and Minimax Rates}\label{sec3.1}
Recall the estimators in \eqref{eq2.4}. Let $\hat{\boldsymbol{P}}_k=n_k^{-1}\hat{\boldsymbol{F}}_k\hat{\boldsymbol{F}}_k^\top$ be the projection matrix. Denote $\tilde{\boldsymbol{Y}}^{(k)}=(\boldsymbol{I}_{n_k}-\hat{\boldsymbol{P}}_k)\boldsymbol{Y}^{(k)}$. Recall that $\hat{\boldsymbol{U}}_k=(\boldsymbol{I}_{n_k}-\hat{\boldsymbol{F}}_k\hat{\boldsymbol{F}}_k^\top/n_k)\boldsymbol{X}^{(k)}=(\boldsymbol{I}_{n_k}-\hat{\boldsymbol{P}}_k)\boldsymbol{X}^{(k)}$, which implies that $\hat{\boldsymbol{F}}_k^\top\hat{\boldsymbol{U}}_k=\boldsymbol{0}_{r_k\times p}$ and \eqref{eq2.4} can be equivalently written as 
\begin{equation}\label{eq3.1}
    \begin{aligned}
        \hat{\boldsymbol{w}}^{(k)}&=\mathop{\arg\min}\limits_{\boldsymbol{w}\in\mathbb{R}^p}\left\{\frac{1}{2n_k}\|(\boldsymbol{I}_{n_k}-\hat{\boldsymbol{P}}_k)(\boldsymbol{Y}^{(k)}-\boldsymbol{X}^{(k)}\boldsymbol{w})\|_2^2+\lambda\|\boldsymbol{w}\|_1\right\}\\
        \hat{\boldsymbol{\gamma}}^{(k)}&=(\hat{\boldsymbol{F}}_k^\top\hat{\boldsymbol{F}}_k)^{-1}\hat{\boldsymbol{F}}_k^\top\boldsymbol{Y}^{(k)}=\frac{1}{n_k}\hat{\boldsymbol{F}}_k^\top\boldsymbol{Y}^{(k)}.
    \end{aligned}
\end{equation}

Inspired by the insights of transfer learning, we denote $\boldsymbol{\beta}=\boldsymbol{w}^{(0)}$ as the target regression coefficient vector and suppose that $\boldsymbol{\beta}$ is $s$-sparse, i.e., $\|\boldsymbol{\beta}\|_0\leq s$ with $s\ll p$. Define the $k$-th contrast vector $\boldsymbol{\delta}^{(k)}=\boldsymbol{\beta}-\boldsymbol{w}^{(k)}$ to quantify the discrepancy between the target and the $k$-th source datasets. We use the $\ell_1$-distance $\|\boldsymbol{\delta}^{(k)}\|_1$ to measure the similarity of the target dataset and the $k$-th source dataset. We begin with considering a general transfer learning algorithm on FARM when the transferable set $\mathcal{A}_{\eta}$ is known, where $\mathcal{A}_\eta=\{1\leq k\leq K:\|\boldsymbol{\delta}^{(k)}\|_1\leq \eta\}$ for some $\eta$. Transfer learning under this setting is expected to enhance estimation accuracy of $\boldsymbol{\beta}$ when $\eta$ is sufficiently small, as the target and source parameters are closely aligned. According to \eqref{eq3.1}, we are motivated to attain the initial estimation of $\boldsymbol{\beta}$ through the following regularized loss function, denoted by $\hat{\boldsymbol{w}}^{\mathcal{A}_\eta}$:
\begin{equation}\label{eq3.2}
    \begin{aligned}
        \hat{\boldsymbol{w}}^{\mathcal{A}_\eta}=\mathop{\arg\min}\limits_{\boldsymbol{w}\in\mathbb{R}^p}\left\{\frac{1}{2(n_{\mathcal{A}_\eta}+n_0)}\sum_{k\in\{0\}\cup\mathcal{A}_\eta}\|\tilde{\boldsymbol{Y}}^{(k)}-\hat{\boldsymbol{U}}_k\boldsymbol{w}\|_2^2+\lambda_{\boldsymbol{w}}\|\boldsymbol{w}\|_1\right\},
    \end{aligned}
\end{equation}
where $n_{\mathcal{A}_\eta}=\sum_{k\in\mathcal{A}_\eta}n_k$, and $\lambda_{\boldsymbol{w}}$ is a tuning parameter. 
The equation above can be understood as finding a solution that converges to its population version $\boldsymbol{w}^{\mathcal{A}_\eta}$ under certain regularization conditions, where $\boldsymbol{w}^{\mathcal{A}_\eta}$ satisfies
\begin{equation}\label{eq3.3}
    \begin{aligned}
        \boldsymbol{w}^{\mathcal{A}_\eta}=\mathop{\arg\min}\limits_{\boldsymbol{w}\in\mathbb{R}^p}\mathbb{E}\left[\frac{1}{2(n_{\mathcal{A}_\eta}+n_0)}\sum_{k\in\{0\}\cup\mathcal{A}_\eta}\|\boldsymbol{Y}^{(k)}-\boldsymbol{U}_k\boldsymbol{w}-\boldsymbol{F}_k\boldsymbol{\gamma}^{(k)}\|_2^2\right].
    \end{aligned}
\end{equation}

According to $\mathbb{E}[\boldsymbol{\mathcal{E}}^{(k)}]=\boldsymbol{0}_{n_k}$, the independence between $\boldsymbol{\mathcal{E}}^{(k)}$ and $\boldsymbol{U}_k$, and the independence between $\boldsymbol{U}_k$ and $\boldsymbol{F}_k$, the first-order optimality from \eqref{eq3.3} ensures that $\boldsymbol{w}^{\mathcal{A}_\eta}$ satisfies
\begin{align*}
    \mathbb{E}\left[\sum_{k\in\{0\}\cup\mathcal{A}_\eta}\alpha_k\boldsymbol{U}_k^\top\boldsymbol{U}_k(\boldsymbol{w}^{(k)}-\boldsymbol{w}^{\mathcal{A}_\eta})/n_k\right]=\mathbb{E}\left[\sum_{k\in\{0\}\cup\mathcal{A}_\eta}\alpha_k\boldsymbol{u}^{(k)}\boldsymbol{u}^{(k)\top}(\boldsymbol{w}^{(k)}-\boldsymbol{w}^{\mathcal{A}_\eta})\right]=\boldsymbol{0}_p,
\end{align*}
where $\alpha_k=n_k/(n_{\mathcal{A}_\eta}+n_0)$. Hence, 
$\boldsymbol{w}^{\mathcal{A}_\eta}$ has an explicit form: $\boldsymbol{w}^{\mathcal{A}_\eta}=(\boldsymbol{\Sigma}_{\boldsymbol{u}}^{\mathcal{A}_\eta})^{-1}\sum_{k\in\{0\}\cup\mathcal{A}_\eta}\alpha_k\\\boldsymbol{\Sigma}_{\boldsymbol{u}}^{(k)}\boldsymbol{w}^{(k)}$, where $\boldsymbol{\Sigma}_{\boldsymbol{u}}^{\mathcal{A}_\eta}=\sum_{k\in\{0\}\cup\mathcal{A}_\eta}\alpha_k\boldsymbol{\Sigma}_{\boldsymbol{u}}^{(k)}$, and $\boldsymbol{\Sigma}_{\boldsymbol{u}}^{(k)}=\mathbb{E}[\boldsymbol{u}^{(k)}\boldsymbol{u}^{(k)\top}]$.

After attaining $\hat{\boldsymbol{w}}^{\mathcal{A}_\eta}$, we define $\boldsymbol{\delta}^{\mathcal{A}_\eta}=\boldsymbol{\beta}-\boldsymbol{w}^{\mathcal{A}_\eta}$, and obtain its estimator $\hat{\boldsymbol{\delta}}^{\mathcal{A}_\eta}$ by solving:
\begin{equation}\label{eq3.4}
    \begin{aligned}
        \hat{\boldsymbol{\delta}}^{\mathcal{A}_\eta}=\mathop{\arg\min}\limits_{\boldsymbol{\delta}\in\mathbb{R}^p}\left\{\frac{1}{2n_0}\|\tilde{\boldsymbol{Y}}^{(0)}-\hat{\boldsymbol{U}}_0(\hat{\boldsymbol{w}}^{\mathcal{A}_\eta}+\boldsymbol{\delta})\|_2^2+\lambda_{\boldsymbol{\delta}}\|\boldsymbol{\delta}\|_1\right\},
    \end{aligned}
\end{equation}
where $\lambda_{\boldsymbol{\delta}}$ is a tuning parameter. We regard $\hat{\boldsymbol{\beta}}^{\mathcal{A}_\eta}=\hat{\boldsymbol{w}}^{\mathcal{A}_\eta}+\hat{\boldsymbol{\delta}}^{\mathcal{A}_\eta}$ as the estimator of $\boldsymbol{\beta}$. The algorithm is summarized in Algorithm \ref{algo1}.

\begin{algorithm}[!ht]
\caption{Oracle-Trans-FARM Algorithm}
\label{algo1}
\KwIn{Target data $(\boldsymbol{X}^{(0)}, \boldsymbol{Y}^{(0)})$, source data $\{(\boldsymbol{X}^{(k)}, \boldsymbol{Y}^{(k)})\}_{k=1}^K$, penalty parameters $\lambda_{\boldsymbol{w}}$ and $\lambda_{\boldsymbol{\delta}}$, transferring set $\mathcal{A}_\eta$, factor bound $\mathcal{K}$.}
\KwOut{Estimated coefficient vector $\hat{\boldsymbol{\beta}}^{\mathcal{A}_\eta}$.}
\underline{\textbf{Preprocessing step}}: Estimate $\hat{r}_k = \arg\max_{i \le \mathcal{K}} \dfrac{\lambda_i(\boldsymbol{X}^{(k)}{\boldsymbol{X}^{(k)}}^\top)}{\lambda_{i+1}(\boldsymbol{X}^{(k)}{\boldsymbol{X}^{(k)}}^\top)}$. Let $\hat{\boldsymbol{F}}_k / n_k^{1/2}$ be the eigenvectors corresponding to the top $\hat{r}_k$ eigenvalues of $\boldsymbol{X}^{(k)}{\boldsymbol{X}^{(k)}}^\top$. Compute $\hat{\boldsymbol{B}}_k^\top = \hat{\boldsymbol{F}}_k^\top \boldsymbol{X}^{(k)} / n_k$ and $\hat{\boldsymbol{U}}_k = (\boldsymbol{I}_{n_k} - \hat{\boldsymbol{F}}_k \hat{\boldsymbol{F}}_k^\top / n_k)\boldsymbol{X}^{(k)}$.\\
\underline{\textbf{Transferring step}}: Compute \eqref{eq3.2}, i.e.,
$$\displaystyle\hat{\boldsymbol{w}}^{\mathcal{A}_\eta}=\mathop{\arg\min}\limits_{\boldsymbol{w}\in\mathbb{R}^p}\left\{\frac{1}{2(n_{\mathcal{A}_\eta}+n_0)}\sum_{k\in\{0\}\cup\mathcal{A}_\eta}\|\tilde{\boldsymbol{Y}}^{(k)}-\hat{\boldsymbol{U}}_k\boldsymbol{w}\|_2^2+\lambda_{\boldsymbol{w}}\|\boldsymbol{w}\|_1\right\}$$ \\
\underline{\textbf{Debiasing step}}: Compute \eqref{eq3.4}, i.e.,
$$\hat{\boldsymbol{\delta}}^{\mathcal{A}_\eta}=\mathop{\arg\min}\limits_{\boldsymbol{\delta}\in\mathbb{R}^p}\left\{\frac{1}{2n_0}\|\tilde{\boldsymbol{Y}}^{(0)}-\hat{\boldsymbol{U}}_0(\hat{\boldsymbol{w}}^{\mathcal{A}_\eta}+\boldsymbol{\delta})\|_2^2+\lambda_{\boldsymbol{\delta}}\|\boldsymbol{\delta}\|_1\right\}$$\\
Let $\hat{\boldsymbol{\beta}}^{\mathcal{A}_\eta} = \hat{\boldsymbol{w}}^{\mathcal{A}_\eta} + \hat{\boldsymbol{\delta}}^{\mathcal{A}_\eta}$.\\
\end{algorithm}

Before presenting the error bound of $\hat{\boldsymbol{\beta}}^{\mathcal{A}_\eta}$, several technical assumptions below are required.

\begin{assumption}[Assumption on $\mathcal{E}^{(k)}$]\label{assum2}
There exists a positive constant $c<\infty$ such that $\max_{k\in\{0\}\cup\mathcal{A}_\eta}\|\mathcal{E}^{(k)}\|_{\psi_2}\leq c$.
\end{assumption}

\begin{assumption}\label{assum3}
There exists a positive constant $C_1\!\!<\!\!\infty$ such that $\max_{k\in\{0\}\cup\mathcal{A}_\eta}\!\|(\boldsymbol{\Sigma}_{\boldsymbol{u}}^{\mathcal{A}_\eta})^{-1}\!\boldsymbol{\Sigma}_{\boldsymbol{u}}^{(k)}\!\|_1\\\leq\!C_1$, where $\boldsymbol{\Sigma}_{\boldsymbol{u}}^{\mathcal{A}_\eta}=\sum_{k\in\{0\}\cup\mathcal{A}_\eta}\alpha_k\boldsymbol{\Sigma}_{\boldsymbol{u}}^{(k)}$, and $\boldsymbol{\Sigma}_{\boldsymbol{u}}^{(k)}=\mathbb{E}[\boldsymbol{u}^{(k)}\boldsymbol{u}^{(k)\top}]$.
\end{assumption}

Assumption \ref{assum2} is standard and mild in high-dimensional regression models \citep{fan2024latent,zhang2017simultaneous,tian2023transfer}. Assumption \ref{assum3} is also commonly adopted in transfer learning settings \citep{tian2023transfer,zhang2025transfer}.
In particular, Assumption \ref{assum3} is reasonable as it does not require exact sparsity of $(\boldsymbol{\Sigma}_{\boldsymbol{u}}^{\mathcal{A}_\eta})^{-1}\boldsymbol{\Sigma}_{\boldsymbol{u}}^{(k)}$ and is satisfied, for example, under the uniformly strictly diagonally dominant matrix $\boldsymbol{\Sigma}_{\boldsymbol{u}}^{\mathcal{A}_\eta}$ and Assumption \ref{assum1}. This assumption can be further interpreted as a constraint on the heterogeneity between the target idiosyncratic component and the source idiosyncratic component; see Condition 4 in \cite{li2023transfer} for a detailed discussion. Since the observed covariates may not satisfy the standard sub-Gaussian design condition with a dimension-independent parameter when strong factors are present, our assumptions can be regarded as weaker conditions. This assumption can also be relaxed at the cost of slightly weaker estimation bounds, following arguments similar to those in Theorem 3 of \cite{tian2023transfer}.

Now, we are ready to present our main results for the estimation error bounds of the Oracle Trans-FARM algorithm proposed in Algorithm \ref{algo1}. Given a fixed transferable set $\mathcal{A}_\eta\subseteq\{1,\ldots,K\}$, and a deterministic collection of nonnegative numbers $\boldsymbol{b}_n=(b_{n,k})_{k\in\{0\}\cup\mathcal{A}_\eta}$ with $b_{n,\min}=\min_{k\in\{0\}\cup\mathcal{A}_\eta}b_{n,k}>0$, consider the parameter space:\vspace{-0.05cm}
\begin{align*}
        \Theta_{\boldsymbol{b}_n}(s,\eta)=\Big\{\boldsymbol{\mathcal{B}}=\big(\boldsymbol{\beta},&\,\{\boldsymbol{\delta}^{(k)}\}_{k\in[K]},\{\boldsymbol{\varphi}^{(k)}\}_{k\in\{0\}\cup[K]},\vartheta\big):\|\boldsymbol{\beta}\|_0\leq s,\sup_{k\in\mathcal{A}_\eta}\|\boldsymbol{\delta}^{(k)}\|_1\leq\eta,\\
        &\|\boldsymbol{\varphi}^{(k)}\|_2\leq b_{n,k},~\text{model \eqref{eq2.1} and Assumptions 1--3 hold}\Big\},
\end{align*}
where $\vartheta$ collects the nuisance parameters, including the loading matrices $\boldsymbol{B}_k$, and distributions of the latent factors $\boldsymbol{f}^{(k)}$, idiosyncratic components $\boldsymbol{u}^{(k)}$, and errors $\boldsymbol{\mathcal{E}}^{(k)}$. All constants in these assumptions are fixed uniformly over $\Theta_{\boldsymbol{b}_n}(s,\eta)$. For each $\boldsymbol{\mathcal{B}}\in\Theta_{\boldsymbol{b}_n}(s,\eta)$, let $\mathbb{P}_{\boldsymbol{\mathcal{B}}}$ denote the induced joint distribution of the observed samples from the domains indexed by $\{0\}\cup\mathcal{A}_\eta$.
\vspace{-0.05cm}

\begin{theorem}\label{thm1}
Suppose that Assumptions \ref{assum1}-\ref{assum2} hold for $k\in\{0\}\cup\mathcal{A}_\eta$ and Assumption \ref{assum3} holds. Assume that $|\mathcal{A}_\eta|+1\leq (p\wedge n_{\min})^{C_K}$ for some constant $C_K>0$ where $n_{\min}=\min_{k\in\{0\}\cup\mathcal{A}_\eta}n_k$, $n_0=O(p\log p)$, $(1\vee\eta)\left(n_0/(n_{\mathcal{A}_\eta}+n_0)\right)^{1/2}=O(1)$, and that condition \eqref{eq3.5} hold:\vspace{-0.05cm}
\begin{equation}
    \begin{aligned}
        &s\!\left(\!\frac{(|\mathcal{A}_\eta|\!+\!1)\!\log p}{n_{\mathcal{A}_\eta}+n_0}\!+\!\frac{1}{p}\!\right)\!\leq c,~\!\left(\!1\!+\!\frac{s\log p}{n_0}\!\right)\!\left(\!\frac{s(1\vee\eta)^2\!\log p}{n_{\mathcal{A}_\eta}+n_0}\!+\!\left(\!\eta^2\!\wedge\!\eta\left(\frac{\log p}{n_0}\right)^{1/2}\right)\!\right)\!\vee\!\eta\left(\frac{\log p}{n_0}\right)^{1/2}\!=O(1),\label{eq3.5}\\
        &\frac{\sum_{k\in\{0\}\cup\mathcal{A}_\eta}\!\!\big(\!\mathcal{V}_{n_k,p}b_{n,k}\!+\!\!(n_k\!\log p/p)^{1/2}\!+\!(\!1\!\vee\!\eta)\!\log p\!\big)}{(1\vee\eta)\left((n_{\mathcal{A}_\eta}+n_0)\log p\right)^{1/2}}\!\vee\!\frac{\mathcal{V}_{n_0,p}b_{n,0}}{(n_0\!\log p)^{1/2}\!}\!\vee\!\frac{\log p}{n_0}\!\vee\!\frac{\eta(n_{\mathcal{A}_\eta}+n_0)^{1/2}}{(1\!\vee\!\eta)p(\log p\!)^{1/2}}\!=\!O(1),
    \end{aligned}
\end{equation}
    for some sufficiently small constant $c>0$, where $\mathcal{V}_{n_k,p}=(n_k/p)+(\log p/n_k)^{1/2}+(n_k\log p/p)^{1/2}$. Take $\lambda_{\boldsymbol{w}}=C_{\boldsymbol{w}}(1\vee\eta)\left(\log p/(n_{\mathcal{A}_\eta}+n_0)\right)^{1/2}$ and $\lambda_{\boldsymbol{\delta}}=C_{\boldsymbol{\delta}}\left(\log p/n_0\right)^{1/2}$ for some sufficiently large constants $C_{\boldsymbol{w}}$ and $C_{\boldsymbol{\delta}}$, then there exists some constants $C,c'>0$ such that 
\begin{align*}
    \inf_{\boldsymbol{\mathcal{B}}\in\Theta_{\boldsymbol{b}_n}(s,\eta)}\mathbb{P}_{\boldsymbol{\mathcal{B}}}\Bigg\{\|\hat{\boldsymbol{\beta}}^{\mathcal{A}_\eta}-\boldsymbol{\beta}\|_2^2&\leq C\Bigg(\left(1+\frac{s\log p}{n_0}\right)\Bigg(\frac{s(1\vee\eta)^2\log p}{n_{\mathcal{A}_\eta}+n_0}+\eta^2\wedge\eta\left(\frac{\log p}{n_0}\right)^{1/2}\Bigg)\Bigg)\Bigg\}\\
    &~~~~~~~~~~~~~~~\geq 1-2(|\mathcal{A}_\eta|+1)e^{-cn_{\min}}-(|\mathcal{A}_\eta|+1)e^{-cp}-8p^{-100},\\
    \inf_{\boldsymbol{\mathcal{B}}\in\Theta_{\boldsymbol{b}_n}(s,\eta)}\mathbb{P}_{\boldsymbol{\mathcal{B}}}\Bigg\{\|\hat{\boldsymbol{\beta}}^{\mathcal{A}_\eta}-\boldsymbol{\beta}\|_1&\leq C\Bigg(s(1\vee\eta)\left(\frac{\log p}{n_{\mathcal{A}_\eta}+n_0}\right)^{1/2}+\eta\Bigg)\Bigg\}\\&~~~~~~~~~~~~~~~\geq 1-2(|\mathcal{A}_\eta|+1)e^{-c'n_{\min}}-(|\mathcal{A}_\eta|+1)e^{-c'p}-8p^{-100},
\end{align*}
\end{theorem}
\begin{remark}\label{remark3.1}
    We provide some illustrations of the conditions in Theorem \ref{thm1}. The dimension-type condition $n_0=O(p\log p)$ is widely applied in factor analysis to ensure consistency. To illustrate condition \eqref{eq3.5}, suppose that $n_k\asymp n_0$ uniformly over $k\in\{0\}\cup\mathcal{A}_\eta$, $\max_{k\in\{0\}\cup\mathcal{A}_\eta}b_{n,k}=O((\log p)^{1/2})$, and $\eta=O(1)$. A simple sufficient condition is $(s\vee(|\mathcal{A}_\eta|+1))\log p/n_0\vee(|\mathcal{A}_\eta|+1)\log p/p\leq c$ for some sufficiently small constant $c>0$.

    Alternatively, suppose that $|\mathcal{A}_\eta|=O(1)$, $\min_{k\in\cup\mathcal{A}_\eta}n_k\gtrsim n_0$, $\max_{k\in\{0\}\cup\mathcal{A}_\eta}b_{n,k}=O((\log p)^{1/2})$, and $\eta\lesssim(\log p/n_0)^{1/2}$. In this case, condition \eqref{eq3.5} is satisfied if $(s/p)\vee[(n_{\mathcal{A}_\eta}+n_0)/p^2]\vee(\log p/n_0)\vee [s^2\log^2p/n_0(n_{\mathcal{A}_\eta}+n_0)]\vee\{\left[1+(s\log p/n_0)\right]\eta^2\}\leq c_0$. In particular, this condition allows $s\log p/n_0$ to diverge when the informative sources provide sufficiently large sample sizes and $\eta$ is sufficiently small.
\end{remark}
The proof of Theorem~\ref{thm1} can be found in Appendix \ref{secA1}. The $\ell_1$- and $\ell_2$-error bounds in Theorem~\ref{thm1} improve upon those in \cite{fan2024latent} provided that $(s\log p/n_0)\vee\eta=O(1)$, $\eta\ll s(\log p/n_0)^{1/2}$, and $n_{\mathcal{A}_\eta}+n_0\gg n_0$. This demonstrates the merit gained through transfer learning. The following Theorem characterizes the fundamental limits of transfer learning for the FARM model. Its proof is provided in Appendix \ref{secA2}.

\begin{theorem}\label{thm-minimax rate}
For model \eqref{eq2.1}, we have that, for some constant $C>0$, 
    \begin{align*}
        &\inf_{\hat{\boldsymbol{\beta}}}\!\!\sup_{\boldsymbol{\mathcal{B}}\in\Theta_{\boldsymbol{b}_n}\!(s,\eta)}\!\!\mathbb{P}\Bigg\{\!\|\hat{\boldsymbol{\beta}}-\boldsymbol{\beta}\|_2^2\!\geq \!C\Bigg[\!\frac{s\log(ep/s)}{n_{\mathcal{A}_\eta}\!+n_0}\!+\!\!\Bigg(\!\eta^2\!\wedge\eta\left(\frac{\log(ep/s)}{n_0}\right)^{1/2}\!\!\wedge\!\frac{s\log(ep/s)}{n_0}\!\Bigg)\!\!+\!\frac{sb_{n,\min}^2}{p^2}\!\Bigg]\Bigg\}\!\geq \!\frac{1}{4},\\
        &\inf_{\hat{\boldsymbol{\beta}}}\sup_{\boldsymbol{\mathcal{B}}\in\Theta_{\boldsymbol{b}_n}(s,\eta)}\mathbb{P}\Bigg\{\|\hat{\boldsymbol{\beta}}-\boldsymbol{\beta}\|_1\geq C\Bigg[s\left(\frac{\log(ep/s)}{n_{\mathcal{A}_\eta}+n_0}\right)^{1/2}+\Bigg(s\left(\frac{\log(ep/s)}{n_0}\right)^{1/2}\wedge \eta\Bigg)+\frac{sb_{n,\min}}{p}\Bigg]\Bigg\}\geq \frac{1}{4}.
    \end{align*}
\end{theorem}

\begin{remark}\label{remark 3.2}
    Theorem \ref{thm-minimax rate} reveals the fundamental limits on estimation in the FARM model. When $\eta=0$, the target and informative source domains share the same regression coefficients. Nevertheless, a surprising phenomenon is that the minimax lower bounds still contain the terms $sb_{n,\min}^2/p^2$ under squared $\ell_2$ loss and $sb_{n,\min}/p$ under $\ell_1$ loss. These terms dominate the rates induced by sample complexity when $\log(ep/s)/(n_{\mathcal{A}_\eta}+n_0)\lesssim b_{n,\min}^2/p^2$. Consequently, over the parameter space considered, increasing the sample size alone cannot reduce the minimax lower error below these dimension-dependent orders. This indicates the fundamental limits caused by the factor structure and may be of independent interest.
\end{remark}

\begin{remark}\label{remark 3.3}
    Suppose that the conditions of Theorem \ref{thm1} hold, and $s\leq p^{1-\epsilon}$ holds for some fixed constant $\epsilon\in(0,1)$. If $\eta\vee (s\log p/n_0)=O(1)$ and $\eta\lesssim s(\log p/n_0)^{1/2}$, the upper and lower bounds match under both squared $\ell_2$ and $\ell_1$ losses. Hence, our estimator attains the minimax rates over $\Theta_{\boldsymbol{b}_n}(s,\eta)$. Indeed, Condition \eqref{eq3.5}, together with $\eta=O(1)$, implies that $b_{n,\min}^2/p^2\lesssim \log p/(n_{\mathcal{A}_\eta}+n_0)$. Thus, under this regime, the dimension-dependent terms in the minimax lower bounds are dominated by the pooled-sample terms.
\end{remark}

\subsection{Source Detection and Its Consistency}\label{sec3.2}
In this section, we consider how to detect the transferable set $\mathcal{A}_\eta$ through a data-driven manner. Since the performance of transfer learning deeply depends on the transferable set, the \textit{negative transfer} \citep{pan2009survey,ge2014handling} would happen if irrelevant sources are included in the transferable set. Therefore, a reliable procedure for identifying which source datasets can be effectively leveraged is crucial. Motivated by \cite{tian2023transfer,yang2025communication,yang2025debiased}, we develop a source detection procedure as follows. We begin by randomly partitioning the target dataset into three folds, denoted by $\{(\boldsymbol{X}^{(0)[r]},\boldsymbol{Y}^{(0)[r]})\}_{r=1}^3$. Note that we choose three folds and suppose $n_0$ is divisible by $3$ only for convenience. For theoretical analysis, the preprocessing step is performed separately using every two folds of target data and on the remaining validation fold. The two training folds are treated as one sample when fitting the target-only and pairwise transferring estimators, while the PCA fit on the validation fold is used only to construct the factor-adjusted observations in the subsequent \eqref{eq3.7}. This ensures that the training estimators are constructed independently of the validation observations. Next, for each source dataset, we apply the transferring step using every two folds of target data and evaluate the resulting model on the remaining fold. We compute the corresponding value of the given loss function on the remaining target fold and average across the three splits to obtain the cross-validated loss $\hat{L}_0^{(k)}$ for each source $k$. As a benchmark, we also fit the Lasso estimator using every pair of target folds and evaluate its performance on the remaining fold, yielding the average cross-validated loss $\hat{L}_0^{(0)}$, which serves as the loss of the target-only model. Finally, we compare the difference $\hat{L}_0^{(k)}-\hat{L}_0^{(0)}$ with a pre-specified threshold. We then include source $k$ in $\widehat{\mathcal{A}}$ if $\hat{L}_0^{(k)}-\hat{L}_0^{(0)}$ is below this threshold. Let $\mathcal{R}^{[r]}$ be the set of target observations in the $r$-th fold and write $\mathcal{R}^{[-r]}=[n_0]\setminus\mathcal R^{[r]}$. For each $r$, we use the observations in $\mathcal{R}^{[-r]}$ as one sample for PCA and model fitting. Separately, we apply PCA to the
observations in $\mathcal{R}^{[r]}$ and obtain
$(\hat{\boldsymbol{U}}_0^{[r]},\tilde{\boldsymbol{Y}}^{(0)[r]})$. For any regression coefficient estimate $\boldsymbol{w}$, the loss function under our FARM setting on the $r$-th fold of target data is 
\begin{equation}\label{eq3.7}
	\hat{L}_0^{[r]}(\boldsymbol{w})=\frac{1}{2}\frac{1}{n_0/3}\|\tilde{\boldsymbol{Y}}^{(0)[r]}-\hat{\boldsymbol{U}}_0^{[r]}\boldsymbol{w}\|_2^2.
\end{equation}
After obtaining the estimation $\widehat{\mathcal{A}}$, we apply Algorithm \ref{algo1} with $\widehat{\mathcal{A}}$ as input. The full procedure is described in Algorithm \ref{algo2}.

\begin{algorithm}[!ht]
\caption{Trans-FARM Algorithm}
\label{algo2}
\KwIn{target data $(\boldsymbol{X}^{(0)}, \boldsymbol{Y}^{(0)})$, source data $\{(\boldsymbol{X}^{(k)}, \boldsymbol{Y}^{(k)})\}_{k=1}^K$, penalty parameters $\{\{\lambda^{(k)[r]}\}_{k=0}^K\}_{r=1}^3$, constants $\epsilon_0$ and $\hat{\sigma}$.}
\KwOut{the estimated coefficient vector $\hat{\boldsymbol{\beta}}$, and the determined transferring set $\widehat{\mathcal{A}}$}
\underline{\textbf{Transferable source detection}}:
Randomly divide $(\boldsymbol{X}^{(0)}, \boldsymbol{Y}^{(0)})$ into three sets of equal size as $\{(\boldsymbol{X}^{(0)[r]}, \boldsymbol{Y}^{(0)[r]})\}_{r=1}^3$ \\
\For{$r = 1$ \KwTo $3$}
	{\textbf{Fold-specific preprocessing}:
    Run step 1 of Algorithm \ref{algo1} on the validation fold, obtaining $(\hat{\boldsymbol{U}}_0^{[r]},\tilde{\boldsymbol{Y}}^{(0)[r]})$.\\
    $\hat{\boldsymbol{\beta}}^{[r]} \leftarrow$ fit the FARM on $\{(\boldsymbol{X}^{(0)[i]}, \boldsymbol{Y}^{(0)[i]})\}_{i=1}^3 \backslash (\boldsymbol{X}^{(0)[r]}, \boldsymbol{Y}^{(0)[r]})$ with penalty parameter $\lambda^{(0)[r]}$ (run step 2 in Algorithm \ref{algo1} with data $\{(\boldsymbol{X}^{(0)[i]}, \boldsymbol{Y}^{(0)[i]})\}_{i=1}^3\backslash(\boldsymbol{X}^{(0)[r]}, \boldsymbol{Y}^{(0)[r]})$ and penalty parameter $\lambda^{(0)[r]}$).\\
	$\hat{\boldsymbol{w}}^{(k)[r]} \leftarrow$ run step 2 (transferring step) in Algorithm \ref{algo1} with $(\{(\boldsymbol{X}^{(0)[i]}, \boldsymbol{Y}^{(0)[i]})\}_{i=1}^3\backslash(\boldsymbol{X}^{(0)[r]}, \boldsymbol{Y}^{(0)[r]}))\cup (\boldsymbol{X}^{(k)}, \boldsymbol{Y}^{(k)})$ and penalty parameter $\lambda^{(k)[r]}$ for all $k \neq 0$.\\
	Compute the loss $\hat{L}_0^{[r]}(\hat{\boldsymbol{w}}^{(k)[r]})$ for $k\neq0$ and $\hat{L}_0^{[r]}(\hat{\boldsymbol{\beta}}^{[r]})$ on $(\boldsymbol{X}^{(0)[r]}, \boldsymbol{Y}^{(0)[r]})$.\\
}
$\hat{L}_0^{(k)} \leftarrow \sum_{r=1}^3\hat{L}_0^{[r]}(\hat{\boldsymbol{w}}^{(k)[r]})/3$ for $k\neq 0$, $\hat{L}_0^{(0)} \leftarrow \sum_{r=1}^3 \hat{L}_0^{[r]}(\hat{\boldsymbol{\beta}}^{[r]})/3$\\
$\widehat{\mathcal{A}} \leftarrow \{k \neq 0: \hat{L}_0^{(k)}\leq \hat{L}_0^{(0)}+\epsilon_0\hat{\sigma}^2\}$ \\
\underline{\textbf{$\widehat{\mathcal{A}}$-Trans-FARM}}: $\hat{\boldsymbol{\beta}} \leftarrow$ run Algorithm \ref{algo1} using $\{(\boldsymbol{X}^{(k)}, \boldsymbol{Y}^{(k)})\}_{k\in \{0\} \cup \widehat{\mathcal{A}}}$
\end{algorithm}

For theoretical analysis, Algorithm \ref{algo2} employs fold-specific preprocessing to ensure independence between the training estimators and the validation observations. In practice, to reduce computational cost and avoid the case that the target sample size is sufficiently small, we alternatively perform the preprocessing step of Algorithm \ref{algo1} once using the full sample of each dataset before splitting the target data into three folds. The resulting factor-adjusted observations are then directly used for source detection, with the remaining steps of Algorithm \ref{algo2} unchanged. We investigate the empirical performance of this computationally simpler implementation in Section \ref{sec5}. To ensure the consistency of Algorithm \ref{algo2}, we impose some assumptions below.

\begin{assumption}\label{assum4}
\leavevmode\par\vspace{0.3ex}%
\begin{itemize}
    \item (a) There exists a positive constant $C_2<\infty$ such that $$\max_{k\in[K]}\left\|\left\{[n_0/(n_0+n_k)]\boldsymbol{\Sigma}_{\boldsymbol{u}}^{(0)}+[n_k/(n_0+n_k)]\boldsymbol{\Sigma}_{\boldsymbol{u}}^{(k)}\right\}^{-1}\boldsymbol{\Sigma}_{\boldsymbol{u}}^{(k)}\right\|_1\leq C_2.$$ Moreover, assume that $(s\log p/n_0)+(s/p)\leq c$ for some sufficiently small constant $c>0$, and $\mathcal{V}_{n_0,p}\|\boldsymbol{\varphi}^{(0)}\|_2/(n_0\log p)^{1/2}=O(1)$.
    \item (b) Let $\tilde{\eta}\geq \eta$ be a deterministic quantity that may depend on the sample sizes and $p$. Assume that $\max_{k\in\mathcal{A}_{\eta}^{c}}\|\boldsymbol{w}^{(k)}-\boldsymbol{\beta}\|_1\leq\tilde{\eta}$, and $\tilde{\eta}$ satisfies $\left\{\mathcal{V}_{n_k,p}\|\boldsymbol{\varphi}^{(k)}\|_2/\left[(1\vee\tilde{\eta})(n_k+n_0)\log p\right]^{1/2}\right\}\vee\left\{\tilde{\eta}(n_k+n_0)^{1/2}/\left[(1\vee\tilde{\eta})p(\log p)^{1/2}\right]\right\}=O(1),~\forall k\in[K]$.
    \item (c) Assume that for some $\eta$, $(s\log p/n_0+\eta^2+\max_{k\in\mathcal{A}_\eta}\tilde{\eta}^2s\log p/(n_0+n_k))/\epsilon_0=o(1),~\epsilon_0\leq c\inf_{k\in\mathcal{A}_\eta^c}[n_k/(n_0+n_k)]^2\|\boldsymbol{w}^{(k)}-\boldsymbol{\beta}\|_2^2$ for some sufficiently small constant $c>0$, where $\epsilon_0$ is defined in Algorithm \ref{algo2}, and
    \begin{align*}
        \max_{k\in\mathcal{A}_\eta^c}\bigg[\frac{s\log p}{n_0}+(\tilde{\eta}\vee\tilde{\eta}^2)\left(\frac{\log p}{n_0}\right)^{1/2}+\tilde{\eta}^2\frac{s\log p}{n_0+n_k}\bigg]\Big/\bigg[\frac{n_k^2}{(n_0+n_k)^2}\|\boldsymbol{w}^{(k)}-\boldsymbol{\beta}\|_2^2\bigg]=o\left(1\right).
    \end{align*}
    \item (d) Suppose that $\sigma_0^2=\mbox{var}(\mathcal{E}^{(0)})$ with $0<\sigma_0^2<\infty$, and $\mathbb{P}(c_\sigma\leq\hat{\sigma}/\sigma_0\leq C_\sigma)\to 1$ for some constants $0<c_\sigma<C_\sigma<\infty$, where $\hat{\sigma}$ is defined in Algorithm \ref{algo2}.
\end{itemize}
\end{assumption}

\begin{remark}\label{remark3.4}
    Assumption \ref{assum4} (a)(b)(c) ensures that the transferring step can be applied to attain some estimator for the consistency of source detection. Similar conditions can be found in \cite{tian2023transfer,yang2025communication,yang2025debiased,zhang2025transfer}. Assumption \ref{assum4} (d) is to ensure that $\hat{\sigma}^2$ used in Algorithm \ref{algo2} is the order of a constant such that our source detection algorithm can be proved to be consistent. This type of choice of $\hat{\sigma}$ is similar to \cite{tian2023transfer,yang2025communication,yang2025debiased}. As stated in \cite{yang2025debiased}, the source detection is not sensitive to the choice of $\epsilon_0\hat{\sigma}^2$ in practice. 
\end{remark}
Under Assumption \ref{assum4}, Theorem \ref{thm2} establishes the consistency of our proposed source detection algorithm. The proof is provided in Appendix \ref{secA3}.
\begin{theorem}\label{thm2}
    Suppose that Assumptions \ref{assum1}-\ref{assum2} hold for $k\in\{0\}\cup[K]$ and Assumptions \ref{assum3}-\ref{assum4} also hold. In addition, assume that $n_0=O(p\log p)$, $\log p/n_0=O(1)$, $\log K\leq c(\min_{k\in\{0\}\cup[K]}n_k\wedge p)$ for some sufficiently small constant $c>0$, $K=o(p^{100})$, and condition \eqref{eq3.5} are satisfied. If we choose $\lambda^{(0)[r]}=C_{[0]}(\log p/n_0)^{1/2}$ and $\lambda^{(k)[r]}=C_{[k]}(1\vee\tilde{\eta})(\log p/(n_0+n_k))^{1/2}$ for $k\in[K]$ for some sufficiently large constants $C_{[0]}$ and $C_{[k]}$, then $\mathbb{P}(\widehat{\mathcal{A}}=\mathcal{A}_\eta)\to1$.
\end{theorem}
\vspace{-0.2cm}

\section{Simultaneous Inference}\label{sec4}
In this section, we study the construction of confidence intervals for the regression coefficients $\boldsymbol{\beta}$ in a more stable and reliable way by leveraging auxiliary datasets, as it is important to quantify the uncertainty. Both entrywise confidence intervals \citep{van2014asymptotically,javanmard2014confidence,tian2023transfer} and groupwise simultaneous confidence intervals \citep{zhang2017simultaneous,fan2024latent,yan2023confidence,cai2025statistical} can be obtained. We focus on the construction of groupwise simultaneous confidence intervals, as the entrywise intervals are simpler and arise as a direct byproduct of the groupwise procedure. Let $\mathcal{G}\subseteq[p]$ be the indices of parameters of interest, and denote $\|\boldsymbol{Q}\|_{\mathcal{G}}=\|\boldsymbol{Q}_{\mathcal{G}}\|_{\infty}$ as the maximum norm on the sub-vector $\boldsymbol{Q}_{\mathcal{G}}$. The key is to attain a debiased-type estimator via constructing the inverse matrix of the Hessian matrix of the population loss function on the target dataset, i.e. $\boldsymbol{\Theta}=(\boldsymbol{\Sigma}_{\boldsymbol{u}}^{(0)})^{-1}$.

Note that the transfer learning is also extended to attain a better estimator of $\boldsymbol{\Theta}$ by using nodewise regression on the target dataset and the informative dataset $\widehat{\mathcal{A}}_\eta$ \citep{tian2023transfer,yang2025communication,yang2025debiased}. However, there is no guarantee that the informative sets for regression coefficients and design matrix are the same. Another problem is that the regressors $\hat{\boldsymbol{U}}^{(k)}$ in FARM are estimated such that estimating $\boldsymbol{\Theta}$ via transfer learning does not improve estimation accuracy. Hence, we construct the debiased estimator for $\boldsymbol{\beta}$ as 
\begin{align}\label{eq4.1}
    \tilde{\boldsymbol{\beta}}=\hat{\boldsymbol{\beta}}+\frac{1}{n_0}\hat{\boldsymbol{\Theta}}\hat{\boldsymbol{U}}_0^\top(\tilde{\boldsymbol{Y}}^{(0)}-\hat{\boldsymbol{U}}_0\hat{\boldsymbol{\beta}})
\end{align}
where $\hat{\boldsymbol{\beta}}$ is the estimator attained via Algorithm \ref{algo1}, and $\hat{\boldsymbol{\Theta}}$ is an estimator of $\boldsymbol{\Theta}$ by estimating the inverse matrix of $\hat{\boldsymbol{\Sigma}}_{\boldsymbol{u}}^{(0)}=\hat{\boldsymbol{U}}_0^{\top}\hat{\boldsymbol{U}}_0/n_0$. We do not constrain $\boldsymbol{\Theta}$ to be attained from any specific type, but only need to satisfy the following assumption.
\begin{assumption}\label{assum5}
    There exist positive $\Delta_1$, $\Delta_{\max}$, $\Delta_{\infty}$, $0<c_\sigma\leq\sigma_0\leq C_\sigma<\infty$, and $0<\Delta_{\sigma}\leq 1$ for some constant $c$ such that $|\hat{\sigma}/\sigma_0-1|=O_P(\Delta_{\sigma})$, and $$\|\boldsymbol{I}_p-\hat{\boldsymbol{\Theta}}\hat{\boldsymbol{\Sigma}}_{\boldsymbol{u}}^{(0)}\|_{\max}=O_P(\Delta_1),\|\hat{\boldsymbol{\Theta}}-\boldsymbol{\Theta}\|_{\max}=O_P(\Delta_{\max}),\|\hat{\boldsymbol{\Theta}}-\boldsymbol{\Theta}\|_\infty=O_P(\Delta_\infty),$$
    where $\sigma_0^2=\mbox{var}(\mathcal{E}^{(0)})$, and $\Delta_1$, $\Delta_{\max}$, $\Delta_{\infty}$, and $\Delta_{\sigma}$ satisfy $$(n_0\log p)^{1/2}\Delta_1\mathcal{R}_1+\Delta_\infty\log p\to 0\quad\text{and}\quad\left(\mathcal{V}_{n_0,p}\|\boldsymbol{\varphi}^{(0)}\|_2+(\log p)^{1/2}\right)\|\boldsymbol{\Theta}\|_\infty\left(\frac{\log p}{n_0}\right)^{1/2}\to 0,$$
    where $\mathcal{R}_1$ is the $\ell_1$-error bound for $\hat{\boldsymbol{\beta}}$ established in Theorem \ref{thm1}. We further assume that $\|\hat{\boldsymbol{\Theta}}-\boldsymbol{\Theta}\|_\infty=O_P(\|\boldsymbol{\Theta}\|_\infty)$ without loss of generality.
\end{assumption}
This assumption on the estimation of $\hat{\boldsymbol{\Theta}}$ can be easily satisfied for the CLIME-type estimator \citep{cai2011constrained,javanmard2014confidence,yan2023confidence,cai2025statistical} and the nodewise regression \citep{zhang2014confidence,van2014asymptotically,zhang2017simultaneous}. For the assumption on the estimation of $\hat{\sigma}$, one can adopt the refitted cross-validation approach in \cite{fan2012variance} or the scaled-Lasso \citep{sun2012scaled} to attain a consistent estimation of $\sigma_0$ by following \cite{fan2024latent}. 

Now we begin by the case that $\mathcal{G}=[p]$ and extend it to any subset $\mathcal{G}\subseteq[p]$. The following theorem establishes a uniform Gaussian approximation for the sup-norm of the debiased estimator $\tilde{\boldsymbol{\beta}}$, which forms the basis for simultaneous inference. 
\begin{theorem}\label{thm3}
    Suppose that $\hat{\boldsymbol{\beta}}$ satisfies the error bounds in Theorem \ref{thm1}, $\log^7 (pn_0)/n_0=o(1)$, $n_0=O(p\log p)$, Assumptions \ref{assum1}-\ref{assum2} hold for $k=0$, and Assumption \ref{assum5} also holds, then there exists a $p$-dimensional Gaussian random vector $\boldsymbol{Q}$ with mean $\boldsymbol{0}_p$ and variance $\sigma_0^2\boldsymbol{\Theta}$ such that 
    $$\sup_{t\in\mathbb{R}}\left|\mathbb{P}\left(\|n_0^{1/2}(\tilde{\boldsymbol{\beta}}-\boldsymbol{\beta})\|_\infty\leq t\right)-\mathbb{P}(\|\boldsymbol{Q}\|_\infty\leq t)\right|=o(1).$$
\end{theorem}
The proof of Theorem~\ref{thm3} is provided in Appendix \ref{secA4}. From Theorem \ref{thm3}, it is natural to attain the $(1-\alpha)$ quantile of $\boldsymbol{Q}$, denoted by $c_{1-\alpha}$, and then construct the simultaneous confidence intervals. However, when $p$ is large, it is still unclear whether the sup-norm of $\boldsymbol{Q}$ can be reliably estimated by sampling from the $p$-dimensional Gaussian distribution with an estimated covariance matrix. Hence, we adopt the multiplier bootstrap procedure. Let $e_i\sim \mathbb{N}(0,1),i=1,\cdots,n_0$ be independent of the data $\mathcal{D}=\{(\boldsymbol{x}_i^{(0)},y_i^{(0)},\boldsymbol{f}_i^{(0)},\boldsymbol{u}_i^{(0)},\mathcal{E}_i^{(0)})\}_{i=1}^{n_0}$, and consider $\hat{\boldsymbol{Q}}^e=\hat{\sigma}\hat{\boldsymbol{\Theta}}\sum_{i=1}^{n_0}e_i\hat{\boldsymbol{u}}_{i}^{(0)}/n_0^{1/2}$. Denote the $\alpha$-quantile of $\|\hat{\boldsymbol{Q}}^e\|_\infty$ by $\hat{c}(\alpha)=\inf_{t\in\mathbb{R}}\{t\in\mathbb{R}:\mathbb{P}_e(\|\hat{\boldsymbol{Q}}^e\|_\infty\leq t)\geq\alpha\}$, where $\mathbb{P}_e(\mathcal{H})$ represents the probability of the event $\mathcal{H}$ with respect to $e_1,\cdots,e_{n_0}$. The validity of the bootstrap approach is established by the following theorem.
\begin{theorem}\label{thm4}
    Under the conditions of Theorem \ref{thm3}, if $(\log\log p)\log p(\Delta_1\|\boldsymbol{\Theta}\|_\infty+\Delta_{\infty}+\Delta_{\sigma})=o(1),$
    we have $\sup_{t\in\mathbb{R}}\left|\mathbb{P}\left(\|\hat{\boldsymbol{Q}}^e\|_\infty\leq t|\mathcal{D}\right)-\mathbb{P}\left(\|\boldsymbol{Q}\|_\infty\leq t\right)\right|=o_P(1).$
\end{theorem}
The proof of Theorem~\ref{thm4} can be found in Appendix \ref{secA5}. Following Theorem \ref{thm3} and Theorem \ref{thm4}, our simultaneous confidence interval for $\boldsymbol{\beta}$ is $\left[\tilde{\beta}_i-n_0^{-1/2}\hat{c}(1-\alpha),\tilde{\beta}_i+n_0^{-1/2}\hat{c}(1-\alpha)\right],i\in[p].$

In contrast to the target-only inferential scheme of \cite{fan2024latent}, our procedure is more reliable because the estimator $\hat{\boldsymbol{\beta}}$ obtained from Algorithm \ref{algo1} provides a more accurate estimate of $\boldsymbol{\beta}$. In particular, Assumption \ref{assum5} implies that part of the inferential error captured by $\mathcal{R}_1$ can be substantially improved as long as $\mathcal{R}_1$ is smaller than the traditional $\ell_1$-error $s(\log p/n_0)^{1/2}$ in \cite{fan2024latent}. Now we extend the previous results for any $\mathcal{G}\subseteq[p]$. The following corollary presents the validity of constructing simultaneous confidence intervals for $\boldsymbol{\beta}_{\mathcal{G}}$ via bootstrap procedure.

\begin{corollary}\label{cor1}
   Under the conditions of Theorem \ref{thm4}, there exists a $p$-dimensional Gaussian random vector $\boldsymbol{Q}$ with mean $\boldsymbol{0}_p$ and variance $\sigma_0^2\boldsymbol{\Theta}$ such that 
    \begin{align*}
        &\sup_{t\in\mathbb{R}}\left|\mathbb{P}\left(\|n_0^{1/2}(\tilde{\boldsymbol{\beta}}-\boldsymbol{\beta})\|_{\mathcal{G}}\leq t\right)-\mathbb{P}(\|\boldsymbol{Q}\|_{\mathcal{G}}\leq t)\right|=o(1),\\
        &\sup_{t\in\mathbb{R}}\left|\mathbb{P}\left(\|\hat{\boldsymbol{Q}}^e\|_{\mathcal{G}}\leq t|\mathcal{D}\right)-\mathbb{P}\left(\|\boldsymbol{Q}\|_{\mathcal{G}}\leq t\right)\right|=o_P(1).
    \end{align*}
\end{corollary}
The proof of Corollary~\ref{cor1} is presented in Appendix \ref{secA6}. From Corollary \ref{cor1}, let the $\alpha$-quantile of $\|\hat{\boldsymbol{Q}}^e\|_{\mathcal{G}}$ by $\hat{c}_{\mathcal{G}}(\alpha)=\inf_{t\in\mathbb{R}}\{t\in\mathbb{R}:\mathbb{P}_e(\|\hat{\boldsymbol{Q}}^e\|_{\mathcal{G}}\leq t)\geq\alpha\},$ the simultaneous $(1-\alpha)$-level confidence interval for $\boldsymbol{\beta}_{\mathcal{G}}$ can be constructed as 
$\left[\tilde{\beta}_i-n_0^{-1/2}\hat{c}_{\mathcal{G}}(1-\alpha),\tilde{\beta}_i+n_0^{-1/2}\hat{c}_{\mathcal{G}}(1-\alpha)\right],i\in\mathcal{G}.$ Moreover, we propose a studentized version $n_0^{1/2}(\hat{\boldsymbol{\Theta}}_{j,j})^{-1/2}(\tilde{\boldsymbol{\beta}}-\boldsymbol{\beta})_j$ for $j\in\mathcal{G}$ by leveraging the idea in \cite{zhang2017simultaneous,cai2025statistical}, which can also be considered to attain confidence intervals with varying length. Denote $\boldsymbol{S}=\mbox{diag}(\boldsymbol{\Theta})$ and $\hat{\boldsymbol{S}}=\mbox{diag}(\hat{\boldsymbol{\Theta}})$. Let $\hat{\boldsymbol{Q}}_{stu}^{e}=\hat{\sigma}\hat{\boldsymbol{S}}^{-1/2}\hat{\boldsymbol{\Theta}}\frac{1}{n_0^{1/2}}\sum_{i=1}^{n_0}e_i\hat{\boldsymbol{u}}_{i}^{(0)}$, and denote the $\alpha$-quantile of $\|\hat{\boldsymbol{Q}}_{stu}^{e}\|_{\mathcal{G}}$ by $\hat{c}_{\mathcal{G},stu}(\alpha)=\inf_{t\in\mathbb{R}}\{t\in\mathbb{R}:\mathbb{P}_e(\|\hat{\boldsymbol{Q}}_{stu}^{e}\|_{\mathcal{G}}\leq t)\geq\alpha\}.$
The validity of the studentized bootstrap procedure is established in the following theorem.
\begin{theorem}\label{thm5}
    Under the conditions of Theorem \ref{thm4}, there exists a $p$-dimensional Gaussian random vector $\boldsymbol{Q}_{stu}$ with mean $\boldsymbol{0}_p$ and variance $\sigma_0^2\boldsymbol{S}^{-1/2}\boldsymbol{\Theta}\boldsymbol{S}^{-1/2}$ such that 
    \begin{align*}
        &\sup_{t\in\mathbb{R}}\left|\mathbb{P}\left(\|n_0^{1/2}\hat{\boldsymbol{S}}^{-1/2}(\tilde{\boldsymbol{\beta}}-\boldsymbol{\beta})\|_{\mathcal{G}}\leq t\right)-\mathbb{P}(\|\boldsymbol{Q}_{stu}\|_{\mathcal{G}}\leq t)\right|=o(1),\\
        &\sup_{t\in\mathbb{R}}\left|\mathbb{P}\left(\|\hat{\boldsymbol{Q}}_{stu}^{e}\|_{\mathcal{G}}\leq t|\mathcal{D}\right)-\mathbb{P}(\|\boldsymbol{Q}_{stu}\|_{\mathcal{G}}\leq t)\right|=o_P(1).
    \end{align*}
\end{theorem}
The proof of Theorem~\ref{thm5} is given in Appendix \ref{secA7}. Following Theorem \ref{thm5}, the simultaneous $(1-\alpha)$-level confidence interval with varying length for $\boldsymbol{\beta}_{\mathcal{G}}$ can be constructed as 
$\left[\tilde{\beta}_i-n_0^{-1/2}(\hat{\boldsymbol{\Theta}}_{i,i})^{1/2}\hat{c}_{\mathcal{G},stu}(1-\alpha),\tilde{\beta}_i+n_0^{-1/2}(\hat{\boldsymbol{\Theta}}_{i,i})^{1/2}\hat{c}_{\mathcal{G},stu}(1-\alpha)\right],i\in\mathcal{G}.$ The complete steps for obtaining $\tilde{\boldsymbol{\beta}}$ and constructing simultaneous confidence intervals with varying lengths are summarized in Algorithm \ref{algo3}. The corresponding procedure for the non-studentized version follows an analogous structure and is therefore omitted for brevity.

\begin{algorithm}[!ht]
\caption{Simultaneous confidence interval construction}
\label{algo3}
\KwIn{target data $(\boldsymbol{X}^{(0)}, \boldsymbol{Y}^{(0)})$, source data $\{(\boldsymbol{X}^{(k)}, \boldsymbol{Y}^{(k)})\}_{k=1}^K$, $\mathcal{G}$, $\alpha$, $B$, and $\hat{\sigma}$.}
\KwOut{the simultaneous confidence interval for $\boldsymbol{\beta}_{\mathcal{G}}$}
Run Algorithm \ref{algo2} to attain $\hat{\boldsymbol{\beta}}$, estimate $\hat{\boldsymbol{\Theta}}$ using CLIME or nodewise regression, and calculate $\tilde{\boldsymbol{\beta}}$ via \eqref{eq4.1}.\\
\For{$\ell = 1$ \KwTo $B$}{
	Draw standard normal random variables $e_{\ell,i},i=1,\cdots,n_0$.\\
    Calculate $\hat{\boldsymbol{Q}}_{stu,\ell}^{e}=\hat{\sigma}\hat{\boldsymbol{S}}^{-1/2}\hat{\boldsymbol{\Theta}}\sum_{i=1}^{n_0}e_{\ell,i}\hat{\boldsymbol{u}}_i^{(0)}/n_0^{1/2}$
}
\underline{\textbf{Simultaneous confidence intervals construction}}: Find the $(1-\alpha)$-quantile of $\{\|\hat{\boldsymbol{Q}}_{stu,\ell}^{e}\|_{\mathcal{G}}\}_{\ell=1}^{B}$, represented as $\hat{c}_{\mathcal{G},stu}(1-\alpha)$. The simultaneous $(1-\alpha)$-level confidence interval for $\beta_i$ is $\mathcal{I}_i=\left[\tilde{\beta}_i-n_0^{-1/2}(\hat{\boldsymbol{\Theta}}_{i,i})^{1/2}\hat{c}_{\mathcal{G},stu}(1-\alpha),\tilde{\beta}_i+n_0^{-1/2}(\hat{\boldsymbol{\Theta}}_{i,i})^{1/2}\hat{c}_{\mathcal{G},stu}(1-\alpha)\right],i\in\mathcal{G}.$
Output the simultaneous confidence intervals $\{\mathcal{I}_i\}_{i\in\mathcal{G}}$.
\end{algorithm}
\vspace{-0.3cm}

\section{Numerical Studies}\label{sec5}
In this section, we provide numerical experiments to illustrate the finite-sample performance of our proposed Trans-FARM estimator under the FARM framework. The code is made available at \href{https://github.com/BoFuxjtu/Trans-FARM.git}{https://github.com/BoFuxjtu/Trans-FARM.git}.\vspace{-0.1cm}
\subsection{Simulations of Estimation}\label{sec5.1}
In this section, we study the estimation performance of different methods under various settings of $\eta$. The methods include Only-FARM (Algorithm \ref{algo1} with $\mathcal{A}$ being $\varnothing$; \cite{fan2024latent}), Trans-FARM (Algorithm \ref{algo2}), Oracle-Trans-FARM (Algorithm \ref{algo1} with $\mathcal{A}=\mathcal{A}_\eta$), Pooled-Trans-FARM (Algorithm \ref{algo1} with $\mathcal{A}$ being all sources), Only-Lasso (Lasso on target data), Trans-Lasso \citep[Algorithm 2]{tian2023transfer}, Oracle-Trans-Lasso (Algorithm 1 in \cite{tian2023transfer} with $\mathcal{A}=\mathcal{A}_\eta$), Pooled-Trans-Lasso (Algorithm 1 in \cite{tian2023transfer} with $\mathcal{A}$ being all sources). For source detection, we use the default \texttt{glmtrans} cutoff for the loss difference. The penalty parameters are selected by cross-validation using the default settings of the \texttt{glmtrans} package \cite{tian2023transfer}. All experiments are conducted in R.

We generate the data from the FARM model as follows. We let $r_0=\ldots=r_K=2$, and the design matrix is generated by  $\boldsymbol{X}^{(k)}=\boldsymbol{F}_k\boldsymbol{B}_k^\top+\boldsymbol{U}_k\in\mathbb{R}^{n_k\times p}$ with entries of $\boldsymbol{B}_k\in\mathbb{R}^{p\times r_k}$ generated from $\mbox{Unif}(-1,1)$, entries of $\boldsymbol{F}_k\in\mathbb{R}^{n_k\times r_k}$ generated from $\mathbb{N}(0,1)$, and $\boldsymbol{U}_k=(\boldsymbol{u}_1^{(k)\top},\ldots,\boldsymbol{u}_{n_k}^{(k)\top})^\top\in\mathbb{R}^{n_k\times p}$ with each row being generated from $\mathbb{N}(\boldsymbol{0}_p,\boldsymbol{\Sigma}_k)$ or multivariate t distribution $t_{10}(\boldsymbol{0},\boldsymbol{\Sigma}_k)$, where $\boldsymbol{\Sigma}_0$ is a Toeplitz matrix with $\Sigma_{i,j}=0.5^{|i-j|}$ and $\boldsymbol{\Sigma}_k=\boldsymbol{\Sigma}_0+\boldsymbol{\epsilon}\boldsymbol{\epsilon}^\top$ with $\boldsymbol{\epsilon}\sim\mathbb{N}(\boldsymbol{0}_p,0.2^2\boldsymbol{I}_p)$ for $k\in\{1,\ldots,K\}$. The response vector $\boldsymbol{Y}$ follows $\boldsymbol{Y}^{(k)}=\boldsymbol{U}_k\boldsymbol{w}^{(k)}+\boldsymbol{F}_k\boldsymbol{\gamma}^{(k)}+\boldsymbol{E}^{(k)}$ with each entry of $\boldsymbol{E}^{(k)}\in\mathbb{R}^{n_k}$ being generated from the standard normal distribution $\mathbb{N}(0,1)$ or $t_5$ distribution. The regression coefficient is set as $\boldsymbol{w}^{(0)}=(0.5\cdot\boldsymbol{1}_s,\boldsymbol{0}_{p-s})^\top$, and $\boldsymbol{w}^{(k)}=\boldsymbol{w}^{(0)}+(\eta/p)\mathcal{R}_1^{(k)}$ for $k\in\mathcal{A}_\eta$ and $\boldsymbol{w}^{(k)}=\boldsymbol{w}^{(0)}+(2\eta/p)\boldsymbol{\mathcal{R}}_1^{(k)}$ for $k\in\mathcal{A}_\eta^c$, where $\boldsymbol{\mathcal{R}}_1^{(k)}$ is a $p$-dimensional vector with entries being independent Rademacher variables and $\boldsymbol{\mathcal{R}}_1^{(k)}$ and $\boldsymbol{\mathcal{R}}_1^{(k')}$ are mutually independent for any $k\neq k'\neq0$. For each noninformative source, we additionally add $0.5$ to the coefficients at indices $s+1,\ldots,2s$ and at $s$ further indices sampled uniformly without replacement from $\{2s+1,\ldots,p\}$. The impact of factors $\boldsymbol{\gamma}^{(k)}$ is defined as $\boldsymbol{\gamma}^{(0)}=(0.5,0.5)^\top$, and $\boldsymbol{\gamma}^{(k)}=\boldsymbol{\gamma}^{(0)}+0.1\boldsymbol{\mathcal{R}}_2^{(k)}$ for $k\in\mathcal{A}_\eta$ and $\boldsymbol{\gamma}^{(k)}=\boldsymbol{\gamma}^{(0)}+0.5\boldsymbol{\mathcal{R}}_2^{(k)}$ for $k\in\mathcal{A}_\eta^c$, where $\{\boldsymbol{\mathcal{R}}_2^{(k)}\}_{k=1}^{K}$ are i.i.d. $2$-dimensional Rademacher vectors similar to the previous definition. Note that this design does not satisfy the separation condition in Assumption~\ref{assum4}(c), we provide additional results under a design satisfying this condition in the Supplementary Material.
\begin{figure}[!ht]
    \centering
    \includegraphics[width=0.8\textwidth]{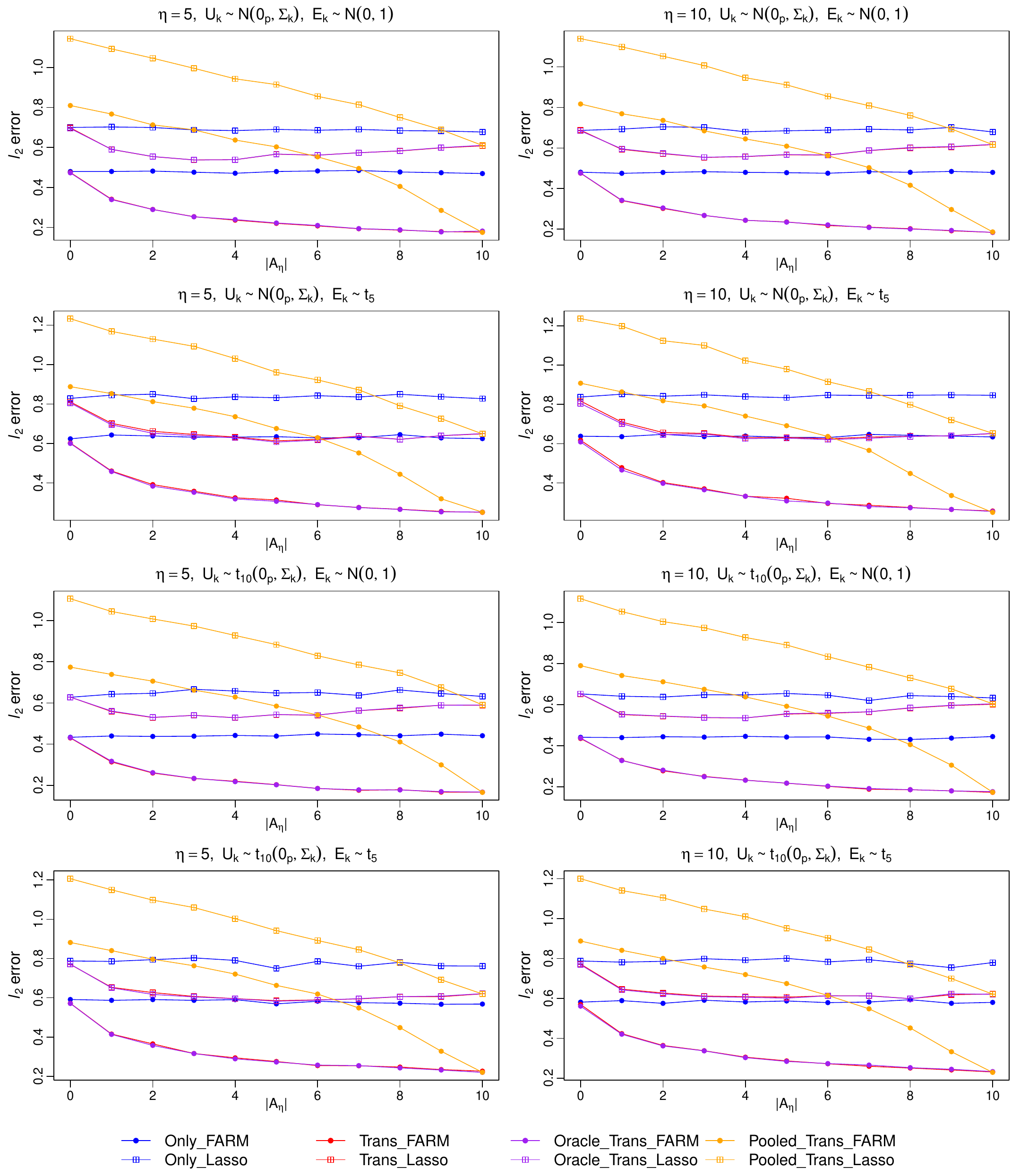}
    
    \captionsetup{font=footnotesize}
    \caption{Comparison of $\ell_2$ errors averaged over $200$ replications versus the number of informative sources ($|\mathcal{A}_\eta|$) across four cases with $\eta\in\{5,10\}$, $K=10$, $n_0=n_1=\ldots=n_K=300$, $p=500$, and $s=20$.}
    \label{fig1}
\end{figure}

We set $n_0=300$, $n_1=\ldots=n_{K}=300$, $p=500$, $s=20$ and $\eta=5$ or $\eta=10$. We let $K=10$ and the informative set $\mathcal{A}_\eta$ is random selected from $\{1,\cdots,K\}$ with $|\mathcal{A}_\eta|$ varying from $0$ to $10$. We repeat each simulation $200$ times and exhibit the $\ell_2$ errors of the $8$ estimators under the four scenarios in Figure \ref{fig1}: (a) $\boldsymbol{U}_k\sim\mathbb{N}(\boldsymbol{0}_p,\boldsymbol{\Sigma}_k)$, $\boldsymbol{E}^{(k)}\sim \mathbb{N}(0,1)$; (b) $\boldsymbol{U}_k\sim\mathbb{N}(\boldsymbol{0}_p,\boldsymbol{\Sigma}_k)$, $\boldsymbol{E}^{(k)}\sim t_5$; (c) $\boldsymbol{U}_k\sim t_{10}(\boldsymbol{0},\boldsymbol{\Sigma}_k)$, $\boldsymbol{E}^{(k)}\sim \mathbb{N}(0,1)$; (d) $\boldsymbol{U}_k\sim t_{10}(\boldsymbol{0},\boldsymbol{\Sigma}_k)$, $\boldsymbol{E}^{(k)}\sim t_5$.

From Figure \ref{fig1}, we draw several conclusions. First, across all experimental regimes, FARM-based estimators outperform Lasso-based estimators in \cite{tian2023transfer}, consistent with the model misspecification of Lasso in the presence of latent factors. Second, Oracle-Trans-FARM exhibits smaller estimation error than the Target-only estimators, with errors decreasing as the number of informative sources increases, which aligns with the theoretical benefits of borrowing strength across similar datasets. Moreover, the benefits of Oracle-Trans-FARM compared with Oracle-Trans-Lasso increase as the number of informative sources increases, which states that Oracle-Trans-FARM can capture more information from the informative sources in this setting. Third, Trans-FARM achieves estimation errors close to those of Oracle-Trans-FARM in these settings, indicating that our source detection algorithm effectively mitigates the risk of negative transfer. Finally, the poor performance of Pooled-Trans-FARM relative to Trans-FARM highlights the necessity of the source detection algorithm. These results collectively show that Trans-FARM achieves near-oracle accuracy even though this design does not satisfy Assumption~\ref{assum4}(c). 
\vspace{-0.15cm}

\subsection{Simulations of Simultaneous Inference}\label{sec5.2}
In this section, we generate data as in Section \ref{sec5.1}, except that we set $|\mathcal{A}_\eta|=5$, $\eta=5$, and the $s=5$ nonzero coefficients of $\boldsymbol{\beta}$ to $1$. The data is simulated with the settings $n_1=\ldots=n_K=500$, and we let $n_0\in\{200,500\}$ and $p\in\{200,500\}$. The parameters of interest correspond to the first $p/4$ features, and our objective is to construct the simultaneous confidence intervals for these features. \cite{fan2024latent} develop a factor-adjusted debiased test with good size control and power in their simulations. Following their discussion of simultaneous confidence intervals, we implement their debiasing framework using only the target data, denoted by Only-FARM, and compare it with Trans-FARM (Algorithm \ref{algo3}) to assess the benefits of auxiliary data. Both methods use the same target-based nodewise estimate $\hat{\boldsymbol{\Theta}}$ of $\boldsymbol{\Theta}$ \citep{van2014asymptotically} and the same $\hat{\sigma}$ obtained by refitted cross-validation \citep{fan2012variance}, following \cite{fan2024latent}. The results (in $500$ replications) of average coverage probability (ACP) and average length (AL) of $95\%$ confidence intervals under four cases in Section \ref{sec5.1} are presented in Table~\ref{tab1}.
\begin{table*}[htbp]
	\caption{Comparison of ACP and AL averaged over $500$ replications for $p,n_0\in\{200,500\}$.}
    \centering
	\resizebox{0.9\linewidth}{!}{
		\begin{tabular}{ccccccccccccc}
			\hline
			\multirow{2}{*}{Methods}&\multirow{2}{*}{Case}&\multicolumn{2}{c}{$p=200,n_0=200$}& &\multicolumn{2}{c}{$p=500,n_0=200$}& &\multicolumn{2}{c}{$p=200,n_0=500$}& &\multicolumn{2}{c}{$p=500,n_0=500$}\cr
            & &ACP&AL& &ACP&AL& &ACP&AL& &ACP&AL\cr
			\hline
			Only-FARM&\multirow{2}{*}{(a)}&0.736&0.554& &0.800&0.586& &0.946&1.794& &0.838&0.376\cr
			Trans-FARM& &0.942&0.548& &0.952&0.577& &0.938&2.070& &0.938&0.374\cr
            Only-FARM&\multirow{2}{*}{(b)}&0.646&0.702& &0.680&0.749& &0.924&2.298& &0.788&0.481\cr
			Trans-FARM& &0.930&0.694& &0.958&0.737& &0.956&2.663& &0.950&0.479\cr
            Only-FARM&\multirow{2}{*}{(c)}&0.728&0.503& &0.776&0.535& &0.946&1.820& &0.806&0.338\cr
			Trans-FARM& &0.952&0.494& &0.976&0.519& &0.968&2.084& &0.956&0.335\cr
            Only-FARM&\multirow{2}{*}{(d)}&0.618&0.638& &0.696&0.686& &0.944&2.318& &0.700&0.433\cr
			Trans-FARM& &0.932&0.628& &0.944&0.665& &0.962&2.646& &0.936&0.431\cr
			\hline
		\end{tabular}
	}
	\label{tab1}
\end{table*}

In these experiments, Trans-FARM generally achieves coverage closer to the nominal level with comparable interval lengths. Only-FARM uses the same target sample for both the initial estimation and debiasing, inducing dependence between the two terms that determine the debiasing remainder. Sample splitting reduces this dependence at the cost of fewer observations per fit, while cross-fitting requires additional fits. Trans-FARM retains the full target sample for both steps and borrows information from informative sources to improve initial estimation. The resulting coverage gains are consistent with a smaller debiasing remainder. At $n_0=500$ and $p=200$, the confidence intervals become wider, which may be caused by the rank deficiency in $\hat{\boldsymbol{U}}_0$ that leads nodewise cross-validation to favor excessively small penalties when $p=200$, leading to unstable precision matrix estimates and potentially longer confidence intervals. This also explains why the studentized Trans-FARM intervals are wider than those of Only-FARM in this case, and we provide further discussion and results for the non-studentized simultaneous confidence intervals and provide more discussions in the Section \ref{secC.2}. This effect may be less pronounced when $p=500$, as exact training fits need not generalize to validation folds.

\subsection{Real Data Analysis}\label{sec5.3}
In this section, we apply our method to the FRED-MD dataset \citep{mccracken2016fred,fan2023farmdata}, which contains $126$ monthly U.S. macroeconomic time series capturing various aspects of economic activity and exhibiting strong correlations and latent factor structures. In our analysis, we consider `HOUSTNE' (housing starts in the northeast region) and `GS5' ($5$-year treasury rate) separately as responses, with the remaining variables as covariates. To mitigate nonstationarity associated with the $2008$ financial crisis, following \cite{fan2024latent}, we analyze two relatively stable subsamples: April 1992--December 2007 and August 2010--February 2020.

For each predictive target subsample, we implement a rolling-window forecasting strategy with a window size of 90 months. Each rolling window constitutes the target dataset for the subsequent observation prediction. The non-overlapping observations in the same subsample are used as the first source dataset, and the remaining subsample is treated as an additional source dataset, enabling information transfer across time periods. Prediction performance is then evaluated by aggregating results across all rolling windows. Specifically, for a given subsample and any $0<t\leq T-90$ (where $T$ is the sample size in the subsample), the target data is $\{(\boldsymbol{x}_{i},y_{i})\}_{i=t}^{t+89}$, the corresponding objective is to predict $y_{t+90}$, denoted by $\hat{y}_{t+90}$. We compare our proposed Trans-FARM method with Only-FARM \citep{fan2024latent}, Only-Lasso, and Trans-Lasso \citep[Algorithm 2]{tian2023transfer}, where we utilize the two aforementioned source datasets as auxiliary datasets in the transfer learning methods. We measure the prediction accuracy through the out-of-sample $R^2$:
$R^2=1-\sum_{i=91}^{T}(y_i-\hat{y}_i)^2/\sum_{i=91}^{T}(y_i-\bar{y}_i)^2,$ where $\bar{y}_i$ is the in-sample average $\sum_{j=i-90}^{i-1}y_j/90$. The out-of-sample $R^2$ results regarding two subsamples and two responses `HOUSTNE' and `GS5' are exhibited in Table \ref{tab2}.
\begin{table*}[htbp]
	\caption{Comparison of out-of-sample $R^2$ for four models in predicting two responses across two time periods. Bold entries indicate the highest $R^2$ in each setting.}
	\centering
    \resizebox{0.86\linewidth}{!}{
	\begin{tabular}{cccccc}
		\hline
        Time period&Response&Only-Lasso&Trans-Lasso&Only-FARM&Trans-FARM\cr
		\hline
		\multirow{2}{*}{04.1992-12.2007}&HOUSTNE&0.701&0.723&0.722&\textbf{0.734}\cr
        &GS5&0.687&0.722&0.726&\textbf{0.740}\cr
        \multirow{2}{*}{08.2010-02.2020}&HOUSTNE&0.390&0.810&0.770&\textbf{0.885}\cr
        &GS5&0.757&0.784&0.678&\textbf{0.818}\cr
		\hline
	\end{tabular}
    }
	\label{tab2}
\end{table*}
Table \ref{tab2} shows that Trans-FARM outperforms the other three methods under all predictive targets. These results demonstrate the effectiveness of Trans-FARM in exploiting auxiliary datasets to improve predictive accuracy. Notably, in contrast to approaches that handle structural breaks via explicit subsample partitioning \citep{fan2024latent}, Trans-FARM can still extract useful information despite the presence of pronounced disruptions associated with the financial crisis. This further illustrates the significance of our method.
\vspace{-0.3cm}

\section{Conclusion}\label{sec6}\vspace{-0.03cm}
In this paper, we study transfer learning for factor-augmented sparse linear models to accommodate strong covariate correlations and latent factor structures while leveraging auxiliary datasets to improve estimation accuracy. We establish non-asymptotic $\ell_1$- and $\ell_2$-error bounds for the resulting estimators and derive the minimax lower bounds. For a known informative source set, the estimator attains the minimax rates under suitable conditions. These results identify regimes in which informative auxiliary data improve estimation accuracy over using the target data alone. The lower bounds also reveal a limitation associated with the latent factor representation. Even when informative domains share the same regression coefficient vector, observational non-identifiability within the model class gives rise to dimension-dependent terms of orders $sb_{n,\min}^2/p^2$ and $sb_{n,\min}/p$ under squared $\ell_2$ and $\ell_1$ losses, respectively. Increasing the sample size alone cannot eliminate these terms without additional restrictions on the model class. To mitigate negative transfer, we develop a data-driven source detection method and establish its consistency. We also construct valid simultaneous confidence intervals using a studentized bootstrap procedure, yielding coefficient-specific interval lengths that reflect heterogeneous estimation uncertainty. Simulation experiments and a real-data study demonstrate improvements in estimation accuracy and inference reliability across heterogeneous environments. Future work could extend the proposed framework to nonlinear or generalized factor-augmented models and incorporate robust loss functions to improve stability under heavy-tailed errors and domain shifts.

\section*{Acknowledgments}
The authors are supported by NSFC under Grant No. 12571311.


\spacingset{0.87} 
{\small
\bibliography{reference.bib}}
\bibliographystyle{apalike}
\spacingset{1.5} 

\newpage
\begin{appendices}
\renewcommand{\thetheorem}{A\arabic{theorem}}
\setcounter{theorem}{0}

\section{Proof of Main Results}\label{secA}
This section contains the proofs of the main results. Auxiliary lemmas are deferred to Appendix \ref{secB}.

Denote $\hat{L}(\boldsymbol{w})=\frac{1}{2(n_{\mathcal{A}_\eta}+n_0)}\sum_{k\in\{0\}\cup\mathcal{A}_\eta}\|\tilde{\boldsymbol{Y}}^{(k)}-\hat{\boldsymbol{U}}_k\boldsymbol{w}\|_2^2,\,\nabla\hat{L}(\boldsymbol{w})=-\frac{1}{n_{\mathcal{A}_\eta}+n_0}\sum_{k\in\{0\}\cup\mathcal{A}_\eta}\\\hat{\boldsymbol{U}}_k^\top(\tilde{\boldsymbol{Y}}^{(k)}-\hat{\boldsymbol{U}}_k\boldsymbol{w}),\quad\delta\hat{L}(\boldsymbol{t})=\hat{L}(\boldsymbol{w}^{\mathcal{A}_\eta}+\boldsymbol{t})-\hat{L}(\boldsymbol{w}^{\mathcal{A}_\eta})-{\nabla\hat{L}(\boldsymbol{w}^{\mathcal{A}_\eta})}^\top\boldsymbol{t}$, and 
\begin{small}
\begin{align*}
D(\boldsymbol{\Delta})=\boldsymbol{\Delta}^\top\left(\frac{1}{n_{\mathcal{A}_\eta}+n_0}\sum_{k\in\{0\}\cup\mathcal{A}_\eta}\boldsymbol{U}_k^\top\boldsymbol{U}_k\right)\boldsymbol{\Delta},\quad\hat{D}(\boldsymbol{\Delta})=\boldsymbol{\Delta}^\top\left(\frac{1}{n_{\mathcal{A}_\eta}+n_0}\sum_{k\in\{0\}\cup\mathcal{A}_\eta}\hat{\boldsymbol{U}}_k^\top\hat{\boldsymbol{U}}_k\right)\boldsymbol{\Delta}.
\end{align*}
\end{small}
\subsection{Proof of Theorem \ref{thm1}}\label{secA1}
\begin{proof}
    Throughout this proof, we assume that the events in Lemma \ref{lemmaA6} hold with $\mathcal{L}=\{0\}\cup\mathcal{A}_\eta$ and $t=100\log p$, which occurs with probability at least $1-2(|\mathcal{A}_\eta|+1)e^{-cn_{\min}}-(|\mathcal{A}_\eta|+1)e^{-cp}-5p^{-100}$. Define $\zeta_{t,\mathcal{L}}=t+\log(2p|\mathcal{L}|)$ for some nonempty set $\mathcal{L}\subseteq\{0\}\cup[K]$ and $$\xi_{t,\mathcal{A}_\eta,k}=\sqrt{\left(\frac{n_k}{p}+\zeta_{t,\{0\}\cup\mathcal{A}_\eta}\right)\big(\log(e(|\mathcal{A}_\eta|+1))+t\big)},~\xi_{t,\{0\}}=\sqrt{\left(\frac{n_0}{p}+t+\log(2p)\right)(t+1)}.$$ Further define
    \begin{align*}
        \mathcal{V}_{n_k,p}(t,\mathcal{L})=\frac{n_k}{p}+\sqrt{\frac{\zeta_{t,\mathcal{L}}}{n_k}}+\sqrt{\frac{n_k\zeta_{t,\mathcal{L}}}{p}},
    \end{align*}
    we have $$\mathcal{V}_{n_k,p}:=\mathcal{V}_{n_k,p}(\log p,\{0\}\cup\mathcal{A}_\eta)\lesssim\frac{n_k}{p}+\sqrt{\frac{\log p}{n_k}}+\sqrt{\frac{n_k\log p}{p}}$$ when $|\mathcal{A}_\eta|+1\leq p^{C_K}$ for some constant $C_K>0$, and 
    $$\mathcal{V}_{n_0,p}:=\mathcal{V}_{n_0,p}(\log p,\{0\})\lesssim\frac{n_0}{p}+\sqrt{\frac{\log p}{n_0}}+\sqrt{\frac{n_0\log p}{p}}.$$

    Let $\hat{\boldsymbol{\mathcal{T}}}^{\mathcal{A}_\eta}=\hat{\boldsymbol{w}}^{\mathcal{A}_\eta}-\boldsymbol{w}^{\mathcal{A}_\eta}$ and $\hat{\boldsymbol{\Delta}}^{\mathcal{A}_\eta}=\hat{\boldsymbol{\beta}}^{\mathcal{A}_\eta}-\boldsymbol{\beta}$ be the estimation errors of the transferring step and eventual estimator respectively. We firstly consider the error bound for $\hat{\boldsymbol{\mathcal{T}}}^{\mathcal{A}_\eta}$ from which we are committed to derive the error bound for $\hat{\boldsymbol{\Delta}}^{\mathcal{A}_\eta}$.\\
    \textbf{Step 1: $\ell_2$ error bounds for $\hat{\boldsymbol{\mathcal{T}}}^{\mathcal{A}_\eta}$.} Let $\mathcal{S}=\mathrm{supp}(\boldsymbol{\beta})$. By the definition of $\hat{\boldsymbol{w}}^{\mathcal{A}_\eta}$, Lemma \ref{lemmaA1}, and the convexity of $\hat{L}$, under the event $\mathcal{E}_{\boldsymbol{w}}=\{\lambda_{\boldsymbol{w}}\geq 2\|\nabla \hat{L}(\boldsymbol{w}^{\mathcal{A}_\eta})\|_\infty\}$, it follows that
    \begin{equation}\label{eqA1}
        \begin{aligned}
            0&\leq\delta\hat{L}(\hat{\boldsymbol{\mathcal{T}}}^{\mathcal{A}_\eta})\leq \lambda_{\boldsymbol{w}}(\|\boldsymbol{w}_{\mathcal{S}}^{\mathcal{A}_\eta}\|_1+\|\boldsymbol{w}_{\mathcal{S}^c}^{\mathcal{A}_\eta}\|_1)-\lambda_{\boldsymbol{w}}(\|\hat{\boldsymbol{w}}_{\mathcal{S}}^{\mathcal{A}_\eta}\|_1+\|\hat{\boldsymbol{w}}_{\mathcal{S}^c}^{\mathcal{A}_\eta}\|_1)-{\nabla\hat{L}(\boldsymbol{w}^{\mathcal{A}_\eta})}^\top\hat{\boldsymbol{\mathcal{T}}}^{\mathcal{A}_\eta}\\
            &\leq \lambda_{\boldsymbol{w}}(\|\boldsymbol{w}_{\mathcal{S}}^{\mathcal{A}_\eta}\|_1+\|\boldsymbol{w}_{\mathcal{S}^c}^{\mathcal{A}_\eta}\|_1)-\lambda_{\boldsymbol{w}}(\|\hat{\boldsymbol{w}}_{\mathcal{S}}^{\mathcal{A}_\eta}\|_1+\|\hat{\boldsymbol{w}}_{\mathcal{S}^c}^{\mathcal{A}_\eta}\|_1)+\frac{1}{2}\lambda_{\boldsymbol{w}}\|\hat{\boldsymbol{\mathcal{T}}}^{\mathcal{A}_\eta}\|_1\\
            &\leq \frac{3}{2}\lambda_{\boldsymbol{w}}\|\hat{\boldsymbol{\mathcal{T}}}_{\mathcal{S}}^{\mathcal{A}_\eta}\|_1+\lambda_{\boldsymbol{w}}(\|\boldsymbol{w}_{\mathcal{S}^c}^{\mathcal{A}_\eta}\|_1-\|\hat{\boldsymbol{w}}_{\mathcal{S}^c}^{\mathcal{A}_\eta}\|_1)+\frac{1}{2}\lambda_{\boldsymbol{w}}\|\hat{\boldsymbol{\mathcal{T}}}_{\mathcal{S}^c}^{\mathcal{A}_\eta}\|_1\\
            &\leq \frac{3}{2}\lambda_{\boldsymbol{w}}\|\hat{\boldsymbol{\mathcal{T}}}_{\mathcal{S}}^{\mathcal{A}_\eta}\|_1-\frac{1}{2}\lambda_{\boldsymbol{w}}\|\hat{\boldsymbol{\mathcal{T}}}_{\mathcal{S}^c}^{\mathcal{A}_\eta}\|_1+2\lambda_{\boldsymbol{w}}\|\boldsymbol{w}_{\mathcal{S}^c}^{\mathcal{A}_\eta}\|_1\leq \frac{3}{2}\lambda_{\boldsymbol{w}}\|\hat{\boldsymbol{\mathcal{T}}}_{\mathcal{S}}^{\mathcal{A}_\eta}\|_1-\frac{1}{2}\lambda_{\boldsymbol{w}}\|\hat{\boldsymbol{\mathcal{T}}}_{\mathcal{S}^c}^{\mathcal{A}_\eta}\|_1+2\lambda_{\boldsymbol{w}}C_1\eta\\
        \end{aligned}
    \end{equation}
It can be implied from \eqref{eqA1} that $\|\hat{\boldsymbol{\mathcal{T}}}_{\mathcal{S}^c}^{\mathcal{A}_\eta}\|_1\leq3\|\hat{\boldsymbol{\mathcal{T}}}_{\mathcal{S}}^{\mathcal{A}_\eta}\|_1+4C_1\eta$. By Lemma \ref{lemmaA2} and Lemma \ref{lemmaA6}, it is straightforward that
\begin{align*}
    &\hat{D}(\hat{\boldsymbol{\mathcal{T}}}^{\mathcal{A}_\eta})\geq D(\hat{\boldsymbol{\mathcal{T}}}^{\mathcal{A}_\eta})-|\hat{D}(\hat{\boldsymbol{\mathcal{T}}}^{\mathcal{A}_\eta})-D(\hat{\boldsymbol{\mathcal{T}}}^{\mathcal{A}_\eta})|\\
    \geq& \kappa_1\|\hat{\boldsymbol{\mathcal{T}}}^{\mathcal{A}_\eta}\|_2^2-\kappa_2\frac{\log p}{n_{\mathcal{A}_\eta}+n_0}\|\hat{\boldsymbol{\mathcal{T}}}^{\mathcal{A}_\eta}\|_1^2-\|\hat{\boldsymbol{\mathcal{T}}}^{\mathcal{A}_\eta}\|_1^2\frac{1}{n_{\mathcal{A}_\eta}+n_0}\left\|\sum_{k\in\{0\}\cup\mathcal{A}_\eta}(\boldsymbol{U}_k^\top\boldsymbol{U}_k-\hat{\boldsymbol{U}}_k^\top\hat{\boldsymbol{U}}_k)\right\|_{\max}\\
    \geq&\kappa_1\|\hat{\boldsymbol{\mathcal{T}}}^{\mathcal{A}_\eta}\|_2^2-\tau_w\|\hat{\boldsymbol{\mathcal{T}}}^{\mathcal{A}_\eta}\|_1^2,\quad\text{with} \quad \tau_w\leq C\left(\frac{1}{p}+\frac{(|\mathcal{A}_\eta|+1)\zeta_{t,\{0\}\cup\mathcal{A}_\eta}}{n_{\mathcal{A}_\eta}+n_0}\right).
\end{align*}
Since $\delta\hat{L}(\hat{\boldsymbol{\mathcal{T}}}^{\mathcal{A}_\eta})=\hat{D}(\hat{\boldsymbol{\mathcal{T}}}^{\mathcal{A}_\eta})/2$, we obtain that 
\begin{align*}
    \frac{\kappa_1\|\hat{\boldsymbol{\mathcal{T}}}^{\mathcal{A}_\eta}\|_2^2}{2}\leq& \frac{\tau_w}{2}(4\|\hat{\boldsymbol{\mathcal{T}}}_{\mathcal{S}}^{\mathcal{A}_\eta}\|_1+4C_1\eta)^2+\frac{3}{2}\lambda_{\boldsymbol{w}}\|\hat{\boldsymbol{\mathcal{T}}}_{\mathcal{S}}^{\mathcal{A}_\eta}\|_1-\frac{1}{2}\lambda_{\boldsymbol{w}}\|\hat{\boldsymbol{\mathcal{T}}}_{\mathcal{S}^c}^{\mathcal{A}_\eta}\|_1+2\lambda_{\boldsymbol{w}}C_1\eta\\
    \leq&\frac{\tau_w}{2}(4\sqrt{s}\|\hat{\boldsymbol{\mathcal{T}}}_{\mathcal{S}}^{\mathcal{A}_\eta}\|_2+4C_1\eta)^2+\frac{3}{2}\lambda_{\boldsymbol{w}}\sqrt{s}\|\hat{\boldsymbol{\mathcal{T}}}_{\mathcal{S}}^{\mathcal{A}_\eta}\|_2+2\lambda_{\boldsymbol{w}}C_1\eta\\
    \leq&\tau_w(16s\|\hat{\boldsymbol{\mathcal{T}}}_{\mathcal{S}}^{\mathcal{A}_\eta}\|_2^2+16C_1^2\eta^2)+\frac{\kappa_1}{8}\|\hat{\boldsymbol{\mathcal{T}}}_{\mathcal{S}}^{\mathcal{A}_\eta}\|_2^2+\frac{18}{\kappa_1}s\lambda_{\boldsymbol{w}}^2+2\lambda_{\boldsymbol{w}}C_1\eta,
\end{align*}
which implies that 
\begin{align*}
    \left(\frac{3\kappa_1}{8}-16s\tau_w\right)\|\hat{\boldsymbol{\mathcal{T}}}^{\mathcal{A}_\eta}\|_2^2\leq 16\tau_w C_1^2\eta^2+\frac{18}{\kappa_1}s\lambda_{\boldsymbol{w}}^2+2\lambda_{\boldsymbol{w}}C_1\eta.
\end{align*}
As long as $16s\tau_w\leq \kappa_1/8$, we have $$\|\hat{\boldsymbol{\mathcal{T}}}^{\mathcal{A}_\eta}\|_2^2\leq \frac{64}{\kappa_1}\tau_w C_1^2\eta^2+\frac{72}{\kappa_1^2}s\lambda_{\boldsymbol{w}}^2+\frac{8}{\kappa_1}\lambda_{\boldsymbol{w}}C_1\eta.$$

By choosing $\lambda_{\boldsymbol{w}}=C_\omega(1\vee\eta)\sqrt{\log p/(n_{\mathcal{A}_\eta}+n_0)}$, then we obtain 
\begin{align*}
\|\hat{\boldsymbol{\mathcal{T}}}^{\mathcal{A}_\eta}\|_2^2&\lesssim s\lambda_{\boldsymbol{w}}^2+C_1\eta\lambda_{\boldsymbol{w}}+\tau_w C_1^2\eta^2,\\
\|\hat{\boldsymbol{\mathcal{T}}}^{\mathcal{A}_\eta}\|_1&\leq4\sqrt{s}\|\hat{\boldsymbol{\mathcal{T}}}_{\mathcal{S}}^{\mathcal{A}_\eta}\|_2+4C_1\eta
\end{align*}
It remains to show that $\mathbb{P}(\mathcal{E}_{\boldsymbol{w}})\to 1$. Lemma \ref{lemmaA4} ensures that $\mathcal{E}_{\boldsymbol{w}}$ occurs with probability at least $1-2p^{-100}$ if $C_\omega$ is sufficiently large, and 
    \begin{align*}
        \frac{\sum_{k\in\{0\}\cup\mathcal{A}_\eta}\big(\mathcal{V}_{n_k,p}(\log p,\{0\}\cup\mathcal{A}_\eta)\|\varphi^{(k)}\|_2+\xi_{\log p,\mathcal{A}_\eta,k}+\eta\zeta_{\log p,\{0\}\cup\mathcal{A}_\eta}\big)}{(1\vee\eta)\sqrt{(n_{\mathcal{A}_\eta}+n_0)\log p}}+\frac{\eta\sqrt{n_{\mathcal{A}_\eta}+n_0}}{(1\vee\eta)p\sqrt{\log p}}\leq C
    \end{align*}
    for some constant $C>0$. Additionally, note that we require $$s\left(\frac{(|\mathcal{A}_\eta|+1)\zeta_{\log p,\{0\}\cup\mathcal{A}_\eta}}{n_{\mathcal{A}_\eta}+n_0}+\frac{1}{p}\right)\leq c$$ for some sufficiently small constant $c>0$ to ensure $16s\tau_w\leq \kappa_1/8$.

\textbf{Step 2: $\ell_2$ error bounds for $\hat{\boldsymbol{\Delta}}^{\mathcal{A}_\eta}$.}
Define 
\vspace{-0.3mm}
\begin{align*}
\hat{L}^{(0)}(\boldsymbol{w})&=\frac{1}{2n_0}\|\tilde{\boldsymbol{Y}}^{(0)}-\hat{\boldsymbol{U}}_0\boldsymbol{w}\|_2^2,\quad\nabla\hat{L}^{(0)}(\boldsymbol{w})=-\frac{1}{n_0}\hat{\boldsymbol{U}}_0^\top(\tilde{\boldsymbol{Y}}^{(0)}-\hat{\boldsymbol{U}}_0\boldsymbol{w}),\\
\boldsymbol{\delta}^{\mathcal{A}_\eta}&=\boldsymbol{\beta}-\boldsymbol{w}^{\mathcal{A}_\eta},\quad\hat{\boldsymbol{\beta}}^{\mathcal{A}_\eta}=\hat{\boldsymbol{w}}^{\mathcal{A}_\eta}+\hat{\boldsymbol{\delta}}^{\mathcal{A}_\eta},\quad\hat{\boldsymbol{\Delta}}^{\mathcal{A}_\eta}=\hat{\boldsymbol{\beta}}^{\mathcal{A}_\eta}-\boldsymbol{\beta},\\
    \hat{\boldsymbol{d}}^{\mathcal{A}_\eta}&=\hat{\boldsymbol{\delta}}^{\mathcal{A}_\eta}-\boldsymbol{\delta}^{\mathcal{A}_\eta},\quad\delta\hat{L}^{(0)}(\boldsymbol{\Delta})=\hat{L}^{(0)}(\boldsymbol{\beta}+\boldsymbol{\Delta})-\hat{L}^{(0)}(\boldsymbol{\beta})-{\nabla\hat{L}^{(0)}(\boldsymbol{\beta})}^\top\boldsymbol{\Delta},\\
    D^{(0)}(\boldsymbol{\Delta})&=\frac{1}{n_0}\boldsymbol{\Delta}^\top\boldsymbol{U}_0^\top\boldsymbol{U}_0\boldsymbol{\Delta},\quad \hat{D}^{(0)}(\boldsymbol{\Delta})=\frac{1}{n_0}\boldsymbol{\Delta}^\top\hat{\boldsymbol{U}}_0^\top\hat{\boldsymbol{U}}_0\boldsymbol{\Delta}
\end{align*}
Notice that $\hat{\boldsymbol{w}}^{\mathcal{A}_\eta}+\boldsymbol{\delta}^{\mathcal{A}_\eta}=\boldsymbol{\beta}+\hat{\boldsymbol{\mathcal{T}}}^{\mathcal{A}_\eta}$ and $\hat{\boldsymbol{\beta}}^{\mathcal{A}_\eta}=\hat{\boldsymbol{w}}^{\mathcal{A}_\eta}+\hat{\boldsymbol{\delta}}^{\mathcal{A}_\eta}=\boldsymbol{\beta}+\hat{\boldsymbol{\mathcal{T}}}^{\mathcal{A}_\eta}+\hat{\boldsymbol{d}}^{\mathcal{A}_\eta}$. Similar to \eqref{eqA1}, under the event $\mathcal{E}_{\boldsymbol{\delta}}=\{\lambda_{\boldsymbol{\delta}}\geq 2\|\nabla \hat{L}^{(0)}(\boldsymbol{\beta})\|_\infty\}$, we have
\begin{equation}
        \begin{aligned}
            0\leq& \frac{1}{2}\hat{D}^{(0)}(\hat{\boldsymbol{d}}^{\mathcal{A}_\eta})=\hat{L}^{(0)}(\hat{\boldsymbol{\beta}}^{\mathcal{A}_\eta})-\hat{L}^{(0)}(\hat{\boldsymbol{w}}^{\mathcal{A}_\eta}+\boldsymbol{\delta}^{\mathcal{A}_\eta})-{\nabla\hat{L}^{(0)}(\boldsymbol{\beta})}^\top\hat{\boldsymbol{d}}^{\mathcal{A}_\eta}-\frac{1}{n_0}\hat{\boldsymbol{\mathcal{T}}}^{\mathcal{A}_\eta\top}\hat{\boldsymbol{U}}_0^\top\hat{\boldsymbol{U}}_0\hat{\boldsymbol{d}}^{\mathcal{A}_\eta}\\
            \leq&\lambda_{\boldsymbol{\delta}}\|\boldsymbol{\delta}^{\mathcal{A}_\eta}\|_1-\lambda_{\boldsymbol{\delta}}\|\boldsymbol{\delta}^{\mathcal{A}_\eta}+\hat{\boldsymbol{d}}^{\mathcal{A}_\eta}\|_1+\frac{1}{2}\lambda_{\boldsymbol{\delta}}\|\hat{\boldsymbol{d}}^{\mathcal{A}_\eta}\|_1+\left(\hat{D}^{(0)}(\hat{\boldsymbol{\mathcal{T}}}^{\mathcal{A}_\eta})\hat{D}^{(0)}(\hat{\boldsymbol{d}}^{\mathcal{A}_\eta})\right)^{1/2}\\
            \leq& \lambda_{\boldsymbol{\delta}}(2\|\boldsymbol{\delta}^{\mathcal{A}_\eta}\|_1-\|\hat{\boldsymbol{d}}^{\mathcal{A}_\eta}\|_1)+\frac{1}{2}\lambda_{\boldsymbol{\delta}}\|\hat{\boldsymbol{d}}^{\mathcal{A}_\eta}\|_1+\hat{D}^{(0)}(\hat{\boldsymbol{\mathcal{T}}}^{\mathcal{A}_\eta})+\frac{1}{4}\hat{D}^{(0)}(\hat{\boldsymbol{d}}^{\mathcal{A}_\eta})\\
            \leq& 2\lambda_{\boldsymbol{\delta}}C_1\eta-\frac{1}{2}\lambda_{\boldsymbol{\delta}}\|\hat{\boldsymbol{d}}^{\mathcal{A}_\eta}\|_1+\hat{D}^{(0)}(\hat{\boldsymbol{\mathcal{T}}}^{\mathcal{A}_\eta})+\frac{1}{4}\hat{D}^{(0)}(\hat{\boldsymbol{d}}^{\mathcal{A}_\eta})
        \end{aligned}
    \end{equation}
where the second inequality follows from the Cauchy-Schwarz inequality, the optimality $\hat{\boldsymbol{\beta}}^{\mathcal{A}_\eta}$, and the event $\mathcal{E}_{\boldsymbol{\delta}}$, the third inequality follows from the triangle inequality and Young's inequality. This implies that 
$$0\leq \hat{D}^{(0)}(\hat{\boldsymbol{d}}^{\mathcal{A}_\eta})\leq 8\lambda_{\boldsymbol{\delta}}C_1\eta+4\hat{D}^{(0)}(\hat{\boldsymbol{\mathcal{T}}}^{\mathcal{A}_\eta})-2\lambda_{\boldsymbol{\delta}}\|\hat{\boldsymbol{d}}^{\mathcal{A}_\eta}\|_1\leq 8\lambda_{\boldsymbol{\delta}}C_1\eta+4\hat{D}^{(0)}(\hat{\boldsymbol{\mathcal{T}}}^{\mathcal{A}_\eta}),$$ and $$\|\hat{\boldsymbol{d}}^{\mathcal{A}_\eta}\|_1\leq 4C_1\eta+\frac{2}{\lambda_{\boldsymbol{\delta}}}\hat{D}^{(0)}(\hat{\boldsymbol{\mathcal{T}}}^{\mathcal{A}_\eta}).$$
Additionally, according to Lemmas \ref{lemmaA3} and \ref{lemmaA6}, we have similar results for $\hat{D}^{(0)}(\hat{\boldsymbol{d}}^{\mathcal{A}_\eta})$:
\begin{align*}
    &\hat{D}^{(0)}(\hat{\boldsymbol{d}}^{\mathcal{A}_\eta})\geq D^{(0)}(\hat{\boldsymbol{d}}^{\mathcal{A}_\eta})-|\hat{D}^{(0)}(\hat{\boldsymbol{d}}^{\mathcal{A}_\eta})-D^{(0)}(\hat{\boldsymbol{d}}^{\mathcal{A}_\eta})|\\
    \geq& \kappa_1\|\hat{\boldsymbol{d}}^{\mathcal{A}_\eta}\|_2^2-\kappa_2\frac{\log p}{n_0}\|\hat{\boldsymbol{d}}^{\mathcal{A}_\eta}\|_1^2-\|\hat{\boldsymbol{d}}^{\mathcal{A}_\eta}\|_1^2\frac{1}{n_0}\left\|\boldsymbol{U}_0^\top\boldsymbol{U}_0-\hat{\boldsymbol{U}}_0^\top\hat{\boldsymbol{U}}_0\right\|_{\max}\\
    \geq&\kappa_1\|\hat{\boldsymbol{d}}^{\mathcal{A}_\eta}\|_2^2-\tau_\delta\|\hat{\boldsymbol{d}}^{\mathcal{A}_\eta}\|_1^2,\quad\text{with} \quad \tau_\delta\leq C\left(\frac{1}{p}+\frac{\zeta_{t,\{0\}}}{n_0}\right).
\end{align*}
Therefore, we obtain that 
\begin{align*}
    \kappa_1\|\hat{\boldsymbol{d}}^{\mathcal{A}_\eta}\|_2^2\leq& \hat{D}^{(0)}(\hat{\boldsymbol{d}}^{\mathcal{A}_\eta})+\tau_\delta\|\hat{\boldsymbol{d}}^{\mathcal{A}_\eta}\|_1^2\leq 8\lambda_{\boldsymbol{\delta}}C_1\eta+4\hat{D}^{(0)}(\hat{\boldsymbol{\mathcal{T}}}^{\mathcal{A}_\eta})+\tau_\delta\left(4C_1\eta+\frac{2}{\lambda_{\boldsymbol{\delta}}}\hat{D}^{(0)}(\hat{\boldsymbol{\mathcal{T}}}^{\mathcal{A}_\eta})\right)^2\\
    \leq& 8\lambda_{\boldsymbol{\delta}}C_1\eta+4\hat{D}^{(0)}(\hat{\boldsymbol{\mathcal{T}}}^{\mathcal{A}_\eta})+32\tau_\delta C_1^2\eta^2+8\tau_\delta\frac{\big(\hat{D}^{(0)}(\hat{\boldsymbol{\mathcal{T}}}^{\mathcal{A}_\eta})\big)^2}{\lambda_{\boldsymbol{\delta}}^2}
\end{align*}
Consider the following event
$$\mathcal{E}_{\delta'}:=\left\{\frac{\tau_\delta}{\lambda_{\boldsymbol{\delta}}^2}(\hat{D}^{(0)}(\hat{\boldsymbol{\mathcal{T}}}^{\mathcal{A}_\eta})+C_1\eta\lambda_{\boldsymbol{\delta}})\leq 1\right\}.$$
Under this event, we have the simplified result:
$$\|\hat{\boldsymbol{d}}^{\mathcal{A}_\eta}\|_2^2\leq \frac{40}{\kappa_1}\lambda_{\boldsymbol{\delta}}C_1\eta+\frac{12}{\kappa_1}\hat{D}^{(0)}(\hat{\boldsymbol{\mathcal{T}}}^{\mathcal{A}_\eta})\leq C_0(\lambda_{\boldsymbol{\delta}}C_1\eta+\hat{D}^{(0)}(\hat{\boldsymbol{\mathcal{T}}}^{\mathcal{A}_\eta})),$$ where $C_0=40/\kappa_1$. If $\lambda_{\boldsymbol{\delta}}C_1\eta\leq \hat{D}^{(0)}(\hat{\boldsymbol{\mathcal{T}}}^{\mathcal{A}_\eta})$, we have $\|\hat{\boldsymbol{d}}^{\mathcal{A}_\eta}\|_2^2\leq 2C_0\hat{D}^{(0)}(\hat{\boldsymbol{\mathcal{T}}}^{\mathcal{A}_\eta})\leq 2C_0(\hat{D}^{(0)}(\hat{\boldsymbol{\mathcal{T}}}^{\mathcal{A}_\eta})+C_1^2\eta^2)$. If $\lambda_{\boldsymbol{\delta}}C_1\eta> \hat{D}^{(0)}(\hat{\boldsymbol{\mathcal{T}}}^{\mathcal{A}_\eta})$, we have $$\|\hat{\boldsymbol{d}}^{\mathcal{A}_\eta}\|_1\leq 4C_1\eta+\frac{2}{\lambda_{\boldsymbol{\delta}}}\hat{D}^{(0)}(\hat{\boldsymbol{\mathcal{T}}}^{\mathcal{A}_\eta})\leq 6C_1\eta,$$ which implies that $\|\hat{\boldsymbol{d}}^{\mathcal{A}_\eta}\|_2^2\leq \|\hat{\boldsymbol{d}}^{\mathcal{A}_\eta}\|_1^2\leq 36C_1^2\eta^2$. Hence, let $C_0'=\max\{2C_0,36\}$, we obtain that $\|\hat{\boldsymbol{d}}^{\mathcal{A}_\eta}\|_2^2\leq C_0' (\hat{D}^{(0)}(\hat{\boldsymbol{\mathcal{T}}}^{\mathcal{A}_\eta})+C_1^2\eta^2)$. Combining the two bounds for $\|\hat{\boldsymbol{d}}^{\mathcal{A}_\eta}\|_2^2$ yields that $$\|\hat{\boldsymbol{d}}^{\mathcal{A}_\eta}\|_2^2\leq C_0' (\hat{D}^{(0)}(\hat{\boldsymbol{\mathcal{T}}}^{\mathcal{A}_\eta})+C_1\lambda_{\boldsymbol{\delta}}\eta\wedge C_1^2\eta^2).$$

Now, for $\hat{D}^{(0)}(\hat{\boldsymbol{\mathcal{T}}}^{\mathcal{A}_\eta})$, by Lemmas \ref{lemmaA3-upper bound} and \ref{lemmaA6}, we have 
\begin{align*}
    \hat{D}^{(0)}(\hat{\boldsymbol{\mathcal{T}}}^{\mathcal{A}_\eta})\leq& |\hat{D}^{(0)}(\hat{\boldsymbol{\mathcal{T}}}^{\mathcal{A}_\eta})-D^{(0)}(\hat{\boldsymbol{\mathcal{T}}}^{\mathcal{A}_\eta})|+D^{(0)}(\hat{\boldsymbol{\mathcal{T}}}^{\mathcal{A}_\eta})\\
    \leq& \|\hat{\boldsymbol{\mathcal{T}}}^{\mathcal{A}_\eta}\|_1^2\frac{1}{n_0}\left\|\boldsymbol{U}_0^\top\boldsymbol{U}_0-\hat{\boldsymbol{U}}_0^\top\hat{\boldsymbol{U}}_0\right\|_{\max}+D^{(0)}(\hat{\boldsymbol{\mathcal{T}}}^{\mathcal{A}_\eta})\\
    \leq&C\|\hat{\boldsymbol{\mathcal{T}}}^{\mathcal{A}_\eta}\|_2^2+\tau_\delta\|\hat{\boldsymbol{\mathcal{T}}}^{\mathcal{A}_\eta}\|_1^2,
\end{align*}
This implies that 
\begin{align*}
    \hat{D}^{(0)}(\hat{\boldsymbol{\mathcal{T}}}^{\mathcal{A}_\eta})\leq& C\|\hat{\boldsymbol{\mathcal{T}}}^{\mathcal{A}_\eta}\|_2^2+\tau_\delta\|\hat{\boldsymbol{\mathcal{T}}}^{\mathcal{A}_\eta}\|_1^2\leq C\|\hat{\boldsymbol{\mathcal{T}}}^{\mathcal{A}_\eta}\|_2^2+\tau_\delta\left(4\sqrt{s}\|\hat{\boldsymbol{\mathcal{T}}}_{\mathcal{S}}^{\mathcal{A}_\eta}\|_2+4C_1\eta\right)^2\\
    \leq&(C+32s\tau_{\delta})\|\hat{\boldsymbol{\mathcal{T}}}^{\mathcal{A}_\eta}\|_2^2+32C_1^2\eta^2\tau_{\delta}.
\end{align*}
Therefore, we have 
\begin{align*}
    \|\hat{\boldsymbol{\Delta}}^{\mathcal{A}_\eta}\|_2^2\leq& 2\|\hat{\boldsymbol{\mathcal{T}}}^{\mathcal{A}_\eta}\|_2^2+2\|\hat{\boldsymbol{d}}^{\mathcal{A}_\eta}\|_2^2\\
    \lesssim& (1\vee s\tau_{\delta})(s\lambda_{\boldsymbol{w}}^2+C_1\eta\lambda_{\boldsymbol{w}}+\tau_w C_1^2\eta^2)+C_1^2\eta^2\tau_{\delta}+(C_1\lambda_{\boldsymbol{\delta}}\eta\wedge C_1^2\eta^2)
\end{align*}
We now show that $\mathbb{P}(\mathcal{E}_{\boldsymbol{w}}\cap\mathcal{E}_{\boldsymbol{\delta}}\cap\mathcal{E}_{\boldsymbol{\delta}}')\to 1$. The event $\mathcal{E}_{\delta'}$ holds if
\begin{align*}
    \frac{\tau_\delta}{\lambda_{\boldsymbol{\delta}}^2}\Bigg\{C\Big[(1\vee s\tau_\delta)(s\lambda_{\boldsymbol{w}}^2+C_1\eta\lambda_{\boldsymbol{w}}+\tau_w C_1^2\eta^2)+C_1^2\eta^2\tau_\delta\Big]+C_1\eta\lambda_{\boldsymbol{\delta}}\Bigg\}\leq1.
\end{align*}
Choosing $\lambda_{\boldsymbol{\delta}}=C_{\boldsymbol{\delta}}\sqrt{\log p/n_0}$ for some sufficiently large constant $C_{\boldsymbol{\delta}}$ and assuming that $$\frac{\mathcal{V}_{n_0,p}(\log p,\{0\})\|\varphi^{(0)}\|_2+\xi_{\log p,\{0\}}}{\sqrt{n_0\log p}}\leq C$$ for some constant $C>0$, we obtain that $\mathbb{P}(\mathcal{E}_{\boldsymbol{\delta}})\geq1-p^{-100}$. Under the above conditions, we have $\mathbb{P}(\mathcal{E}_{\boldsymbol{w}}\cap\mathcal{E}_{\boldsymbol{\delta}}\cap\mathcal{E}_{\boldsymbol{\delta}}')\geq 1-3p^{-100}$. Then we have
\begin{align*}
    \|\hat{\boldsymbol{\Delta}}^{\mathcal{A}_\eta}\|_2^2\leq C\left((1\vee s\tau_{\delta})(s\lambda_{\boldsymbol{w}}^2+C_1\eta\lambda_{\boldsymbol{w}}+\tau_w C_1^2\eta^2)+C_1^2\eta^2\tau_{\delta}+(C_1\lambda_{\boldsymbol{\delta}}\eta\wedge C_1^2\eta^2)\right).
\end{align*}
Assume that 
$\tau_{w}=O(1)$, $\lambda_{\boldsymbol{\delta}}=O(1)$, $\eta\tau_w\lesssim\lambda_{\boldsymbol{w}}\lesssim\lambda_{\boldsymbol{\delta}}$, $\tau_\delta\lesssim \lambda_{\boldsymbol{\delta}}^2$, and $C_1\eta\lambda_{\boldsymbol{\delta}}\lesssim1$. By Young's inequality, we have
$C_1\eta\lambda_{\boldsymbol{w}}\leq s\lambda_{\boldsymbol{w}}^2+(C_1^2\eta^2/4\wedge C_1\eta\lambda_{\boldsymbol{w}})$, $C_1^2\eta^2\tau_{w}\lesssim C_1^2\eta^2\wedge C_1\lambda_{\boldsymbol{\delta}}\eta$, and $C_1^2\eta^2\tau_{\delta}\lesssim C_1^2\eta^2\wedge C_1\lambda_{\boldsymbol{\delta}}\eta$, which implies that
\begin{align*}
    \|\hat{\boldsymbol{\Delta}}^{\mathcal{A}_\eta}\|_2^2\leq C\left((1\vee s\tau_{\delta})\big(s\lambda_{\boldsymbol{w}}^2+(C_1^2\eta^2\wedge C_1\eta\lambda_{\boldsymbol{\delta}})\big)\right).
\end{align*}
All the additional conditions above can be summarized as $|\mathcal{A}_\eta|+1\leq p^{C_K}$ for some constant $C_K>0$, 
$$\frac{\log p}{n_0}\vee\frac{n_0}{p\log p}\vee(1\vee\eta)\sqrt{\frac{n_0}{n_{\mathcal{A}_\eta}+n_0}}=O(1),~\frac{\mathcal{V}_{n_0,p}\|\varphi^{(0)}\|_2+\log p}{\sqrt{n_0\log p}}=O(1)$$
    \begin{align*}
        \frac{\sum_{k\in\{0\}\cup\mathcal{A}_\eta}\big(\mathcal{V}_{n_k,p}\|\varphi^{(k)}\|_2+\sqrt{n_k\log p/p}+(1\vee\eta)\log p\big)}{(1\vee\eta)\sqrt{(n_{\mathcal{A}_\eta}+n_0)\log p}}+\frac{\eta\sqrt{n_{\mathcal{A}_\eta}+n_0}}{(1\vee\eta)p\sqrt{\log p}}=O(1),
    \end{align*}
    $$\left(1+\frac{s\log p}{n_0}\right)\left(\frac{s(1\vee\eta)^2\log p}{n_{\mathcal{A}_\eta}+n_0}+\left(\eta^2\wedge\eta\sqrt{\frac{\log p}{n_0}}\right)\right)+\eta\sqrt{\frac{\log p}{n_0}}=O(1),$$ and 
    $$s\left(\frac{(|\mathcal{A}_\eta|+1)\log p}{n_{\mathcal{A}_\eta}+n_0}+\frac{1}{p}\right)\leq c$$
    for some sufficiently small constant $c>0$. Further simplifying the conditions completes the proof of $\ell_2$ error bound.

\textbf{Step 3: $\ell_1$ error Bound for $\hat{\boldsymbol{\Delta}}^{\mathcal{A}_\eta}$.}
Similar to step 2, by the optimality of $\hat{\boldsymbol{\delta}}^{\mathcal{A}_\eta}$, under the event $\mathcal{E}_{\boldsymbol{\delta}}$, we obtain that 
\begin{equation}
        \begin{aligned}
            0\leq& \frac{1}{2}\hat{D}^{(0)}(\hat{\boldsymbol{\Delta}}^{\mathcal{A}_\eta})=\delta\hat{L}^{(0)}(\hat{\boldsymbol{\Delta}}^{\mathcal{A}_\eta})\leq \lambda_{\boldsymbol{\delta}}\|\boldsymbol{\beta}-\hat{\boldsymbol{w}}^{\mathcal{A}_\eta}\|_1-\lambda_{\boldsymbol{\delta}}\|\hat{\boldsymbol{\delta}}^{\mathcal{A}_\eta}\|_1+\frac{1}{2}\lambda_{\boldsymbol{\delta}}\|\hat{\boldsymbol{\Delta}}^{\mathcal{A}_\eta}\|_1\\
            \leq&2\lambda_{\boldsymbol{\delta}}\|\boldsymbol{\beta}-\hat{\boldsymbol{w}}^{\mathcal{A}_\eta}\|_1-\frac{1}{2}\lambda_{\boldsymbol{\delta}}\|\hat{\boldsymbol{\Delta}}^{\mathcal{A}_\eta}\|_1
        \end{aligned}
    \end{equation}
    where the second step follows from $\hat{\boldsymbol{\delta}}^{\mathcal{A}_\eta}=\hat{\boldsymbol{\Delta}}^{\mathcal{A}_\eta}+\boldsymbol{\beta}-\hat{\boldsymbol{w}}^{\mathcal{A}_\eta}$. This implies that 
    \begin{align*}
        \|\hat{\boldsymbol{\Delta}}^{\mathcal{A}_\eta}\|_1\leq& 4\|\boldsymbol{\beta}-\hat{\boldsymbol{w}}^{\mathcal{A}_\eta}\|_1\leq 4\|\hat{\boldsymbol{\mathcal{T}}}^{\mathcal{A}_\eta}\|_1+4C_1\eta\leq C(s\lambda_{\boldsymbol{w}}+\sqrt{s\eta\lambda_{\boldsymbol{w}}}+\sqrt{s\tau_w}\eta+\eta)\\
        \leq&C'\left(s(1\vee\eta)\sqrt{\frac{\log p}{n_{\mathcal{A}_\eta}+n_0}}+\eta\right).
    \end{align*}
    This completes the proof.

\end{proof}

\subsection{Proof of Theorem \ref{thm-minimax rate}}\label{secA2}

\begin{proof}
    We establish lower bounds for the three terms separately by restricting attention to suitable sub-models of the parameter space. We first provide the proof of $\ell_2$ minimax lower bound, followed by the $\ell_1$ minimax lower bound. We restrict attention to the dataset indexed by $\{0\}\cup{\mathcal{A}}_\eta$ for a prespecified $\eta$. 
    
    For any probability measures $\mathbb{P}$ and $\mathbb{Q}$, define the Kullback-Leibler (KL) divergence by $D_{KL}(\mathbb{P}||\mathbb{Q})=\int \log(\mathrm{d}\mathbb{P}/\mathrm{d}\mathbb{Q})\mathrm{d}\mathbb{P}$. Let $p_0=\lfloor p/2\rfloor$ and $b_{n,\min}=\min_{k\in\{0\}\cup\mathcal{A}_\eta}b_{n,k}$. We assume that $1\leq s\leq p_0/4$ without loss of generality, as the case $s\geq p_0/4$ is harder such that $\Theta_{\boldsymbol{b}_n}(s,\eta) \subseteq \Theta_{\boldsymbol{b}_n}(s',\eta)$ with $s\leq s'$.

    \textbf{$\ell_2$ bound: Step 1.} There exists some constant $c\in[1,\infty)$ such that$$\left(\eta^2\wedge\eta\sqrt{\frac{\log(p_0/s)}{n_0}}\wedge \frac{s\log(p_0/s)}{n_0}\right)\vee\frac{sb_{n,\min}^2}{p^2}\leq c\frac{s\log(p_0/s)}{n_{\mathcal{A}_\eta}+n_0}.$$
    
    In this case, we provide the fundamental limits caused by the sample size. We consider the following model: $$\boldsymbol{\Sigma}_{\boldsymbol{u}}^{(k)}=\boldsymbol{I}_p,~\boldsymbol{\gamma}^{(k)}=\boldsymbol{0}_{r_k},~\boldsymbol{B}_k=\begin{bmatrix}
        \boldsymbol{0}_{p_0\times r_k}\\
        \boldsymbol{H}_k
    \end{bmatrix},$$
    where $\boldsymbol{H}_k\in\mathbb{R}^{(p-p_0)\times r_k}$ is fixed and satisfies Assumption \ref{assum1} (e.g., 
    \begin{align*}
        \boldsymbol{H}_k=\begin{bmatrix}
            \boldsymbol{I}_{r_k}\otimes \boldsymbol{1}_{m_k}\\
            \boldsymbol{0}_{(p-p_0-m_kr_k)\times r_k}
        \end{bmatrix}, \quad\text{where}\quad m_k=\left\lfloor\frac{p-p_0}{r_k}\right\rfloor,
    \end{align*}
    for $r_k>1$, where $\otimes$ is the Kronecker product, and $\boldsymbol{H}_k=\boldsymbol{1}_{p-p_0}$ for $r_k=1$), and 
    $$\boldsymbol{f}^{(k)}\sim\mathbb{N}(\boldsymbol{0}_{r_k},\boldsymbol{I}_{r_k}),\boldsymbol{u}^{(k)}\sim\mathbb{N}(\boldsymbol{0}_p,\boldsymbol{\Sigma}_{\boldsymbol{u}}^{(k)}),\mathcal{E}^{(k)}\sim\mathbb{N}(0,1),\boldsymbol{w}^{(k)}=\boldsymbol{w}^{(0)}=\boldsymbol{\beta},~\forall k\in\{0\}\cup{\mathcal{A}}_\eta,$$
    where $\mathrm{supp}(\boldsymbol{\beta})\subseteq[p_0]$. Under this sub-model, we have $y^{(k)}=\boldsymbol{x}^{(k)\top}\boldsymbol{\beta}+\mathcal{E}^{(k)}$, $\|\boldsymbol{\varphi}^{(k)}\|_2=0$, and $\boldsymbol{x}^{(k)}\sim\mathbb{N}(\boldsymbol{0}_p,\boldsymbol{\Sigma}_{\boldsymbol{x}}^{(k)})$ with $\boldsymbol{\Sigma}_{\boldsymbol{x}}^{(k)}=\begin{bmatrix}
        \boldsymbol{I}_{p_0}& \boldsymbol{0}_{p_0,p-p_0}\\
        \boldsymbol{0}_{p-p_0,p_0}& \boldsymbol{I}_{p-p_0}+\boldsymbol{H}_k\boldsymbol{H}_k^\top
    \end{bmatrix}$.

    By applying Lemma G.2 of \cite{chen2013sparse} (see also Lemma 4.10 of \cite{massart2007concentration}) with the ambient dimension $p_0$ and sparsity level $s$, there exist $\phi_1,\ldots, \phi_m\in\{0,1\}^{p_0}$ such that $\|\phi_i\|_0=s,\|\phi_i-\phi_j\|_0\geq s/2$ with $\log m\geq c_0s\log(p_0/s)$. 
    
    For some amplitude $a_1>0$ to be specified later, define $\boldsymbol{\beta}_j=(a_1\phi_j^\top,\boldsymbol{0}_{p-p_0}^\top)^\top$, and take $\boldsymbol{\beta}_0=\boldsymbol{0}_p$ as the reference parameter. We then have $\|\boldsymbol{\beta}_j\|_0=s$, $\|\boldsymbol{\beta}_j\|_2^2=a_1^2s$, $\max_{i\neq j}\|\boldsymbol{\beta}_i-\boldsymbol{\beta}_j\|_2^2\leq 2a_1^2s$, and $\min_{i\neq j}\|\boldsymbol{\beta}_i-\boldsymbol{\beta}_j\|_2^2\geq a_1^2s/2$. For notational simplicity, let $\mathbb{P}_{\boldsymbol{\beta}}$ be the joint distribution of $(\boldsymbol{x},y)$ with $y\mid\boldsymbol{x}\sim\mathbb{N}(\boldsymbol{x}^\top\boldsymbol{\beta},1)$, the KL divergence between two Gaussian linear models can be computed as $$D_{KL}(\mathbb{P}_{\boldsymbol{\beta}_1}||\mathbb{P}_{\boldsymbol{\beta}_2})=\frac{1}{2}\mathbb{E}_{\boldsymbol{x}}[(\boldsymbol{x}^\top\boldsymbol{\beta}_1-\boldsymbol{x}^\top\boldsymbol{\beta}_2)^2]=\frac{1}{2}\|\boldsymbol{\beta}_1-\boldsymbol{\beta}_2\|_2^2,$$ which implies that $$\sup_{\boldsymbol{\beta}_1,\boldsymbol{\beta}_2}D_{KL}\Big(\bigotimes_{k\in\{0\}\cup\mathcal{A}_\eta}\mathbb{P}_{\boldsymbol{\beta}_1}^{\otimes n_k}||\bigotimes_{k\in\{0\}\cup\mathcal{A}_\eta}\mathbb{P}_{\boldsymbol{\beta}_2}^{\otimes n_k}\Big)=\frac{1}{2}(n_{\mathcal{A}_\eta}+n_0)\sup_{\boldsymbol{\beta}_1,\boldsymbol{\beta}_2}\|\boldsymbol{\beta}_1-\boldsymbol{\beta}_2\|_2^2\leq(n_{\mathcal{A}_\eta}+n_0)a_1^2s.$$
    Moreover, since $\min_{i\neq j}\|\boldsymbol{\beta}_i-\boldsymbol{\beta}_j\|_2\geq a_1\sqrt{s/2}$, $\{\boldsymbol{\beta}_j\}_{j\in[m]}$, under the $\ell_2$ metric, this forms an $\epsilon=a_1\sqrt{s/2}$ packing. By Fano's lemma (e.g., Lemma 3 of \cite{yu1997assouad}), we obtain that 
    \begin{align*}
        \inf_{\hat{\boldsymbol{\beta}}}\sup_{\boldsymbol{\mathcal{B}}\in\Theta_{\boldsymbol{b}_n}(s,0)}\mathbb{P}\left(\|\hat{\boldsymbol{\beta}}-\boldsymbol{\beta}\|_2^2\geq \frac{\epsilon^2}{4}\right)\geq 1-\frac{(n_{\mathcal{A}_\eta}+n_0)a_1^2s+\log 2}{c_0s\log(p_0/s)}.
    \end{align*}
Choosing $a_1^2=(c_0/20)\log(p_0/s)/(n_{\mathcal{A}_\eta}+n_0)$, we obtain that 
    \begin{align*}
        \inf_{\hat{\boldsymbol{\beta}}}\sup_{\boldsymbol{\mathcal{B}}\in\Theta_{\boldsymbol{b}_n}(s,0)}\mathbb{P}\left(\|\hat{\boldsymbol{\beta}}-\boldsymbol{\beta}\|_2^2\geq \frac{c_0s\log(p_0/s)}{160(n_{\mathcal{A}_\eta}+n_0)}\right)\geq 0.9
    \end{align*}
    when $c_0s\log(p_0/s)\geq 20\log 2$. Note that the case where $c_0s\log(p_0/s)< 20\log 2$ is trivial and simply follows from standard Le Cam's two-point argument by considering $\boldsymbol{\beta}_0=\boldsymbol{0}_p$ and $\boldsymbol{\beta}_1=a_1\boldsymbol{e}_1$ with $a_1^2\asymp 1/(n_{\mathcal{A}_\eta}+n_0)$ (although the probability in the case may be degenerated to $1/4$).

    \textbf{$\ell_2$ bound: Step 2.} In this case, we quantify the transfer bias induced by heterogeneity. Consider the case $$\eta^2\wedge\eta\sqrt{\frac{\log(p_0/s)}{n_0}}\wedge \frac{s\log(p_0/s)}{n_0}\geq c\frac{s\log(p_0/s)}{n_{\mathcal{A}_\eta}+n_0}\vee\frac{sb_{n,\min}^2}{p^2}.$$
    
    In this case, we consider the following sub-model:  $$\boldsymbol{\Sigma}_{\boldsymbol{u}}^{(k)}=\boldsymbol{I}_p,~\boldsymbol{\gamma}^{(k)}=\boldsymbol{0}_{r_k},~\boldsymbol{B}_k=\begin{bmatrix}
        \boldsymbol{0}_{p_0\times r_k}\\
        \boldsymbol{H}_k
    \end{bmatrix},$$
    where $\boldsymbol{H}_k\in\mathbb{R}^{(p-p_0)\times r_k}$ is chosen as in Step 1, 
    $$\boldsymbol{f}^{(k)}\sim\mathbb{N}(\boldsymbol{0}_{r_k},\boldsymbol{I}_{r_k}),~\boldsymbol{u}^{(k)}\sim\mathbb{N}(\boldsymbol{0}_p,~\boldsymbol{\Sigma}_{\boldsymbol{u}}^{(k)}),~\mathcal{E}^{(k)}\sim\mathbb{N}(0,1),~k\in\{0\}\cup{\mathcal{A}}_\eta,$$
    and $\boldsymbol{w}^{(k)}=\boldsymbol{0}_p,~\forall k\in[K],~\boldsymbol{w}^{(0)}=\boldsymbol{\beta},$ where $\boldsymbol{\beta}$ is chosen to satisfy $\|\boldsymbol{\beta}\|_0\leq s$, $\|\boldsymbol{\beta}\|_1\leq \eta$, and $\mathrm{supp}(\boldsymbol{\beta})\subseteq[p_0]$. This belongs to the parameter space $\Theta_{\boldsymbol{b}_n}(s,\eta)$. Under this sub-model, we have $y^{(k)}=\boldsymbol{x}^{(k)\top}\boldsymbol{w}^{(k)}+\mathcal{E}^{(k)}$, $\|\boldsymbol{\varphi}^{(k)}\|_2=0$, and $\boldsymbol{x}^{(k)}\sim\mathbb{N}(\boldsymbol{0}_p,\boldsymbol{\Sigma}_{\boldsymbol{x}}^{(k)})$ with $\boldsymbol{\Sigma}_{\boldsymbol{x}}^{(k)}=\begin{bmatrix}
        \boldsymbol{I}_{p_0}& \boldsymbol{0}_{p_0,p-p_0}\\
        \boldsymbol{0}_{p-p_0,p_0}& \boldsymbol{I}_{p-p_0}+\boldsymbol{H}_k\boldsymbol{H}_k^\top
    \end{bmatrix}$. We assume that $1\leq s\leq p_0/4$ without loss of generality, as the case $s\geq p_0/4$ is harder such that $\Theta_{\boldsymbol{b}_n}(s,\eta) \subseteq \Theta_{\boldsymbol{b}_n}(s',\eta)$ with $s\leq s'$. For some positive integer $q\leq s$, similar to Step 1, we apply Lemma G.2 of \cite{chen2013sparse} with the ambient dimension $p_0$ and sparsity level $q$, there exist $\phi_1,\ldots, \phi_m\in\{0,1\}^{p_0}$ such that $\|\phi_i\|_0=q,\|\phi_i-\phi_j\|_0\geq q/2$ with $\log m\geq c_0q\log(p_0/q)$.

    For some amplitude $a_2>0$ to be specified later, define $\boldsymbol{\beta}_j=(a_2\phi_j^\top,\boldsymbol{0}_{p-p_0}^\top)^\top$, and take $\boldsymbol{\beta}_0=\boldsymbol{0}_p$ as the reference parameter. We then have $\|\boldsymbol{\beta}_j\|_0=q$, $\|\boldsymbol{\beta}_j\|_1=a_2q$, $\|\boldsymbol{\beta}_j\|_2^2=a_2^2q$, $\max_{i\neq j}\|\boldsymbol{\beta}_i-\boldsymbol{\beta}_j\|_2^2\leq 2a_2^2q$, and $\min_{i\neq j}\|\boldsymbol{\beta}_i-\boldsymbol{\beta}_j\|_2^2\geq a_2^2q/2$. The KL divergence can be computed as $$\sup_{i,j\in[m]}D_{KL}\bigg(\mathbb{P}_{\boldsymbol{\beta}_i}^{\otimes n_0}\otimes\Big(\bigotimes_{k\in\mathcal{A}_\eta}\mathbb{P}_{\boldsymbol{0}_p}^{\otimes n_k}\Big)\Big|\Big|\mathbb{P}_{\boldsymbol{\beta}_j}^{\otimes n_0}\otimes\Big(\bigotimes_{k\in\mathcal{A}_\eta}\mathbb{P}_{\boldsymbol{0}_p}^{\otimes n_k}\Big)\bigg)=\frac{1}{2}n_0\sup_{i,j\in[m]}\|\boldsymbol{\beta}_i-\boldsymbol{\beta}_j\|_2^2\leq n_0a_2^2q.$$
    Moreover, since $\min_{i\neq j\in[m]}\|\boldsymbol{\beta}_i-\boldsymbol{\beta}_j\|_2\geq a_2\sqrt{q/2}$, under the $\ell_2$ metric, this forms an $\epsilon_2=a_2\sqrt{q/2}$ packing. By the Fano's Lemma (e.g., Lemma 3 of \cite{yu1997assouad}), we obtain that 
    \begin{align*}
        \inf_{\hat{\boldsymbol{\beta}}}\sup_{\boldsymbol{\mathcal{B}}\in\Theta_{\boldsymbol{b}_n}(s,\eta)}\mathbb{P}\left(\|\hat{\boldsymbol{\beta}}-\boldsymbol{\beta}\|_2^2\geq \frac{\epsilon_2^2}{4}\right)\geq 1-\frac{n_0a_2^2q+\log 2}{c_0q\log(p_0/q)}.
    \end{align*}
    Choosing $$a_2^2=c\left(\frac{\log(p_0/q)}{n_0}\wedge\frac{\eta^2}{q^2}\right),$$ we have $\|\boldsymbol{\beta}_j\|_1\leq \eta$, and
    \begin{align*}
        \inf_{\hat{\boldsymbol{\beta}}}\sup_{\boldsymbol{\mathcal{B}}\in\Theta_{\boldsymbol{b}_n}(s,\eta)}\mathbb{P}\left(\|\hat{\boldsymbol{\beta}}-\boldsymbol{\beta}\|_2^2\geq C\left(\frac{q\log(p_0/q)}{n_0}\wedge\frac{\eta^2}{q}\right)\right)\geq \frac{1}{4}
    \end{align*}
    for some constant $C>0$. Note that $q\log(p_0/q)\geq q\log(p_0/s)$ when $q\leq s$.
    
    If $\eta\leq\sqrt{\log(p_0/s)/n_0}$, by choosing $q=1$, we have 
    $$\frac{q\log(p_0/q)}{n_0}\wedge\frac{\eta^2}{q}=\frac{\log(p_0)}{n_0}\wedge\eta^2=\eta^2.$$ 
    Else if $\sqrt{\log(p_0/s)/n_0}< \eta<s\sqrt{\log(p_0/s)/n_0}$, by choosing $q=\eta/\sqrt{\log(p_0/s)/n_0}\in(1,s)$ (choosing $q=\lfloor \eta/\sqrt{\log(p_0/s)/n_0}\rfloor$ if this is not integer, and we can obtain $\frac{q\log(p_0/s)}{n_0}\wedge\frac{\eta^2}{q}\geq\eta\sqrt{\log(p_0/s)/n_0}/2$ in this case), we have 
    $$\frac{q\log(p_0/s)}{n_0}\wedge\frac{\eta^2}{q}=\eta\sqrt{\log(p_0/s)/n_0}.$$ 
    Else, when $\eta\geq s\sqrt{\log(p_0/s)/n_0}$, by choosing $q=s$, we have 
    $$\frac{q\log(p_0/q)}{n_0}\wedge\frac{\eta^2}{q}=\frac{s\log(p_0/s)}{n_0}\wedge\frac{\eta^2}{s}=\frac{s\log(p_0/s)}{n_0}.$$ 
    Notice that the induced rate of the lower bound also corresponds to the smallest value among the three rates under each case. Hence, we conclude that
    \begin{align*}
        \inf_{\hat{\boldsymbol{\beta}}}\sup_{\boldsymbol{\mathcal{B}}\in\Theta_{\boldsymbol{b}_n}(s,\eta)}\mathbb{P}\left(\|\hat{\boldsymbol{\beta}}-\boldsymbol{\beta}\|_2^2\geq C\left(\eta^2\wedge\eta\sqrt{\frac{\log(p_0/s)}{n_0}}\wedge\frac{s\log(p_0/s)}{n_0}\right)\right)\geq \frac{1}{4}.
    \end{align*}

    \textbf{$\ell_2$ bound: Step 3.} Consider the case that $$\frac{sb_{n,\min}^2}{p^2}\geq c\frac{s\log(p_0/s)}{n_{\mathcal{A}_\eta}+n_0}\vee\left(\eta^2\wedge\eta\sqrt{\frac{\log(p_0/s)}{n_0}}\wedge \frac{s\log(p_0/s)}{n_0}\right).$$
    
    In this case, we construct two distinct parameter points that induce the same joint distribution of the observed data following \cite{agarwal2012noisy}.

Let $S_1,S_2\subseteq[p_0]$ be disjoint
with $|S_1|=|S_2|=s$, and take
$g=b_{n,\min}$. Define
\begin{align*}
    \boldsymbol{v}=\frac{a_3}{p}(\boldsymbol{z}_{S_1}-\boldsymbol{z}_{S_2}),\qquad\boldsymbol{\beta}_{[1]}=\frac{ga_3}{p}\boldsymbol{z}_{S_1},\qquad\boldsymbol{\beta}_{[2]}=\frac{ga_3}{p}\boldsymbol{z}_{S_2},
\end{align*}
for some $a_3>0$ to be specified later, where $\boldsymbol{z}_S=\boldsymbol{1}_p\circ (\mathds{1}(1\in S),\ldots, \mathds{1}(p\in S))^\top$, $\mathds{1}$ is the indicator function, and $\circ$ is the Hadamard product. For $r_k>1$, define
\begin{align*}
        \boldsymbol{H}_k=\begin{bmatrix}
            \boldsymbol{0}_{p_0\times (r_k-1)}\\
            \boldsymbol{I}_{r_k-1}\otimes \boldsymbol{1}_{m_k}\\
            -\boldsymbol{I}_{r_k-1}\otimes \boldsymbol{1}_{m_k}\\
            \boldsymbol{0}_{(p-p_0-2m_k(r_k-1))\times (r_k-1)}
        \end{bmatrix}, \quad\text{where}\quad m_k=\left\lfloor\frac{p-p_0}{2(r_k-1)}\right\rfloor\quad \text{when}\quad r_k>1,
    \end{align*}
and let $\boldsymbol{H}_k$ be the empty matrix when $r_k=1$. Thus,
\begin{align*}
    \boldsymbol{H}_k^\top\boldsymbol{1}_p=\boldsymbol{0}_{r_k-1},~\boldsymbol{H}_k^\top\boldsymbol{v}
    =\boldsymbol{0}_{r_k-1},~\boldsymbol{H}_k^\top\boldsymbol{\beta}_{[l]}=\boldsymbol{0}_{r_k-1},~\boldsymbol{H}_k^\top\boldsymbol{H}_k=2m_k\boldsymbol{I}_{r_k-1},~\|\boldsymbol{H}_k\|_{\max}\leq 2.
\end{align*}
For $l=1,2$ and
$k\in\{0\}\cup\mathcal{A}_\eta$, define
\begin{align*}
    \boldsymbol{B}_{k[1]}&=\begin{bmatrix}
    \boldsymbol{1}_p-\boldsymbol{v}/2&
    \boldsymbol{H}_k
    \end{bmatrix},\quad\boldsymbol{B}_{k[2]}
    =\begin{bmatrix}
        \boldsymbol{1}_p+\boldsymbol{v}/2& \boldsymbol{H}_k
    \end{bmatrix},\quad\boldsymbol{w}_{[l]}^{(k)}=\boldsymbol{\beta}_{[l]},\quad\boldsymbol{\gamma}_{[l]}^{(k)}=\boldsymbol{B}_{k[l]}^\top\boldsymbol{\beta}_{[l]}+g\boldsymbol{e}_1.\\
    \boldsymbol{\Sigma}_{\boldsymbol{u}[1]}^{(k)}&=\boldsymbol{I}_p+\frac{\boldsymbol{1}_p\boldsymbol{v}^\top+\boldsymbol{v}\boldsymbol{1}_p^\top}{2}-\frac{\boldsymbol{v}\boldsymbol{v}^\top}{4},\quad\boldsymbol{\Sigma}_{\boldsymbol{u}[2]}^{(k)}=\boldsymbol{I}_p-\frac{\boldsymbol{1}_p\boldsymbol{v}^\top+\boldsymbol{v}\boldsymbol{1}_p^\top}{2}-\frac{\boldsymbol{v}\boldsymbol{v}^\top}{4},
\end{align*}
Furthermore, let
\begin{align*}
    \boldsymbol{u}_{[l]}^{(k)}\sim\mathbb{N}(\boldsymbol{0}_p,\boldsymbol{\Sigma}_{\boldsymbol{u}[l]}^{(k)}),\quad\boldsymbol{f}_{[l]}^{(k)}\sim\mathbb{N}(\boldsymbol{0}_{r_k},\boldsymbol{I}_{r_k}),\quad\mathcal{E}_{[l]}^{(k)}\sim\mathbb{N}(0,1),\quad l=1,2
\end{align*}
be mutually independent. 

Elementary calculations yield that
\begin{align*}
    \boldsymbol{B}_{k[l]}^\top\boldsymbol{B}_{k[l]}=\mathrm{diag}\left(p+\frac{\|\boldsymbol{v}\|_2^2}{4},2m_k,\ldots,2m_k\right),\|\boldsymbol{B}_{k[l]}\|_{\max}\leq \max\{1+\frac{a_3}{2p},2\}.
\end{align*}
Since $2s\leq p$,
\begin{align*}
    &\|\boldsymbol{\Sigma}_{\boldsymbol{u}[l]}^{(k)}-\boldsymbol{I}_p\|_1\leq \|\boldsymbol{1}_p\boldsymbol{v}^\top\|_1/2+\|\boldsymbol{v}\boldsymbol{1}_p^\top\|_1/2+\|\boldsymbol{v}\boldsymbol{v}^\top\|_1/4=\frac{a_3}{2}+\frac{sa_3}{p}+\frac{a_3^2s}{2p^2}\leq a_3+\frac{a_3^2}{4},\\
    &\mathrm{Var}(\boldsymbol{u}_i\boldsymbol{u}_j)=\Sigma_{ii}\Sigma_{jj}+\Sigma_{ij}^2\geq \left(1-a_3-\frac{a_3^2}{4}\right)^2.
\end{align*}
By choosing sufficiently small constant $a_3$ such that
\begin{align*}
    a_3+\frac{a_3^2}{4}\leq\min\left\{1-\sqrt{\iota},\,\iota^{-1}-1,\,\frac{1}{4}\right\},
\end{align*}
the eigenvalues of $\boldsymbol{\Sigma}_{\boldsymbol{u}[l]}^{(k)}$
are uniformly bounded away from zero and infinity, and
\begin{align*}
    \|\boldsymbol{\Sigma}_{\boldsymbol{u}[l]}^{(k)}\|_1\leq\iota^{-1},\qquad\mathrm{Var}(\boldsymbol{u}_i\boldsymbol{u}_j)\geq \iota.
\end{align*}
The remaining moment conditions follow
from the Gaussian distribution and the
uniformly bounded eigenvalues. In addition,
\begin{align*}
    \boldsymbol{\varphi}_{[l]}^{(k)}=\boldsymbol{\gamma}_{[l]}^{(k)}-\boldsymbol{B}_{k[l]}^\top\boldsymbol{w}_{[l]}^{(k)}=g\boldsymbol{e}_1,
\end{align*}
which implies that
$\|\boldsymbol{\varphi}_{[l]}^{(k)}\|_2
=b_{n,\min}\leq b_{n,k}$. Since the regression coefficients and idiosyncratic covariance matrices are identical across domains within each parameter point, Assumption \ref{assum3}
holds with $C_1=1$. Thus, both parameter points belong to $\Theta_{\boldsymbol{b}_n}(s,\eta)$.

For both parameter points, a direct calculation gives $\boldsymbol{\Sigma}_{\boldsymbol{x}[1]}^{(k)}=\boldsymbol{\Sigma}_{\boldsymbol{x}[2]}^{(k)}=\boldsymbol{I}_p+\boldsymbol{1}_p\boldsymbol{1}_p^\top+\boldsymbol{H}_k\boldsymbol{H}_k^\top$, and $\boldsymbol{f}_{[l]}^{(k)}|\boldsymbol{x}_{[l]}^{(k)}\sim\mathbb{N}\big(\boldsymbol{B}_{k[l]}^\top(\boldsymbol{\Sigma}_{\boldsymbol{x}[l]}^{(k)})^{-1}\boldsymbol{x}_{[l]}^{(k)},\boldsymbol{I}_{r_k}-\boldsymbol{B}_{k[l]}^\top(\boldsymbol{\Sigma}_{\boldsymbol{x}[l]}^{(k)})^{-1}\boldsymbol{B}_{k[l]}\big)$ by conditional Gaussian distributions. Hence, $y_{[l]}^{(k)}|\boldsymbol{x}_{[l]}^{(k)}\sim\mathbb{N}\Big(\boldsymbol{x}_{[l]}^{(k)\top}\big((\boldsymbol{\Sigma}_{\boldsymbol{x}[l]}^{(k)})^{-1}\boldsymbol{B}_{k[l]}\boldsymbol{\varphi}_{[l]}^{(k)}+\boldsymbol{\beta}_{[l]}^{(k)}\big),1+\boldsymbol{\varphi}_{[l]}^{(k)\top}\big(\boldsymbol{I}_{r_k}-\boldsymbol{B}_{k[l]}^\top(\boldsymbol{\Sigma}_{\boldsymbol{x}[l]}^{(k)})^{-1}\boldsymbol{B}_{k[l]}\big)\boldsymbol{\varphi}_{[l]}^{(k)}\Big)$. By construction, we have 
\begin{align*}
    (\boldsymbol{\Sigma}_{\boldsymbol{x}}^{(k)})^{-1}\boldsymbol{1}_p=\frac{\boldsymbol{1}_p}{p+1},\quad (\boldsymbol{\Sigma}_{\boldsymbol{x}}^{(k)})^{-1}\boldsymbol{v}=\boldsymbol{v},\quad \boldsymbol{B}_{k[1]}\boldsymbol{\varphi}_{[1]}^{(k)}=g(\boldsymbol{1}_p-\boldsymbol{v}/2),\quad \boldsymbol{B}_{k[2]}\boldsymbol{\varphi}_{[2]}^{(k)}=g(\boldsymbol{1}_p+\boldsymbol{v}/2),
\end{align*}
which implies that 
$$(\boldsymbol{\Sigma}_{\boldsymbol{x}[l]}^{(k)})^{-1}\boldsymbol{B}_{k[l]}\boldsymbol{\varphi}_{[l]}^{(k)}+\boldsymbol{\beta}_{[l]}^{(k)}=\frac{ga_3}{2p}(\boldsymbol{z}_{S_1}+\boldsymbol{z}_{S_2})+\frac{g}{p+1}\boldsymbol{1}_p,$$
and 
\begin{align*}
    \boldsymbol{\varphi}_{[l]}^{(k)\top}\big(\boldsymbol{I}_{r_k}-\boldsymbol{B}_{k[l]}^\top(\boldsymbol{\Sigma}_{\boldsymbol{x}[l]}^{(k)})^{-1}\boldsymbol{B}_{k[l]}\big)\boldsymbol{\varphi}_{[l]}^{(k)}=&g^2\Big(1-(\boldsymbol{B}_{k[l]}\boldsymbol{e}_1)^\top(\boldsymbol{\Sigma}_{\boldsymbol{x}[l]}^{(k)})^{-1}(\boldsymbol{B}_{k[l]}\boldsymbol{e}_1)\Big)\\
    =&g^2\Big(1-\big(\boldsymbol{1}_{p}\pm\frac{\boldsymbol{v}}{2}\big)^\top(\boldsymbol{\Sigma}_{\boldsymbol{x}[l]}^{(k)})^{-1}\big(\boldsymbol{1}_{p}\pm\frac{\boldsymbol{v}}{2}\big)\Big)\\
    =&g^2\Big(1-\frac{p}{p+1}-\frac{\|\boldsymbol{v}\|_2^2}{4}\Big)>0,
\end{align*}
where $\boldsymbol{e}_1$ is the first standard basis vector in $\mathbb{R}^{r_k}$. Since the two models are jointly Gaussian with the same mean vector and covariance matrix, they induce the same observation law. Consequently, $\mathbb{P}_1=\mathbb{P}_2$.

So far, we have proved that the two distinct parameter points induce the same observation law, with the associated KL divergence being $0$, by the Le Cam's two point test \cite[p.476, Lemma 1]{le2012asymptotic} and the Pinsker inequality, we have 
    \begin{align*}
        &\inf_{\hat{\boldsymbol{\beta}}}\max_{l\in\{1,2\}}\mathbb{P}_{l}\Big(\|\hat{\boldsymbol{\beta}}-\boldsymbol{\beta}_{[l]}\|_2^2\geq \|\boldsymbol{\beta}_{[1]}-\boldsymbol{\beta}_{[2]}\|_2^2/4\Big)\geq\inf_{\hat{\boldsymbol{\beta}}}\max_{l\in\{1,2\}}\mathbb{P}_{l}\Big(\mathop{\arg\min}\limits_{l'\in\{1,2\}}\|\hat{\boldsymbol{\beta}}-\boldsymbol{\beta}_{[l']}\|_2^2\neq l\Big)\\
        \geq& \frac{1}{2}\inf_{\hat{\boldsymbol{\beta}}}\Big\{\mathbb{P}_{1}\Big(\mathop{\arg\min}\limits_{l'\in\{1,2\}}\|\hat{\boldsymbol{\beta}}-\boldsymbol{\beta}_{[l']}\|_2^2\neq 1\Big)+\mathbb{P}_{2}\Big(\mathop{\arg\min}\limits_{l'\in\{1,2\}}\|\hat{\boldsymbol{\beta}}-\boldsymbol{\beta}_{[l']}\|_2^2\neq 2\Big)\Big\}\\
        \geq& \frac{1}{2}\{1-\|\mathbb{P}_{1}-\mathbb{P}_{2}\|_{TV}\}\geq \frac{1}{2}\left\{1-\sqrt{\frac{1}{2}D_{KL}(\mathbb{P}_{1}||\mathbb{P}_{2})}\right\}=\frac{1}{2},
    \end{align*}
    where $\mathbb{P}_{1}=\bigotimes_{k\in\{0\}\cup\mathcal{A}_\eta}\mathbb{P}_{\boldsymbol{\beta}_{[1]}}^{\otimes n_k},~\mathbb{P}_{2}=\bigotimes_{k\in\{0\}\cup\mathcal{A}_\eta}\mathbb{P}_{\boldsymbol{\beta}_{[2]}}^{\otimes n_k}$. Note that $$\|\boldsymbol{\beta}_{[1]}-\boldsymbol{\beta}_{[2]}\|_2^2=g^2\|\boldsymbol{v}\|_2^2=\frac{2a_3^2sb_{n,\min}^2}{p^2},$$ we have 
    \begin{align*}
        \inf_{\hat{\boldsymbol{\beta}}}\sup_{\boldsymbol{\mathcal{B}}\in\Theta_{\boldsymbol{b}_n}(s,\eta)}\mathbb{P}\left(\|\hat{\boldsymbol{\beta}}-\boldsymbol{\beta}\|_2^2\geq C'\frac{sb_{n,\min}^2}{p^2}\right)\geq 0.5.
    \end{align*}
    
    \textbf{$\ell_2$ bound: Conclusion.} Finally, since $p_0=\lfloor p/2\rfloor$, combining the three cases yields that 
    \begin{align*}
        \inf_{\hat{\boldsymbol{\beta}}}\!\sup_{\boldsymbol{\mathcal{B}}\in\Theta_{\boldsymbol{b}_n}(s,\eta)}\!\mathbb{P}\!\left(\!\frac{\|\hat{\boldsymbol{\beta}}-\boldsymbol{\beta}\|_2^2}{C'}\!\geq \!\frac{s\log(ep/s)}{n_{\mathcal{A}_\eta}+n_0}+\left(\!\eta^2\!\wedge\eta\sqrt{\frac{\log(ep/s)}{n_0}}\wedge\frac{s\log(ep/s)}{n_0}\right)\!+\frac{sb_{n,\min}^2}{p^2}\!\right)\!\geq \frac{1}{4}.
    \end{align*}
    for some constant $C'>0$. This completes the proof of $\ell_2$ lower bound.

    Now we provide the proof of minimax lower bound of the $\ell_1$ estimation error.
    
    \textbf{$\ell_1$ bound: Step 1.} We follow Step 1 of the $\ell_2$ bound. In this case, the constructed $\{\boldsymbol{\beta}_j\}_{j\in[m]}$ satisfies $\min_{i\neq j}\|\boldsymbol{\beta}_i-\boldsymbol{\beta}_j\|_1\geq a_1s/2$. Hence, we have 
    \begin{align*}
        \inf_{\hat{\boldsymbol{\beta}}}\sup_{\boldsymbol{\mathcal{B}}\in\Theta_{\boldsymbol{b}_n}(s,0)}\mathbb{P}\left(\|\hat{\boldsymbol{\beta}}-\boldsymbol{\beta}\|_1^2\geq \frac{(a_1s/2)^2}{4}\right)\geq 1-\frac{(n_{\mathcal{A}_\eta}+n_0)a_1^2s+\log 2}{c_0s\log(p_0/s)}.
    \end{align*}
    Choosing the same $a_1$ as in Step 1 of the $\ell_2$ bound yields that 
    \begin{align*}
        \inf_{\hat{\boldsymbol{\beta}}}\sup_{\boldsymbol{\mathcal{B}}\in\Theta_{\boldsymbol{b}_n}(s,0)}\mathbb{P}\left(\|\hat{\boldsymbol{\beta}}-\boldsymbol{\beta}\|_1^2\geq C\frac{s^2\log(p_0/s)}{n_{\mathcal{A}_\eta}+n_0}\right)\geq \frac{1}{4}.
    \end{align*}

    \textbf{$\ell_1$ bound: Step 2.} We follow Step 2 of the $\ell_2$ bound. In this case, the constructed $\{\boldsymbol{\beta}_j\}_{j\in[m]}$ satisfies $\min_{i\neq j}\|\boldsymbol{\beta}_i-\boldsymbol{\beta}_j\|_1\geq a_2q/2$. Hence, we have 
    \begin{align*}
        \inf_{\hat{\boldsymbol{\beta}}}\sup_{\boldsymbol{\mathcal{B}}\in\Theta_{\boldsymbol{b}_n}(s,\eta)}\mathbb{P}\left(\|\hat{\boldsymbol{\beta}}-\boldsymbol{\beta}\|_1^2\geq \frac{(a_2q/2)^2}{4}\right)\geq 1-\frac{n_0a_2^2q+\log 2}{c_0q\log(p_0/q)}.
    \end{align*}
    Choosing the same $a_2$ as in Step 2 of the $\ell_2$ bound and $q=s$ yield that 
    \begin{align*}
        \inf_{\hat{\boldsymbol{\beta}}}\sup_{\boldsymbol{\mathcal{B}}\in\Theta_{\boldsymbol{b}_n}(s,\eta)}\mathbb{P}\left(\|\hat{\boldsymbol{\beta}}-\boldsymbol{\beta}\|_1^2\geq C\left(\frac{s^2\log(p/s)}{n_0}\wedge \eta^2\right)\right)\geq \frac{1}{4}.
    \end{align*}

    \textbf{$\ell_1$ bound: Step 3.} We follow Step 3 of the $\ell_2$ bound. In this case, by selecting the same $a_3$ and $g$, we have 
    $$\|\boldsymbol{\beta}_{[1]}-\boldsymbol{\beta}_{[2]}\|_1=\frac{2a_3sb_{n,\min}}{p},$$ which implies that 
    \begin{align*}
        \inf_{\hat{\boldsymbol{\beta}}}\sup_{\boldsymbol{\mathcal{B}}\in\Theta_{\boldsymbol{b}_n}(s,\eta)}\mathbb{P}\left(\|\hat{\boldsymbol{\beta}}-\boldsymbol{\beta}\|_1^2\geq C'\frac{s^2b_{n,\min}^2}{p^2}\right)\geq 0.5.
    \end{align*}

    \textbf{$\ell_1$ bound: Conclusion.} Finally, combining the three steps yields that 
    \begin{align*}
        \inf_{\hat{\boldsymbol{\beta}}}\sup_{\boldsymbol{\mathcal{B}}\in\Theta_{\boldsymbol{b}_n}(s,\eta)}\mathbb{P}\left(\|\hat{\boldsymbol{\beta}}-\boldsymbol{\beta}\|_1^2/C'\geq \frac{s^2\log(ep/s)}{n_{\mathcal{A}_\eta}+n_0}+\left(\frac{s^2\log(ep/s)}{n_0}\wedge \eta^2\right)+\frac{s^2b_{n,\min}^2}{p^2}\right)\geq \frac{1}{4}.
    \end{align*}
    for some constant $C'>0$. This completes the proof of $\ell_1$ lower bound.

\end{proof}

\subsection{Proof of Theorem \ref{thm2}}\label{secA3}
\begin{proof}
Throughout this proof, we assume that the events in Lemma \ref{lemmaA6} hold with $\mathcal{L}=\{0\}\cup[K]$ and $t=100\log p$, which occurs with probability at least $1-2Ke^{-cn_{\min}}-Ke^{-cp}-5p^{-100}$. 

We write $\boldsymbol{U}_{0}^{[r]}=\mathrm{diag}(\{\mathds{1}(i\in R^{[r]})\}_{i=1}^{n_0})\boldsymbol{U}_{0}$ and $\boldsymbol{F}_{0}^{[r]}=\mathrm{diag}(\{\mathds{1}(i\in R^{[r]})\}_{i=1}^{n_0})\boldsymbol{F}_{0}$, where $\mathrm{diag}(\{\mathds{1}(i\in R^{[r]}\}_{i=1}^{n_0})$ is the $n_0\times n_0$ diagonal matrix whose $i$-th diagonal entry is $\mathds{1}(i\in R^{[r]})$, and define $$L_0^{[r]}(\boldsymbol{w})=\frac{1}{2}\mathbb{E}\left[\frac{1}{n_0/3}\|\boldsymbol{Y}^{(0)[r]}-\boldsymbol{U}_0^{[r]}\boldsymbol{w}-\boldsymbol{F}_0^{[r]}\boldsymbol{\gamma}^{(0)})\|_2^2\right]$$
for convenience. First, we consider $k\in\mathcal{A}_{\eta}$. For $r\in\{1,2,3\}$, $\hat{L}_{0}^{[r]}(\hat{\boldsymbol{w}}^{(k)[r]})-\hat{L}_{0}^{[r]}(\hat{\boldsymbol{\beta}}^{[r]})$ can be decomposed as
    \begin{align*}
        \hat{L}_{0}^{[r]}(\hat{\boldsymbol{w}}^{(k)[r]})-\hat{L}_{0}^{[r]}(\hat{\boldsymbol{\beta}}^{[r]})=&\{\hat{L}_{0}^{[r]}(\hat{\boldsymbol{w}}^{(k)[r]})-\hat{L}_{0}^{[r]}(\boldsymbol{\beta})-L_{0}^{[r]}(\hat{\boldsymbol{w}}^{(k)[r]})+L_{0}^{[r]}(\boldsymbol{\beta})\}+\{L_{0}^{[r]}(\hat{\boldsymbol{w}}^{(k)[r]})-L_{0}^{[r]}(\boldsymbol{\beta})\}\\
        &-\{\hat{L}_{0}^{[r]}(\hat{\boldsymbol{\beta}}^{[r]})-\hat{L}_{0}^{[r]}(\boldsymbol{\beta})-L_{0}^{[r]}(\hat{\boldsymbol{\beta}}^{[r]})+L_{0}^{[r]}(\boldsymbol{\beta})\}-\{L_{0}^{[r]}(\hat{\boldsymbol{\beta}}^{[r]})-L_{0}^{[r]}(\boldsymbol{\beta})\}\\
        =&I_1^{(k,r)}+I_2^{(k,r)}-I_3^{(k,r)}-I_4^{(k,r)}.
    \end{align*}
    
    For $I_2^{(k,r)}$, we have $$|I_2^{(k,r)}|=|L_{0}^{[r]}(\hat{\boldsymbol{w}}^{(k)[r]})-L_{0}^{[r]}(\boldsymbol{\beta})|=|\frac{1}{2}(\hat{\boldsymbol{w}}^{(k)[r]}-\boldsymbol{\beta})^\top\boldsymbol{\Sigma}_{\boldsymbol{u}}^{(0)}(\hat{\boldsymbol{w}}^{(k)[r]}-\boldsymbol{\beta})|\lesssim\|\hat{\boldsymbol{w}}^{(k)[r]}-\boldsymbol{\beta}\|_2^2$$ by the definition of $\boldsymbol{\beta}$. Similarly, we have $|I_4^{(k,r)}|\lesssim\|\hat{\boldsymbol{\beta}}^{[r]}-\boldsymbol{\beta}\|_2^2$. By Lemma \ref{lemmaA7}, we have 
    \begin{align*}
        \max\{|I_1^{(k,r)}|,|I_3^{(k,r)}|\}\lesssim&
        \sqrt{\frac{\log p}{n_0}}(\|\hat{\boldsymbol{\beta}}^{[r]}-\boldsymbol{\beta}\|_1\vee\|\hat{\boldsymbol{w}}^{(k)[r]}-\boldsymbol{\beta}\|_1\vee\|\hat{\boldsymbol{\beta}}^{[r]}-\boldsymbol{\beta}\|_2^2\vee\|\hat{\boldsymbol{w}}^{(k)[r]}-\boldsymbol{\beta}\|_2^2)\\
        &+\Big(\frac{1}{p}+\frac{\log p}{n_0}\Big)(\|\hat{\boldsymbol{\beta}}^{[r]}-\boldsymbol{\beta}\|_1^2\vee\|\hat{\boldsymbol{w}}^{(k)[r]}-\boldsymbol{\beta}\|_1^2)
    \end{align*}
    with probability at least $1-12|\mathcal{A}_\eta|p^{-100}$. By a similar proof of Theorem \ref{thm1}, if $\frac{\mathcal{V}_{n_0,p}\|\boldsymbol{\varphi}^{(0)}\|_2}{\sqrt{n_0\log p}}=O(1)$ and $\frac{s\log p}{n_0}+\frac{s}{p}\leq c$ for some sufficiently small constant $c>0$, we have
    $\|\hat{\boldsymbol{\beta}}^{[r]}-\boldsymbol{\beta}\|_2\lesssim \sqrt{\frac{s\log p}{n_0}},\|\hat{\boldsymbol{\beta}}^{[r]}-\boldsymbol{\beta}\|_1\lesssim s\sqrt{\frac{\log p}{n_0}}$ simultaneously for all $r\in[3]$ with probability at least $1-3p^{-100}$ according to Lemma \ref{lemmaA5}, and by Lemma \ref{lemmaA4}, with probability at least $1-6|\mathcal{A}_\eta|p^{-100}$ (similar to step 1 in the proof of Theorem \ref{thm1}, we further require $\frac{\mathcal{V}_{n_k,p}\|\boldsymbol{\varphi}^{(k)}\|_2}{(1\vee\tilde{\eta})\sqrt{(n_k+n_0)\log p}}\vee\frac{\tilde{\eta}\sqrt{n_k+n_0}}{(1\vee\tilde{\eta})p\sqrt{\log p}}=O(1)$),
    \begin{align*}
        \|\hat{\boldsymbol{w}}^{(k)[r]}-\boldsymbol{\beta}\|_2&\leq\|\hat{\boldsymbol{w}}^{(k)[r]}-\boldsymbol{w}^{(0,k)}\|_2+\|\boldsymbol{w}^{(0,k)}-\boldsymbol{\beta}\|_2\leq C\left(\eta+(1\vee\tilde{\eta})\sqrt{\frac{s\log p}{n_0+n_k}}\right),\\
        \|\hat{\boldsymbol{w}}^{(k)[r]}-\boldsymbol{\beta}\|_1&\leq \|\hat{\boldsymbol{w}}^{(k)[r]}-\boldsymbol{w}^{(0,k)}\|_1+\|\boldsymbol{w}^{(0,k)}-\boldsymbol{\beta}\|_1\leq C\left( (1\vee\tilde{\eta})s\sqrt{\frac{\log p}{n_0+n_k}}+\eta\right),
    \end{align*}
    where 
    \begin{align*}
        \boldsymbol{w}^{(0,k)}=\mathop{\arg\min}\limits_{\boldsymbol{w}\in\mathbb{R}^p}\mathbb{E}\Bigg[&\frac{1}{2(n_k+2n_0/3)}\Big\{\|(\boldsymbol{I}_{n_0}-\mathrm{diag}(\{\mathds{1}(i\in R^{[r]}\}_{i=1}^{n_0}))(\boldsymbol{Y}^{(0)}-\boldsymbol{U}_0\boldsymbol{w}-\boldsymbol{F}_0\boldsymbol{\gamma}^{(0)})\|_2^2\\
        &~~~~~~~~~~~~~~~~~~~~~~~~~~~~~~+\|\boldsymbol{Y}^{(k)}-\boldsymbol{U}_k\boldsymbol{w}-\boldsymbol{F}_k\boldsymbol{\gamma}^{(k)}\|_2^2\Big\}\Bigg].
\end{align*}
    under Assumptions \ref{assum1}--\ref{assum4} and $s\log p/n_0=O(1)$. Hence, we conclude that 
    \begin{align*}
        |I_1^{(k,r)}|+|I_2^{(k,r)}|+|I_3^{(k,r)}|+|I_4^{(k,r)}|\leq C\left(\frac{s\log p}{n_0}+\eta^2+\max_k\frac{\tilde{\eta}^2s\log p}{n_0+n_k}\right).
    \end{align*}

Now we turn to $k\in\mathcal{A}_{\eta}^{c}$. For $r\in\{1,2,3\}$, $\hat{L}_{0}^{[r]}(\hat{\boldsymbol{w}}^{(k)[r]})-\hat{L}_{0}^{[r]}(\hat{\boldsymbol{\beta}}^{[r]})$ can be decomposed as
    \begin{align*}
        &\hat{L}_{0}^{[r]}(\hat{\boldsymbol{w}}^{(k)[r]})-\hat{L}_{0}^{[r]}(\hat{\boldsymbol{\beta}}^{[r]})\\
        =&\{\hat{L}_{0}^{[r]}(\hat{\boldsymbol{w}}^{(k)[r]})-\hat{L}_{0}^{[r]}(\boldsymbol{w}^{(0,k)})-L_{0}^{[r]}(\hat{\boldsymbol{w}}^{(k)[r]})+L_{0}^{[r]}(\boldsymbol{w}^{(0,k)})\}+\{L_{0}^{[r]}(\hat{\boldsymbol{w}}^{(k)[r]})-L_{0}^{[r]}(\boldsymbol{w}^{(0,k)})\}\\
        &+\{\hat{L}_{0}^{[r]}(\boldsymbol{w}^{(0,k)})-\hat{L}_{0}^{[r]}(\boldsymbol{\beta})-L_{0}^{[r]}(\boldsymbol{w}^{(0,k)})+L_{0}^{[r]}(\boldsymbol{\beta})\}+\{L_{0}^{[r]}(\boldsymbol{w}^{(0,k)})-L_{0}^{[r]}(\boldsymbol{\beta})\}\\
        &-\{\hat{L}_{0}^{[r]}(\hat{\boldsymbol{\beta}}^{[r]})-\hat{L}_{0}^{[r]}(\boldsymbol{\beta})-L_{0}^{[r]}(\hat{\boldsymbol{\beta}}^{[r]})+L_{0}^{[r]}(\boldsymbol{\beta})\}-\{L_{0}^{[r]}(\hat{\boldsymbol{\beta}}^{[r]})-L_{0}^{[r]}(\boldsymbol{\beta})\}\\
        =&E_1^{(k,r)}+E_2^{(k,r)}+E_3^{(k,r)}+E_4^{(k,r)}-E_5^{(k,r)}-E_6^{(k,r)}.
    \end{align*}
Similar to the proof for $k\in\mathcal{A}_\eta$,
by Lemma \ref{lemmaA7}, we have
\begin{align*}
    &\max\{|E_1^{(k,r)}|,|E_3^{(k,r)}|,|E_5^{(k,r)}|\}\\
    \lesssim&\sqrt{\frac{\log p}{n_0}}\Big(\|\hat{\boldsymbol{w}}^{(k)[r]}-\boldsymbol{w}^{(0,k)}\|_1\vee\|\boldsymbol{w}^{(0,k)}-\boldsymbol{\beta}\|_1\vee\|\hat{\boldsymbol{\beta}}^{[r]}-\boldsymbol{\beta}\|_1\vee\|\hat{\boldsymbol{w}}^{(k)[r]}-\boldsymbol{w}^{(0,k)}\|_2^2\\
    &~~~~~~~~~~~~~~~~\vee\|\hat{\boldsymbol{w}}^{(k)[r]}-\boldsymbol{w}^{(0,k)}\|_2\|\boldsymbol{w}^{(0,k)}-\boldsymbol{\beta}\|_2\vee\|\boldsymbol{w}^{(0,k)}-\boldsymbol{\beta}\|_2^2\vee\|\hat{\boldsymbol{\beta}}^{[r]}-\boldsymbol{\beta}\|_2^2\Big)\\
    &+\Big(\frac{1}{p}+\frac{\log p}{n_0}\Big)\big(\|\hat{\boldsymbol{w}}^{(k)[r]}-\boldsymbol{w}^{(0,k)}\|_1^2\vee\|\hat{\boldsymbol{w}}^{(k)[r]}-\boldsymbol{w}^{(0,k)}\|_1\|\boldsymbol{w}^{(0,k)}-\boldsymbol{\beta}\|_1\vee\|\boldsymbol{w}^{(0,k)}-\boldsymbol{\beta}\|_1^2\vee\|\hat{\boldsymbol{\beta}}^{[r]}-\boldsymbol{\beta}\|_1^2\big)
\end{align*}
with probability at least $1-18|\mathcal{A}_\eta^c|p^{-100}$ for all $k\in\mathcal{A}_\eta^c$. Meanwhile, $|E_6^{(k,r)}|\lesssim\|\hat{\boldsymbol{\beta}}^{[r]}-\boldsymbol{\beta}\|_2^2$. For $E_2^{(k,r)}$, we have
\begin{align*}
    |E_2^{(k,r)}|&\leq|\mathbb{E}[\nabla L_{0}^{[r]}(\boldsymbol{w}^{(0,k)})]^\top(\hat{\boldsymbol{w}}^{(k)[r]}-\boldsymbol{w}^{(0,k)})+\frac{1}{2}v_0\|\hat{\boldsymbol{w}}^{(k)[r]}-\boldsymbol{w}^{(0,k)}\|_2^2\\
    &\leq v_0\|\boldsymbol{w}^{(0,k)}-\boldsymbol{\beta}\|_2\|\hat{\boldsymbol{w}}^{(k)[r]}-\boldsymbol{w}^{(0,k)}\|_2+\frac{1}{2}v_0\|\hat{\boldsymbol{w}}^{(k)[r]}-\boldsymbol{w}^{(0,k)}\|_2^2.
\end{align*}
For $E_2^{(k,r)}+E_4^{(k,r)}$, we obtain 
\begin{align*}
    E_2^{(k,r)}+E_4^{(k,r)}\geq&\lambda_{\min}(\boldsymbol{\Sigma}_{\boldsymbol{u}}^{(0)})\|\hat{\boldsymbol{w}}^{(k)[r]}-\boldsymbol{\beta}\|_2^2/2:=\underline{\kappa}\|\hat{\boldsymbol{w}}^{(k)[r]}-\boldsymbol{\beta}\|_2^2\\
    \geq&(\underline{\kappa}/2)\|\boldsymbol{w}^{(0,k)}-\boldsymbol{\beta}\|_2^2-\underline{\kappa}\|\hat{\boldsymbol{w}}^{(k)[r]}-\boldsymbol{w}^{(0,k)}\|_2^2\\
    \geq& \frac{\iota^4}{2}\underline{\kappa}(n_k/(n_0+n_k))^2\|\boldsymbol{w}^{(k)}-\boldsymbol{\beta}\|_2^2-\underline{\kappa}\|\hat{\boldsymbol{w}}^{(k)[r]}-\boldsymbol{w}^{(0,k)}\|_2^2
\end{align*}

By a similar proof of Theorem \ref{thm1}, simultaneously for all $r\in[3]$ and $k\in\mathcal{A}_\eta^c$, with probability at least $1-6|\mathcal{A}_\eta^c|p^{-100}$ (similar to step 1 in the proof of Theorem \ref{thm1}, we further require $\frac{\mathcal{V}_{n_k,p}\|\boldsymbol{\varphi}^{(k)}\|_2}{(1\vee\tilde{\eta})\sqrt{(n_k+n_0)\log p}}\vee\frac{\tilde{\eta}\sqrt{n_k+n_0}}{(1\vee\tilde{\eta})p\sqrt{\log p}}=O(1)$), we have
\begin{align*}
    \|\hat{\boldsymbol{w}}^{(k)[r]}-\boldsymbol{w}^{(0,k)}\|_2\leq& C\left((\tilde{\eta}^{1/2}\vee\tilde{\eta})\left(\frac{\log p}{n_0+n_k}\right)^{1/4}+(1\vee\tilde{\eta})\sqrt{\frac{s\log p}{n_0+n_k}}\right)\\
    \|\hat{\boldsymbol{w}}^{(k)[r]}-\boldsymbol{w}^{(0,k)}\|_1\leq& C\left(\tilde{\eta}+(1\vee\tilde{\eta})s\sqrt{\frac{\log p}{n_0+n_k}}\right).
\end{align*}
Analogously to Lemma \ref{lemmaA1}, we have $\|\boldsymbol{w}^{(0,k)}-\boldsymbol{\beta}\|_2\leq\|\boldsymbol{w}^{(0,k)}-\boldsymbol{\beta}\|_1\lesssim\tilde{\eta}.$ Therefore, we have
\begin{align*}
    &|E_1^{(k,r)}|+|E_3^{(k,r)}|+|E_5^{(k,r)}|+|E_6^{(k,r)}|\leq C\left(\frac{s\log p}{n_0}+(\tilde{\eta}\vee\tilde{\eta}^2)\sqrt{\frac{\log p}{n_0}}+\tilde{\eta}^2\frac{s\log p}{n_0+n_k}\right),\\
    &E_2^{(k,r)}+E_4^{(k,r)}\geq \frac{\iota^4}{2}\underline{\kappa}(n_k/(n_0+n_k))^2\|\boldsymbol{w}^{(k)}-\boldsymbol{\beta}\|_2^2-C\left((\tilde{\eta}\vee\tilde{\eta}^2)\left(\frac{\log p}{n_0+n_k}\right)^{1/2}+(1\vee\tilde{\eta}^2)\frac{s\log p}{n_0+n_k}\right).
\end{align*}

Hence, we conclude that, 
    \begin{align*}
        \mathbb{P}(\widehat{\mathcal{A}}\neq\mathcal{A}_\eta)\leq&\mathbb{P}\Bigg\{\Big\{\hat{L}_{0}^{[r]}(\hat{\boldsymbol{w}}^{(k)[r]})>\hat{L}_{0}^{[r]}(\hat{\boldsymbol{\beta}}^{[r]})+\epsilon_0\hat{\sigma}^2,~\exists k\in\mathcal{A}_\eta\Big\}\\
        &~~~~~~~~~~~~~~~~~~~~~~~~~~\bigcup\Big\{\hat{L}_{0}^{[r]}(\hat{\boldsymbol{w}}^{(k)[r]})\leq\hat{L}_{0}^{[r]}(\hat{\boldsymbol{\beta}}^{[r]})+\epsilon_0\hat{\sigma}^2,~\exists k\in\mathcal{A}_\eta^c\Big\},~\exists r\in[3]\Bigg\}\\
        \leq&3|\mathcal{A}_\eta|\mathbb{P}\left\{\epsilon_0c_{\sigma}^2\sigma_0^2<\max_{k\in\mathcal{A}_\eta,~r\in[3]}|I_1^{(k,r)}|+|I_2^{(k,r)}|+|I_3^{(k,r)}|+|I_4^{(k,r)}|\right\}\\
        &+3|\mathcal{A}_\eta^c|\mathbb{P}\left\{\min_{k\in\mathcal{A}_\eta^c,~r\in[3]}E_2^{(k,r)}+E_4^{(k,r)}-(|E_1^{(k,r)}|+|E_3^{(k,r)}|+|E_5^{(k,r)}|+|E_6^{(k,r)}|)\leq\epsilon_0C_{\sigma}^2\sigma_0^2\right\}\\
        &+\mathbb{P}\left(\{c_\sigma\leq\hat{\sigma}/\sigma_0\leq C_\sigma\}^c\right)\\
        \leq&3|\mathcal{A}_\eta|\mathbb{P}\left\{\epsilon_0c_{\sigma}^2\sigma_0^2< C\left(\frac{s\log p}{n_0}+\eta^2+\max_{k\in\mathcal{A}_\eta}\frac{\tilde{\eta}^2s\log p}{n_0+n_k}\right)\right\}\\
        &+3|\mathcal{A}_\eta^c|\mathbb{P}\Bigg\{\min_{k\in\mathcal{A}_\eta^c}\frac{\iota^4}{2}\underline{\kappa}(n_k/(n_0+n_k))^2\|\boldsymbol{w}^{(k)}-\boldsymbol{\beta}\|_2^2\\
        &\qquad\qquad-C\left(\frac{s\log p}{n_0}+(\tilde{\eta}\vee\tilde{\eta}^2)\sqrt{\frac{\log p}{n_0}}+\tilde{\eta}^2\frac{s\log p}{n_0+n_k}\right)\leq\epsilon_0C_{\sigma}^2\sigma_0^2\Bigg\}\\
        &+2Ke^{-cn_{\min}}+Ke^{-cp}+8p^{-100}+18|\mathcal{A}_\eta|p^{-100}+24|\mathcal{A}_\eta^c|p^{-100}\\
        &+\mathbb{P}\left(\{c_\sigma\leq\hat{\sigma}/\sigma_0\leq C_\sigma\}^c\right).
    \end{align*}
    Assumption \ref{assum4} (d) ensures that the last term on the RHS of the above inequality tends to $0$, while Assumption \ref{assum4} (c) ensures that the first two terms on the RHS of the above inequality tends to $0$. Finally, assuming that $\log K\leq c\min\{\min_{k\in\{0\}\cup[K]}n_k,p\}$ for some sufficiently small constant $c>0$ completes the proof.
\end{proof}

\subsection{Proof of Theorem \ref{thm3}}\label{secA4}
\begin{proof}
Before the proof, the following example illustrates the specific forms of $\Delta_1$, $\Delta_{\max}$, and $\Delta_{\infty}$ in Assumption \ref{assum5}.
\begin{example}[Proposition E.1 of \cite{fan2024latent}]
Suppose that $\hat{\boldsymbol{\Theta}}$ is the nodewise regression estimator constructed under the conditions of Proposition~E.1 of \citet{fan2024latent}, and that each row of $\boldsymbol{\Theta}$ has at least one nonzero off-diagonal entry. Then
\begin{align*}
\|\boldsymbol{I}_p-\hat{\boldsymbol{\Theta}}\hat{\boldsymbol{\Sigma}}_{\boldsymbol{u}}^{(0)}\|_{\max}&=O_P\left(\sqrt{\frac{\log p}{n_0}}+\frac{1}{\sqrt{p}}+\max_{j\in[p]}\frac{\mathcal{V}_{n_0,p}}{n_0}\|\boldsymbol{B}_0^\top\tilde{\boldsymbol{\omega}}_j^{(0)}\|_2\right),\\
\|\hat{\boldsymbol{\Theta}}-\boldsymbol{\Theta}\|_{\max}&=O_P\left(\max_{j\in[p]}\sqrt{|\mathcal{S}_j|\left(\frac{\log p}{n_0}+\frac{1}{p}\right)}+\max_{j\in[p]}\frac{\mathcal{V}_{n_0,p}}{n_0}\sqrt{|\mathcal{S}_j|}\|\boldsymbol{B}_0^\top\tilde{\boldsymbol{\omega}}_j^{(0)}\|_2\right),\\
\|\hat{\boldsymbol{\Theta}}-\boldsymbol{\Theta}\|_{\infty}&=O_P\left(\max_{j\in[p]}|\mathcal{S}_j|\sqrt{\frac{\log p}{n_0}+\frac{1}{p}}+\max_{j\in[p]}\frac{\mathcal{V}_{n_0,p}}{n_0}|\mathcal{S}_j|\|\boldsymbol{B}_0^\top\tilde{\boldsymbol{\omega}}_j^{(0)}\|_2\right),
\end{align*}
where $\mathcal{S}_j=\{l\neq j:\boldsymbol{\Theta}_{jl}\neq 0\}$ denotes the support of the $j$-th column of $\boldsymbol{\Theta}$, and $\boldsymbol{\omega}_j^{(0)}=\mathop{\arg\min}\limits_{\boldsymbol{\omega}\in\mathbb{R}^{p-1}}\mathbb{E}|u_j^{(0)}-\boldsymbol{\omega}^\top\boldsymbol{u}_{-j}^{(0)}|^2$. Here $\tilde{\boldsymbol{\omega}}_j^{(0)}\in\mathbb{R}^p$ is defined by $(\tilde{\boldsymbol{\omega}}_j^{(0)})_j=1$ and
$(\tilde{\boldsymbol{\omega}}_j^{(0)})_{-j}
=-\boldsymbol{\omega}_j^{(0)}$.
\end{example}
First, we have the following decomposition:
 
\begin{align}\label{eqA3}
    \tilde{\boldsymbol{\beta}}-\boldsymbol{\beta}=\frac{1}{n_0}\hat{\boldsymbol{\Theta}}\hat{\boldsymbol{U}}_0^\top\boldsymbol{\mathcal{E}}^{(0)}+\frac{1}{n_0}\hat{\boldsymbol{\Theta}}\hat{\boldsymbol{U}}_0^\top\boldsymbol{F}_0\boldsymbol{\mathcal{\varphi}}^{(0)}+(\boldsymbol{I}_p-\hat{\boldsymbol{\Theta}}\hat{\boldsymbol{\Sigma}}_{\boldsymbol{u}})(\hat{\boldsymbol{\beta}}-\boldsymbol{\beta}).
\end{align}
The last term in \eqref{eqA3} can be bounded as 
    $$\|(\boldsymbol{I}_p-\hat{\boldsymbol{\Theta}}\hat{\boldsymbol{\Sigma}}_{\boldsymbol{u}})(\hat{\boldsymbol{\beta}}-\boldsymbol{\beta})\|_{\max}\leq\|\boldsymbol{I}_p-\hat{\boldsymbol{\Theta}}\hat{\boldsymbol{\Sigma}}_{\boldsymbol{u}}\|_{\max}\|\hat{\boldsymbol{\beta}}-\boldsymbol{\beta}\|_1=O_P(\Delta_1\mathcal{R}_1),$$
     
    where $$\mathcal{R}_1=s(1\vee\eta)\left(\frac{\log p}{n_{\mathcal{A}_\eta}+n_0}\right)^{1/2}+\eta$$ is the $\ell_1$ error bound established in Theorem \ref{thm1}.
    Additionally, we have
    $$\|\frac{1}{n_0}\hat{\boldsymbol{\Theta}}\hat{\boldsymbol{U}}_0^\top\boldsymbol{F}_0\boldsymbol{\mathcal{\varphi}}^{(0)}\|_\infty=O_P(\mathcal{V}_{n_0,p}\|\boldsymbol{\Theta}\|_\infty\|\boldsymbol{\mathcal{\varphi}}^{(0)}\|_2/n_0),$$
    by Lemma \ref{lemmaA6} and
    \begin{align*}
        \frac{1}{n_0}\|\hat{\boldsymbol{\Theta}}\hat{\boldsymbol{U}}_0^\top\boldsymbol{\mathcal{E}}^{(0)}-\boldsymbol{\Theta}\boldsymbol{U}_0^\top\boldsymbol{\mathcal{E}}^{(0)}\|_\infty&\leq\frac{1}{n_0}\|\hat{\boldsymbol{\Theta}}\|_\infty\|(\hat{\boldsymbol{U}}_0-\boldsymbol{U}_0)^\top\boldsymbol{\mathcal{E}}^{(0)}\|_\infty+\frac{1}{n_0}\|\hat{\boldsymbol{\Theta}}-\boldsymbol{\Theta}\|_\infty\|\boldsymbol{U}_0^\top\boldsymbol{\mathcal{E}}^{(0)}\|_\infty\\
        &=O_P\left(\frac{\sqrt{\log p}}{n_0}\|\boldsymbol{\Theta}\|_\infty+\Delta_\infty\sqrt{\frac{\log p}{n_0}}\right).
    \end{align*}
     
    Hence, we have $\left\|\sqrt{n_0}(\tilde{\boldsymbol{\beta}}-\boldsymbol{\beta})-\boldsymbol{\Theta}\boldsymbol{U}_0^\top\boldsymbol{\mathcal{E}}^{(0)}/\sqrt{n_0}\right\|_\infty=O_P(\Delta_{n_0,p})$, where
    \begin{align*}
    \Delta_{n_0,p}&=\sqrt{n_0}\Delta_1\mathcal{R}_1+\frac{\mathcal{V}_{n_0,p}\|\boldsymbol{\Theta}\|_\infty\|\boldsymbol{\mathcal{\varphi}}^{(0)}\|_2}{\sqrt{n_0}}+\sqrt{\frac{\log p}{n_0}}\|\boldsymbol{\Theta}\|_\infty+\Delta_\infty\sqrt{\log p}=o\left(1\right).
\end{align*}
 
Additionally, since $\lambda_{\max}(\boldsymbol{\Sigma}_{\boldsymbol{u}}^{(0)})\leq\|\boldsymbol{\Sigma}_{\boldsymbol{u}}^{(0)}\|_1\leq1/\iota$, we have $\min_{j\in[p]}\boldsymbol{\Theta}_{j,j}\geq\lambda_{\min}(\boldsymbol{\Theta})\geq\iota$. Then, $\{\boldsymbol{\Theta}_j\boldsymbol{u}_i^{(0)}\mathcal{E}_i^{(0)}\}_{i=1}^{n_0}$ are i.i.d. zero-mean sub-exponential random variables with $\sigma_0^2\boldsymbol{\Theta}_{j,j}\geq\sigma_0^2\iota$ and $$\max_{j\in[p]}\|\boldsymbol{\Theta}_j\boldsymbol{u}_i^{(0)}\mathcal{E}_i^{(0)}\|_{\psi_1}\leq\frac{\|\boldsymbol{u}_i^{(0)}\|_{\psi_2}\|\mathcal{E}_i^{(0)}\|_{\psi_2}}{\iota}.$$ By Lemma \ref{LemmaA8}, we have
$$\lim_{n\to\infty}\sup_{t\in\mathbb{R}}\left|\mathbb{P}\left(\left\|\boldsymbol{\Theta}\boldsymbol{U}_0^\top\boldsymbol{\mathcal{E}}^{(0)}/\sqrt{n_0}\right\|_\infty\leq t\right)-\mathbb{P}(\|\boldsymbol{Q}\|_\infty\leq t)\right|=0$$

By Nazarov’s inequality \cite[Proposition 1]{chernozhukov2023high}, since the variance of each element in $\boldsymbol{Q}$ is bounded away from $0$ and $\infty$ by assumptions, we can infer that
$$\sup_{t\in\mathbb{R}}\mathbb{P}\left(|\|\boldsymbol{Q}\|_\infty-t|<\delta\log^{-1/2} (p)\right)\leq C\delta\log^{-1/2} (p)\{\sqrt{2\log p}+2\}.$$
By choosing a $\delta$ small enough, the right hand side of the inequality can be arbitrarily small. From Assumption \ref{assum5}, the remainder $\Delta_{n_0,p}$ satisfies $\Delta_{n_0,p}=o_P\big((\log p)^{-1/2}\big)$, and thus applying Lemma \ref{LemmaA9}, $$\sup_{t\in\mathbb{R}}\left|\mathbb{P}\left(\|\sqrt{n_0}(\tilde{\boldsymbol{\beta}}-\boldsymbol{\beta})\|_\infty\leq t\right)-\mathbb{P}(\|\boldsymbol{Q}\|_\infty\leq t)\right|=o(1).$$
\end{proof}

\subsection{Proof of Theorem \ref{thm4}}\label{secA5}
\begin{proof}
    It is clear that conditional on the data $\mathcal{D}$, one has $\hat{\boldsymbol{Q}}^{e}\sim\mathbb{N}(\mathbf{0}_p,\hat{\boldsymbol{\Upsilon}})$, where $\hat{\boldsymbol{\Upsilon}}=\hat{\sigma}^2\hat{\boldsymbol{\Theta}}\hat{\boldsymbol{\Sigma}}_{\boldsymbol{u}}^{(0)}\hat{\boldsymbol{\Theta}}^\top$ and $\hat{\boldsymbol{\Sigma}}_{\boldsymbol{u}}^{(0)}=n_0^{-1}\hat{\boldsymbol{U}}_0^\top\hat{\boldsymbol{U}}_0$. Denote $\boldsymbol{\Upsilon}=\sigma_0^2\boldsymbol{\Theta}$, now we turn to analyze $\|\sigma_0^2\boldsymbol{\Theta}-\hat{\boldsymbol{\Upsilon}}\|_{\max}$.
    \begin{align*}
        \|\boldsymbol{\Upsilon}-\hat{\boldsymbol{\Upsilon}}\|_{\max}&\leq\hat{\sigma}^2\|\hat{\boldsymbol{\Theta}}\|_\infty
        \|\boldsymbol{I}_p-\hat{\boldsymbol{\Theta}}\hat{\boldsymbol{\Sigma}}_{\boldsymbol{u}}^{(0)}\|_{\max}+\hat{\sigma}^2
        \|\hat{\boldsymbol{\Theta}}-\boldsymbol{\Theta}\|_{\max}+|\sigma_0^2-\hat{\sigma}^2|\|\boldsymbol{\Theta}\|_{\max}\\
        &=O_P(\Delta_1\|\boldsymbol{\Theta}\|_\infty+\Delta_{\infty}+\Delta_{\sigma})
    \end{align*}
    Thus, $\log p\|\sigma_0^2\boldsymbol{\Theta}-\hat{\boldsymbol{\Upsilon}}\|_{\max}=o_P(1)$ by assumptions. Further assume that $(\log\log p)\log p\|\sigma_0^2\boldsymbol{\Theta}-\hat{\boldsymbol{\Upsilon}}\|_{\max}=o_P(1)$, since the smallest eigenvalue of $\sigma_0^2\boldsymbol{\Theta}$ is bounded away from $0$, by the Gaussian comparison result \citep[Lemma 2.1]{chernozhukov2023nearly}, it can be derived that
    $$\sup_{t\in\mathbb{R}}\left|\mathbb{P}\left(\|\hat{\boldsymbol{Q}}^e\|_\infty\leq t|\mathcal{D}\right)-\mathbb{P}\left(\|\boldsymbol{Q}\|_\infty\leq t\right)\right|=o_P(1).$$

\end{proof}

\subsection{Proof of Corollary \ref{cor1}}\label{secA6}
\begin{proof}
    The proof follows from the proof of Theorems \ref{thm3}--\ref{thm4} by considering the subset $\mathcal{G}$.
\end{proof}

\subsection{Proof of Theorem \ref{thm5}}\label{secA7}
\begin{proof}
    Firstly, denote $\boldsymbol{S}=\mbox{diag}(\boldsymbol{\Theta})$ and $\hat{\boldsymbol{S}}=\mbox{diag}(\hat{\boldsymbol{\Theta}})$, it follows that $$\|\hat{\boldsymbol{S}}-\boldsymbol{S}\|_{\infty}\leq\|\boldsymbol{\Theta}-\hat{\boldsymbol{\Theta}}\|_{\infty}=O_P(\Delta_{\infty}).$$
    Since $\iota\leq\lambda_{\min}(\boldsymbol{\Theta})\leq\min_{j\in[p]}\boldsymbol{\Theta}_{j,j}\leq\max_{j\in[p]}\boldsymbol{\Theta}_{j,j}\leq\lambda_{\max}(\boldsymbol{\Theta})\leq\iota^{-1}$, which implies that the diagonal elements of $\boldsymbol{\Theta}$ are bounded away from $0$, it can be inferred that $\|\hat{\boldsymbol{S}}^{-1/2}-\boldsymbol{S}^{-1/2}\|_{\infty}=O_P(\Delta_{\infty}).$ Then it can be concluded that
    \begin{align*}
    &\|\sqrt{n_0}\hat{\boldsymbol{S}}^{-1/2}\{\tilde{\boldsymbol{\beta}}-\boldsymbol{\beta}\}-\boldsymbol{S}^{-1/2}\boldsymbol{\Theta}\boldsymbol{U}_0^\top\boldsymbol{\mathcal{E}}^{(0)}/\sqrt{n_0}\|_{\mathcal{G}}\\
    \leq&\|\hat{\boldsymbol{S}}^{-1/2}-\boldsymbol{S}^{-1/2}\|_{\infty}\|\boldsymbol{\Theta}\boldsymbol{U}_0^\top\boldsymbol{\mathcal{E}}^{(0)}\|_{\max}/\sqrt{n_0}+\sqrt{n_0}\|\hat{\boldsymbol{S}}^{-1/2}\|_{\infty}\|\tilde{\boldsymbol{\beta}}-\boldsymbol{\beta}-\boldsymbol{\Theta}\boldsymbol{U}_0^\top\boldsymbol{\mathcal{E}}^{(0)}/n_0\|_\infty\\
    =&O_P\left(\Delta_{\infty}\sqrt{\log p}+\Delta_{n_0,p}\right),
    \end{align*}
    where $\Delta_{n_0,p}$ is defined in the proof of Theorem \ref{thm3}. The error bound above is $o_P\big((\log p)^{-1/2}\big)$ by assumptions. Let $\boldsymbol{Q}^{stu}$ be a $p$-dimensional Gaussian random vector with mean $\mathbf{0}$ and variance $\sigma_0^2\boldsymbol{S}^{-1/2}\boldsymbol{\Theta}\boldsymbol{S}^{-1/2}$. Similar to the proof of Theorem \ref{thm3}, it can be verified that
    \vspace{-1mm} $$\sup_{t\in\mathbb{R}}\left|\mathbb{P}\left(\|\sqrt{n_0}\hat{\boldsymbol{S}}^{-1/2}(\tilde{\boldsymbol{\beta}}-\boldsymbol{\beta})\|_{\mathcal{G}}\leq t\right)-\mathbb{P}(\|\boldsymbol{Q}^{stu}\|_{\mathcal{G}}\leq t)\right|=o(1).$$
    Following the proof and notation of Theorem \ref{thm4}, we have $\hat{\boldsymbol{Q}}^{e}_{stu}\sim\mathbb{N}(\boldsymbol{0}_p,\hat{\boldsymbol{S}}^{-1/2}\hat{\boldsymbol{\Upsilon}}\hat{\boldsymbol{S}}^{-1/2})$. Since $\hat{\boldsymbol{S}}^{-1/2}$ and $\boldsymbol{S}^{-1/2}$ are bounded diagonal matrices, it can be deduced that 
    \vspace{-1mm}
    \begin{align*}
        \|\hat{\boldsymbol{S}}^{-1/2}&\hat{\boldsymbol{\Upsilon}}\hat{\boldsymbol{S}}^{-1/2}-\boldsymbol{S}^{-1/2}\boldsymbol{\Upsilon}\boldsymbol{S}^{-1/2}\|_{\max}\leq\|\hat{\boldsymbol{S}}^{-1/2}-\boldsymbol{S}^{-1/2}\|_{\max}\|\hat{\boldsymbol{\Upsilon}}\|_{\max}\|\hat{\boldsymbol{S}}^{-1/2}\|_{\max}\\
        &+\|\boldsymbol{S}^{-1/2}\|_{\max}\|\hat{\boldsymbol{\Upsilon}}-\boldsymbol{\Upsilon}\|_{\max}\|\hat{\boldsymbol{S}}^{-1/2}\|_{\max}+\|\boldsymbol{S}^{-1/2}\|_{\max}\|\boldsymbol{\Upsilon}\|_{\max}\|\hat{\boldsymbol{S}}^{-1/2}-\boldsymbol{S}^{-1/2}\|_{\max}\\
        =&O_P(\Delta_1\|\boldsymbol{\Theta}\|_\infty+\Delta_{\infty}+\Delta_{\sigma})=o_P\left((\log p)^{-1}(\log\log p)^{-1}\right)
    \end{align*}
    
    by assumptions. Similar to the proof of Theorem \ref{thm4}, we obtain that 
    $$\sup_{t\in\mathbb{R}}\left|\mathbb{P}\left(\|\hat{\boldsymbol{Q}}_{stu}^{e}\|_{\mathcal{G}}\leq t|\mathcal{D}\right)-\mathbb{P}(\|\boldsymbol{Q}_{stu}\|_{\mathcal{G}}\leq t)\right|=o_P(1).$$
\end{proof}

\renewcommand{\thetheorem}{A\arabic{theorem}}
\setcounter{theorem}{0}
\renewcommand{\thetheorem}{B\arabic{theorem}}
\setcounter{theorem}{0}
\section{Proof of Auxiliary Lemmas}\label{secB}

\begin{proof}[of Lemma \ref{lemma1}]
    We follow the PCA decompositions in \cite{wang2017asymptotics} and provide the concentration bounds needed for simultaneous control over the domains. 

    Define the event $$\mathcal{C}=\bigcap_{k\in\mathcal{L}}\Big\{\|\boldsymbol{F}_k^\top\boldsymbol{F}_k/n_k-\boldsymbol{I}_r\|_2\leq \frac{1}{2},~\|\boldsymbol{U}_k\|_2\leq C\sqrt{p}+C\sqrt{n_k},~\|\boldsymbol{U}_k\boldsymbol{B}_k\|_2\leq C\sqrt{n_kp}\Big\}$$ for some universal constant $C>0$. We show that $\mathbb{P}(\mathcal{C}^c)\leq 2|\mathcal{L}|e^{-cn_{\min}}+|\mathcal{L}|e^{-cp}=\rho_{\mathcal{L}}$.

    First, since $\boldsymbol{u}_i^{(k)}\boldsymbol{u}_i^{(k)\top}$ is the product of two sub-Gaussian random variables, for each $k$, by Theorem 2.8.1 of \cite{vershynin2018high}, we have
    \begin{align*}
        \mathbb{P}\left\{\left|\boldsymbol{a}^\top\left(\sum_{i\in[n_k]}\frac{1}{n_k}\boldsymbol{u}_i^{(k)}\boldsymbol{u}_i^{(k)\top}-\boldsymbol{\Sigma}_{\boldsymbol{u}}^{(k)}\right)\boldsymbol{a}\right|\leq C_0\nu_0^2\left(\sqrt{\frac{u}{n_k}}+\frac{u}{n_k}\right)\right\}\geq 1-2e^{-u}
    \end{align*}
    for some universal constant $C_0>0$, for any fixed $\boldsymbol{a}\in\mathbb{S}^{p-1}$. Choosing an $1/4$-covering of $\mathbb{S}^{p-1}$ (from Corollary 4.2.13 of \cite{vershynin2018high}, the number within this covering is not larger than $9^p$) and taking $u=t+\log 2+p\log 9$, the union bound yields that
    \begin{align*}
        \mathbb{P}\left\{\left\|\sum_{i\in[n_k]}\frac{1}{n_k}\boldsymbol{u}_i^{(k)}\boldsymbol{u}_i^{(k)\top}-\boldsymbol{\Sigma}_{\boldsymbol{u}}^{(k)}\right\|_2\leq C_0'\nu_0^2\left(\sqrt{\frac{p+t}{n_k}}+\frac{p+t}{n_k}\right)\right\}\geq 1-e^{-t}
    \end{align*}
    for $C_0'=2C_0\log 9$. This implies that 
    \begin{align*}
        \mathbb{P}\left\{\left\|\sum_{i\in[n_k]}\frac{1}{n_k}\boldsymbol{u}_i^{(k)}\boldsymbol{u}_i^{(k)\top}\right\|_2\leq \iota^{-1}+C_0'\nu_0^2\left(\sqrt{\frac{p+t}{n_k}}+\frac{p+t}{n_k}\right)\right\}\geq 1-e^{-t}.
    \end{align*}
    This implies that 
    \begin{align*}
        \mathbb{P}\left\{\left\|\boldsymbol{U}_k\right\|_2^2\leq n_k\iota^{-1}+C\nu_0^2\left(\sqrt{n_k(p+t)}+p+t\right)\right\}\geq 1-e^{-t}.
    \end{align*}
    Taking $t=cp$, we obtain
    \begin{align*}
        \mathbb{P}(\|\boldsymbol{U}_k\|_2>C\sqrt{p}+C\sqrt{n_k})\leq e^{-cp}.
    \end{align*}
    for some constants $C,c>0$. Similarly, we have
    \begin{align*}
        \mathbb{P}\left(\|\boldsymbol{F}_k^\top\boldsymbol{F}_k/n_k-\boldsymbol{I}_r\|_2\leq C_0'\nu_0^2\left(\sqrt{\frac{r_k+t}{n_k}}+\frac{r_k+t}{n_k}\right)\right)\geq 1-e^{-t}.
    \end{align*}
    Choosing $t=cn_k$ for sufficiently small constant $c>0$ yields that 
    \begin{align*}
        \mathbb{P}\left(\|\boldsymbol{F}_k^\top\boldsymbol{F}_k/n_k-\boldsymbol{I}_r\|_2>\frac{1}{2}\right)\leq e^{-cn_k}.
    \end{align*}
    Moreover, since $\left\|p^{-1/2}\boldsymbol{B}_k^\top\boldsymbol{u}_i^{(k)}\right\|_{\psi_2}\leq p^{-1/2}\|\boldsymbol{B}_k\|_2\|\boldsymbol{u}_i^{(k)}\|_{\psi_2}\leq C\nu_0$ and $p^{-1}\|\boldsymbol{B}_k^\top\boldsymbol{\Sigma}_{\boldsymbol{u}}^{(k)}\boldsymbol{B}_k\|_2\leq C^2\iota^{-1}$, we similarly obtain that
    \begin{align*}
        \mathbb{P}\left\{\|\boldsymbol{U}_k\boldsymbol{B}_k\|_2^2/(n_kp)\leq C^2\iota^{-1}+C_0'C^2\nu_0^2\left(\sqrt{\frac{r_k+t}{n_k}}+\frac{r_k+t}{n_k}\right)\right\}\geq 1-e^{-t},
    \end{align*}
    which implies that 
    \begin{align*}
        \mathbb{P}\left\{\|\boldsymbol{U}_k\boldsymbol{B}_k\|_2> C\sqrt{n_kp}\right\}\leq e^{-cn_k},
    \end{align*}
    by choosing $t=n_k$. 
    A union bound gives
    \begin{align*}
        \mathbb{P}(\mathcal{C}^c)\leq 2|\mathcal{L}|e^{-cn_{\min}}+|\mathcal{L}|e^{-cp}=\rho_{\mathcal{L}}.
    \end{align*}

    On $\mathcal{C}$, the Weyl's inequality gives, uniformly over all $k\in\mathcal{L}$,
    \begin{align*}
        &\sigma_{r_k}(\boldsymbol{X}^{(k)})\geq \sigma_{r_k}(\boldsymbol{F}_k\boldsymbol{B}_k^\top)-\|\boldsymbol{U}_k\|_2\geq C_1\sqrt{n_kp}-c_1\sqrt p-c_1\sqrt{n_k},\\
        &\|\boldsymbol{X}^{(k)}\|_2\leq \|\boldsymbol{F}_k\boldsymbol{B}_k^\top\|_2+\|\boldsymbol{U}_k\|_2\leq C_2\sqrt{n_kp}+c_2\sqrt p+c_2\sqrt{n_k}.
    \end{align*}
    This implies $cp\leq\lambda_{\min}(\boldsymbol{V}_k)\leq\lambda_{\max}(\boldsymbol{V}_k)\leq Cp$ uniformly over $k\in\mathcal{L}$. The definition of $\boldsymbol{H}_k$ also gives $\|\boldsymbol{H}_k\|_2\leq C$.
    
    The PCA eigen-equation yields
    \begin{align*}
        \boldsymbol{R}_k=\hat{\boldsymbol{F}}_k-\boldsymbol{F}_k\boldsymbol{H}_k^\top=\frac1{n_k}\left(\boldsymbol{F}_k\boldsymbol{B}_k^\top\boldsymbol{U}_k^\top+\boldsymbol{U}_k\boldsymbol{B}_k\boldsymbol{F}_k^\top+\boldsymbol{U}_k\boldsymbol{U}_k^\top\right)\hat{\boldsymbol{F}}_k\boldsymbol{V}_k^{-1}.
    \end{align*}
    Consequently,
    \begin{align*}
        \|\boldsymbol{R}_k\|_2\leq\frac{C}{p\sqrt{n_k}}\left(2\|\boldsymbol{F}_k\|_2\|\boldsymbol{U}_k\boldsymbol{B}_k\|_2+\|\boldsymbol{U}_k\|_2^2\right)\leq C'\left(\sqrt{\frac{n_k}{p}}+\frac{1}{\sqrt{n_k}}\right).
    \end{align*}
    In particular, $\max_k\|\boldsymbol{R}_k\|_2/\sqrt{n_k}=o(1)$  on $\mathcal{E}$. Since $\hat{\boldsymbol{F}}_k=\boldsymbol{F}_k\boldsymbol{H}_k^\top+\boldsymbol{R}_k$,
    \begin{align*}
        \sigma_{\min}(\boldsymbol{F}_k\boldsymbol{H}_k^\top)\geq\sqrt{n_k}-\|\boldsymbol{R}_k\|_2.
    \end{align*}
    Combining this with $\|\boldsymbol{F}_k\|_2\leq C\sqrt{n_k}$ shows that $\|\boldsymbol{H}_k^{-1}\|_2\leq C$ uniformly over $k\in[\mathcal{L}]$.

    Since 
    \begin{align*}
        \boldsymbol{I}_{r_k}=\frac{1}{n_k}\hat{\boldsymbol{F}}_k^\top\hat{\boldsymbol{F}}_k=\frac{1}{n_k}\boldsymbol{H}_k\boldsymbol{F}_k^\top\boldsymbol{R}_k+\frac{1}{n_k}\boldsymbol{R}_k^\top\boldsymbol{F}_k\boldsymbol{H}_k^\top+\frac{1}{n_k}\boldsymbol{R}_k^\top\boldsymbol{R}_k+\frac{1}{n_k}\boldsymbol{H}_k\boldsymbol{F}_k^\top\boldsymbol{F}_k\boldsymbol{H}_k^\top,
    \end{align*}
    we obtain that 
    \begin{align*}
        \|\boldsymbol{H}_k^\top\boldsymbol{H}_k-\boldsymbol{I}_{r_k}\|_{\mathrm{F}}\leq& \|\boldsymbol{H}_k\|_2^2\left\|\boldsymbol{I}_{r_k}-\boldsymbol{F}_k^\top\boldsymbol{F}_k/n_k\right\|_{\mathrm{F}}+\frac{2}{n_k}\|\boldsymbol{H}_k\|_2\|\boldsymbol{F}_k\|_2\|\boldsymbol{R}_k\|_{\mathrm{F}}+\frac{1}{n_k}\|\boldsymbol{R}_k\|_{\mathrm{F}}^2\\
        \leq& C\left\|\boldsymbol{F}_k^\top\boldsymbol{F}_k/n_k-\boldsymbol{I}_{r_k}\right\|_2+C\left(p^{-1/2}+n_k^{-1}\right).
    \end{align*}
    Recall that 
    \begin{align*}
        \mathbb{P}\left(\|\boldsymbol{F}_k^\top\boldsymbol{F}_k/n_k-\boldsymbol{I}_r\|_2\leq C_0'\nu_0^2\left(\sqrt{\frac{r_k+t}{n_k}}+\frac{r_k+t}{n_k}\right)\right)\geq 1-e^{-t},
    \end{align*}
    we have 
    \begin{align*}
        \mathbb{P}(\mathcal{C}_1):=\mathbb{P}\left(\|\boldsymbol{F}_k^\top\boldsymbol{F}_k/n_k-\boldsymbol{I}_r\|_2\leq C\left(\sqrt{\frac{\log(e|\mathcal{L}|)+t}{n_k}}+\frac{\log(e|\mathcal{L}|)+t}{n_k}\right),~\forall k\in\mathcal{L}\right)\geq 1-e^{-t}
    \end{align*}
    for some constant $C>0$. Hence, under the event $\mathcal{C}\cap\mathcal{C}_1$, we have 
    \begin{align*}
        \|\boldsymbol{H}_k^\top\boldsymbol{H}_k-\boldsymbol{I}_{r_k}\|_{\mathrm{F}}\leq C\sqrt{\frac{\log(e|\mathcal{L}|)+t}{n_k}}+Cp^{-1/2}.
    \end{align*}
    if $\log(e|\mathcal{L}|)+t\leq cn_{\min}$. If $\log(e|\mathcal{L}|)+t>n_k$, the same inequality follows from $\|\boldsymbol{H}_k\|_2\leq C$ and $r_k\leq r_*$. Hence this bound holds for every $t\geq 1$.
    
    Next, since $\|u_{ij}^{(k)}\|_{\psi_2}\leq\nu_0$, the random variables $(u_{ij}^{(k)})^2-\mathbb{E}(u_{ij}^{(k)})^2$ are independent across $i$, have zero mean, and satisfy
    \begin{align*}
        \left\|(u_{ij}^{(k)})^2-\mathbb{E}(u_{ij}^{(k)})^2\right\|_{\psi_1}\leq C\nu_0^2.
    \end{align*}
    By Bernstein's inequality, for any $u\geq1$,
    \begin{align*}
        \mathbb{P}\left\{\left|\sum_{i\in[n_k]}\left[(u_{ij}^{(k)})^2-\mathbb{E}(u_{ij}^{(k)})^2\right]\right|>C\nu_0^2(\sqrt{n_ku}+u)\right\}\leq 2e^{-u}.
    \end{align*}
    Since $\mathbb{E}(u_{ij}^{(k)})^2\leq C\nu_0^2$, this implies that
    \begin{align*}
        \mathbb{P}\left\{\|\boldsymbol{U}_{k,\cdot j}\|_2>C(\sqrt{n_k}+\sqrt{u})\right\}\leq 2e^{-u}.
    \end{align*}
    Taking $u=t+\log(2|\mathcal{L}||\mathcal{J}|):=\zeta_t$ and applying a union bound yield, for any fixed nonempty $\mathcal{J}\subseteq[p]$,
    \begin{align*}
        \mathbb{P}(\mathcal{C}_2):=\mathbb{P}\Bigg\{\max_{j\in\mathcal{J}}\|\boldsymbol{U}_{k,\cdot j}\|_2\leq C\left(\sqrt{n_k}+\sqrt{\zeta_t}\right),\forall k\in\mathcal{L}\Bigg\}\geq 1-e^{-t}.
    \end{align*}

    Now we define 
    \begin{align*}
        \boldsymbol{Q}_k:=\hat{\boldsymbol{B}}_k-\boldsymbol{B}_k\boldsymbol{H}_k^{-1}.
    \end{align*}
    Since $\hat{\boldsymbol{F}}_k=\boldsymbol{F}_k\boldsymbol{H}_k^\top+\boldsymbol{R}_k$ and $\hat{\boldsymbol{F}}_k^\top\hat{\boldsymbol{F}}_k/n_k=\boldsymbol{I}_{r_k}$, we have $\boldsymbol{F}_k^\top\hat{\boldsymbol{F}}_k/n_k=\boldsymbol{H}_k^{-1}\left(\boldsymbol{I}_{r_k}-\boldsymbol{R}_k^\top\hat{\boldsymbol{F}}_k/n_k\right)$, which implies that
    \begin{align*}
        \boldsymbol{Q}_k=\boldsymbol{X}^{(k)\top}\hat{\boldsymbol{F}}_k/n_k-\boldsymbol{B}_k\boldsymbol{H}_k^{-1}=-\boldsymbol{B}_k\boldsymbol{H}_k^{-1}\boldsymbol{R}_k^\top\hat{\boldsymbol{F}}_k/n_k+\boldsymbol{U}_k^\top\boldsymbol{F}_k\boldsymbol{H}_k^\top/n_k+\boldsymbol{U}_k^\top\boldsymbol{R}_k/n_k.
    \end{align*}
    Recall that, on $\mathcal{C}$, 
    \begin{align*}
        \|\boldsymbol{H}_k\|_2+\|\boldsymbol{H}_k^{-1}\|_2\leq C,~\|\boldsymbol{R}_k\|_{\mathrm{F}}\leq C\left(\sqrt{\frac{n_k}{p}}+\frac{1}{\sqrt{n_k}}\right),~\|\hat{\boldsymbol{F}}_k\|_2=\sqrt{n_k}
    \end{align*}
    Since $r_k\leq r_*$ and $\|\boldsymbol{B}_k\|_{\max}\leq C$, we also have $\max_{j\in[p]}\|[\boldsymbol{B}_{k}]_{j\cdot}\|_2\leq C'$. Moreover, since $\boldsymbol{u}_i^{(k)}$ and $\boldsymbol{f}_i^{(k)}$ are independent, conditional on $\boldsymbol{F}_k$, for any $j\in\mathcal{J}$ and $h\in[r_k]$, the sub-Gaussian concentration inequality gives
    \begin{align*}
        \mathbb{P}\left\{\left|\sum_{i\in[n_k]}u_{ij}^{(k)}f_{ih}^{(k)}\right|\leq C\nu_0\sqrt{u\sum_{i\in[n_k]}(f_{ih}^{(k)})^2}\bigg|\boldsymbol{F}_k\right\}\geq 1-2e^{-u}
    \end{align*}
    for any $u\geq1$. Recall that, on $\mathcal{C}$, $\|\boldsymbol{F}_k\|_2^2\leq \frac{3}{2}n_k,~\forall k\in\mathcal{L}.$ Therefore,
    \begin{align*}
        &\mathbb{P}\Big(\left\{\left|[\boldsymbol{U}_{k}]_{\cdot j}^\top[\boldsymbol{F}_k]_{\cdot h}\right|> C\sqrt{n_ku}\right\}\cap \mathcal{C}\Big)\leq 2e^{-u}.
    \end{align*}
    
    Taking $u=t+\log(2|\mathcal{L}||\mathcal{J}|r_*)$ and applying a union bound, we obtain
    \begin{align*}
        \mathbb{P}\Bigg(\Bigg\{\max_{j\in\mathcal{J}}\|[\boldsymbol{U}_k]_{\cdot j}^\top\boldsymbol{F}_k\|_2> C\sqrt{n_k\zeta_t},~\text{for some}~k\in\mathcal{L}\Bigg\}\cap \mathcal{C}\Bigg)\leq e^{-t},
    \end{align*}
    where we used $\max_{k\in\mathcal{L}}r_k\leq r_*=O(1)$. Define
    \begin{align*}
        \mathcal{C}_3:=\Bigg\{\max_{j\in\mathcal{J}}\|[\boldsymbol{U}_k]_{\cdot j}^\top\boldsymbol{F}_k\|_2\leq C\sqrt{n_k\zeta_t},~\forall k\in\mathcal{L}\Bigg\}.
    \end{align*}
    Then $\mathbb{P}(\mathcal{C}\cap\mathcal{C}_3^c)\leq e^{-t}.$ Since $\mathbb{P}(\mathcal{C}^c)\leq\rho_{\mathcal{L}}$, we further have
    \begin{align*}
        \mathbb{P}(\mathcal{C}\cap\mathcal{C}_3)\geq1-\mathbb{P}(\mathcal{C}^c)-\mathbb{P}(\mathcal{C}\cap\mathcal{C}_3^c)\geq1-\rho_{\mathcal{L}}-e^{-t}.
    \end{align*}
    Combining these inequalities with the preceding concentration bounds, we obtain
    \begin{align*}
        \max_{j\in\mathcal{J}}\|[\boldsymbol{Q}_k]_{j\cdot}\|_2\leq& \max_{j\in\mathcal{J}}\|[\boldsymbol{B}_k]_{j\cdot}\|_2\|\boldsymbol{H}_k^{-1}\|_2\|\boldsymbol{R}_k\|_2\|\hat{\boldsymbol{F}}_k\|_2/n_k+\max_{j\in\mathcal{J}}\|[\boldsymbol{U}_k]_{\cdot j}^\top\boldsymbol{F}_k\|_2\|\boldsymbol{H}_k\|_2/n_k\\
        &+\max_{j\in\mathcal{J}}\|[\boldsymbol{U}_k]_{\cdot j}\|_2\|\boldsymbol{R}_k\|_2/n_k\\
        \leq& C\left(\sqrt{\frac{1}{p}}+\frac{1}{n_k}+\sqrt{\frac{\zeta_t}{n_k}}\right)\leq C\left(\sqrt{\frac{1}{p}}+\sqrt{\frac{\zeta_t}{n_k}}\right)
    \end{align*}

    Next, using $\boldsymbol{X}^{(k)}=\boldsymbol{F}_k\boldsymbol{B}_k^\top+\boldsymbol{U}_k$ and $\hat{\boldsymbol{B}}_k=\boldsymbol{B}_k\boldsymbol{H}_k^{-1}+\boldsymbol{Q}_k$, we have
    \begin{align*}
        \hat{\boldsymbol{U}}_k-\boldsymbol{U}_k=\boldsymbol{F}_k\boldsymbol{B}_k^\top-\hat{\boldsymbol{F}}_k\hat{\boldsymbol{B}}_k^\top=-\boldsymbol{R}_k\boldsymbol{H}_k^{-\top}\boldsymbol{B}_k^\top-\hat{\boldsymbol{F}}_k\boldsymbol{Q}_k^\top.
    \end{align*}
    Consequently,
    \begin{align*}
        \max_{j\in\mathcal{J}}\|[\hat{\boldsymbol{U}}_k-\boldsymbol{U}_k]_{\cdot j}\|_2^2\leq C\|\boldsymbol{R}_k\|_2^2+Cn_k\max_{j\in\mathcal{J}}\|[\boldsymbol{Q}_{k}]_{j\cdot}\|_2^2\leq C\left(\frac{n_k}{p}+\zeta_t\right)
    \end{align*}

    Finally, we have
    \begin{align*}
        \hat{\boldsymbol{B}}_k-\boldsymbol{B}_k\boldsymbol{H}_k^\top&=\boldsymbol{Q}_k+\boldsymbol{B}_k(\boldsymbol{H}_k^{-1}-\boldsymbol{H}_k^\top),\\
        \boldsymbol{H}_k^{-1}-\boldsymbol{H}_k^\top&=\boldsymbol{H}_k^{-1}(\boldsymbol{I}_{r_k}-\boldsymbol{H}_k\boldsymbol{H}_k^\top).
    \end{align*}
    Since
    \begin{align*}
        \|\boldsymbol{H}_k\boldsymbol{H}_k^\top-\boldsymbol{I}_{r_k}\|_{\mathrm F}= \|\boldsymbol{H}_k^\top\boldsymbol{H}_k-\boldsymbol{I}_{r_k}\|_{\mathrm F},
    \end{align*}
    the rotation bound established above gives, on $\mathcal{C}\cap\mathcal{C}_1$,
    \begin{align*}
        \|\boldsymbol{H}_k^{-1}-\boldsymbol{H}_k^\top\|_{\mathrm F}\leq C\sqrt{\frac{\log(e|\mathcal{L}|)+t}{n_k}}+Cp^{-1/2}.
    \end{align*}
    Thus,
    \begin{align*}
        &\max_{j\in\mathcal{J}}\|[\hat{\boldsymbol{B}}_k-\boldsymbol{B}_k\boldsymbol{H}_k^\top]_{j\cdot}\|_2\leq C\left(\sqrt{\frac{1}{p}}+\sqrt{\frac{\zeta_t}{n_k}}\right).
    \end{align*} 
    Combining the preceding estimates and taking a union bound, we conclude that, for any fixed nonempty $\mathcal{J}\subseteq[p]$, the bounds
    \begin{align*}
        &\max_{j\in\mathcal{J}}\|[\hat{\boldsymbol{B}}_k-\boldsymbol{B}_k\boldsymbol{H}_k^{-1}]_{j\cdot}\|_2^2\leq C\left(\frac{1}{p}+\frac{\zeta_t}{n_k}\right),\\
        &\max_{j\in\mathcal{J}}\|[\hat{\boldsymbol{B}}_k-\boldsymbol{B}_k\boldsymbol{H}_k^\top]_{j\cdot}\|_2^2\leq C\left(\frac{1}{p}+\frac{\zeta_t}{n_k}\right),\\
        &\max_{j\in\mathcal{J}}\|[\hat{\boldsymbol{U}}_k-\boldsymbol{U}_k]_{\cdot j}\|_2^2\leq C\left(\frac{n_k}{p}+\zeta_t\right)
    \end{align*}
    hold simultaneously for all $k\in\mathcal{L}$ with probability at least $1-\rho_{\mathcal{L}}-3e^{-t}$. Taking $\mathcal{J}=[p]$ yields the corresponding uniform bounds over all coordinates.
\end{proof}

\begin{lemma}\label{lemmaA1}
    Under Assumption \ref{assum3}, we have $\|\boldsymbol{\boldsymbol{\delta}}^{\mathcal{A}_\eta}\|_1=\|\boldsymbol{w}^{\mathcal{A}_\eta}-\boldsymbol{\beta}\|_1\leq C_1\eta$.
\end{lemma}
\begin{proof}
    By definition, we have 
       $\sum_{k\in\{0\}\cup\mathcal{A}_\eta}\alpha_k\boldsymbol{\Sigma}^{(k)}_{\boldsymbol{u}}(\boldsymbol{w}^{(k)}-\boldsymbol{w}^{\mathcal{A}_\eta})=\boldsymbol{0}_p,$
    which indicates that
    \begin{align*}
       \sum_{k\in\mathcal{A}_\eta}\alpha_k\boldsymbol{\Sigma}^{(k)}_{\boldsymbol{u}}(\boldsymbol{w}^{(k)}-\boldsymbol{\beta})=\sum_{k\in\{0\}\cup\mathcal{A}_\eta}\alpha_k\boldsymbol{\Sigma}^{(k)}_{\boldsymbol{u}}(\boldsymbol{w}^{\mathcal{A}_\eta}-\boldsymbol{\beta}).
    \end{align*}
    It follows that $\|\boldsymbol{w}^{\mathcal{A}_\eta}-\boldsymbol{\beta}\|_1\leq\sum_{k\in\mathcal{A}_\eta}\alpha_k\|(\sum_{k'\in\{0\}\cup\mathcal{A}_\eta}\alpha_{k'}\boldsymbol{\Sigma}^{(k')}_{\boldsymbol{u}})^{-1}\boldsymbol{\Sigma}^{(k)}_{\boldsymbol{u}}\|_1\|\boldsymbol{w}^{(k)}-\boldsymbol{\beta}\|_1\leq C_1\eta$.
\end{proof}

\begin{lemma}\label{lemmaA2}
    Suppose that Assumption \ref{assum1} holds for $k\in\{0\}\cup\mathcal{A}_\eta$, then there exist some positive constants $\kappa_1,\kappa_2,C_3$ and $C_4$ such that, for any $\boldsymbol{\Delta}\in\mathbb{R}^{p}$, $$D(\boldsymbol{\Delta})\geq \kappa_1\|\boldsymbol{\Delta}\|_2^2-\kappa_2\frac{\log p}{n_{\mathcal{A}_\eta}+n_0}\|\boldsymbol{\Delta}\|_1^2$$
    with probability at least $1-C_3\exp(-C_4(n_{\mathcal{A}_\eta}+n_0))$.
\end{lemma}

\begin{proof}
    First, following Lemma 2 of \cite{tian2023transfer}, 
    there exist some positive constants $\kappa_1',\kappa_2',C_3$ and $C_4$ such that $$\frac{D(\boldsymbol{\Delta})}{\|\boldsymbol{\Delta}\|_2^2}\geq \kappa_1'\left(1-\kappa_2'\sqrt{\frac{\log p}{n_{\mathcal{A}_\eta}+n_0}}\frac{\|\boldsymbol{\Delta}\|_1}{\|\boldsymbol{\Delta}\|_2}\right),\quad\text{for all }\|\boldsymbol{\Delta}\|_2\leq 1$$
    with probability at least $1-C_3\exp(-C_4(n_{\mathcal{A}_\eta}+n_0))$. By Young's inequality, with the same probability, for all $\|\boldsymbol{\Delta}\|_2\leq 1$, we obtain that 
    \begin{align*}
        D(\boldsymbol{\Delta})\geq& \kappa_1'\|\boldsymbol{\Delta}\|_2^2-\kappa_1'\kappa_2'\sqrt{\frac{\log p}{n_{\mathcal{A}_\eta}+n_0}}\|\boldsymbol{\Delta}\|_1\|\boldsymbol{\Delta}\|_2\\
        \geq& \frac{\kappa_1'}{2}\|\boldsymbol{\Delta}\|_2^2-\frac{\kappa_1'\kappa_2'^2}{2}\frac{\log p}{n_{\mathcal{A}_\eta}+n_0}\|\boldsymbol{\Delta}\|_1^2.
    \end{align*}
    Since both sides are quadratic with respect to $\boldsymbol{\Delta}$, the above inequality holds for every $\boldsymbol{\Delta}\in\mathbb{R}^{p}$. This completes the proof with $\kappa_1=\kappa_1'/2$ and $\kappa_2=\kappa_1'\kappa_2'^2/2$.
    
\end{proof}
\begin{lemma}\label{lemmaA3}
    Suppose that Assumption \ref{assum1} holds for $k\in\{0\}\cup\mathcal{A}_\eta$, then there exist some positive constants $\kappa_1,\kappa_2,C_3$ and $C_4$ such that, for any $\boldsymbol{\Delta}\in\mathbb{R}^{p}$,  $$D^{(0)}(\boldsymbol{\Delta})\geq \kappa_1\|\boldsymbol{\Delta}\|_2^2-\kappa_2\frac{\log p}{n_0}\|\boldsymbol{\Delta}\|_1^2$$
    with probability at least $1-C_3\exp(-C_4n_0)$, where $D^{(0)}(\boldsymbol{\Delta})=\boldsymbol{\Delta}^\top\left(\boldsymbol{U}_0^\top\boldsymbol{U}_0/n_0\right)\boldsymbol{\Delta}$.
\end{lemma}
\begin{proof}
    First, following Proposition 2 of \cite{negahban2009unified}, we obtain that $$\frac{D^{(0)}(\boldsymbol{\Delta})}{\|\boldsymbol{\Delta}\|_2^2}\geq \kappa_1'\left(1-\kappa_2'\sqrt{\frac{\log p}{n_0}}\frac{\|\boldsymbol{\Delta}\|_1}{\|\boldsymbol{\Delta}\|_2}\right),\quad\text{for all }\|\boldsymbol{\Delta}\|_2\leq 1$$
    holds with probability at least $1-C_3\exp(-C_4(n_0))$ for some positive constants $\kappa_1',\kappa_2',C_3$ and $C_4$.
    Therefore, similar to the proof of Lemma \ref{lemmaA2}, we complete the proof.
\end{proof}

\begin{lemma}\label{lemmaA3-upper bound}
    Suppose that Assumption \ref{assum1} holds for $k\in\{0\}\cup\mathcal{A}_\eta$, then there exist some positive constants $C,c$ such that, for any $\boldsymbol{\Delta}\in\mathbb{R}^{p}$,  $$D^{(0)}(\boldsymbol{\Delta})\leq C\|\boldsymbol{\Delta}\|_2^2+C\frac{\log p}{n_0}\|\boldsymbol{\Delta}\|_1^2$$
    with probability at least $1-2\exp(-cn_0)$, where $D^{(0)}(\boldsymbol{\Delta})=\boldsymbol{\Delta}^\top\left(\boldsymbol{U}_0^\top\boldsymbol{U}_0/n_0\right)\boldsymbol{\Delta}$.
\end{lemma}

\begin{proof}
    By Lemma 15 of \cite{loh2011high}, let $m=\lfloor c_0n_0/\log p\rfloor$ for some sufficiently small $c_0$, we have 
    \begin{align*}
        \mathbb{P}\left\{\sup_{\substack{\boldsymbol{v}\in\mathbb{R}^{p}:\|\boldsymbol{v}\|_2\leq 1\\\|\boldsymbol{v}\|_0\leq 2m}}|\boldsymbol{v}^\top (\boldsymbol{U}^{(0)\top}\boldsymbol{U}^{(0)}/n_0-\boldsymbol{\Sigma}_{\boldsymbol{u}}^{(0)})\boldsymbol{v}|>t\right\}\leq 2\exp\{-C_1n_0\min\{t^2/v_0^4,t/v_0^2\}+C_2m\log p\}.
    \end{align*}
    Choosing $t=c'$ for some sufficiently large constant $c'>0$ yields that, for some constant $c>0$, 
    \begin{align*}
        \mathbb{P}\left\{\sup_{\substack{\boldsymbol{v}\in\mathbb{R}^{p}:\|\boldsymbol{v}\|_2\leq 1\\\|\boldsymbol{v}\|_0\leq 2m}}|\boldsymbol{v}^\top (\boldsymbol{U}^{(0)\top}\boldsymbol{U}^{(0)}/n_0-\boldsymbol{\Sigma}_{\boldsymbol{u}}^{(0)})\boldsymbol{v}|>c'\right\}\leq 2\exp\{-cn_0\}.
    \end{align*}
    Consider the event $$\mathcal{E}=\left\{\sup_{\substack{\boldsymbol{v}\in\mathbb{R}^{p}:\|\boldsymbol{v}\|_2\leq 1\\\|\boldsymbol{v}\|_0\leq 2m}}|\boldsymbol{v}^\top (\boldsymbol{U}^{(0)\top}\boldsymbol{U}^{(0)}/n_0-\boldsymbol{\Sigma}_{\boldsymbol{u}}^{(0)})\boldsymbol{v}|\leq c'\right\}.$$
    Applying Lemma 12 of \cite{loh2011high}, for any $\boldsymbol{\Delta}\in\mathbb{R}^{p}$, we have $$|\boldsymbol{\Delta}^\top (\boldsymbol{U}^{(0)\top}\boldsymbol{U}^{(0)}/n_0-\boldsymbol{\Sigma}_{\boldsymbol{u}}^{(0)})\boldsymbol{\Delta}|\leq 27c'(\|\boldsymbol{\Delta}\|_2^2+\|\boldsymbol{\Delta}\|_1^2/m)$$
    under the event $\mathcal{E}$. Moreover, by Assumption \ref{assum1}, we have $|\boldsymbol{\Delta}^\top \boldsymbol{\Sigma}_{\boldsymbol{u}}^{(0)}\boldsymbol{\Delta}|\leq \iota^{-1}\|\boldsymbol{\Delta}\|_2^2$. This completes the proof by using the triangle inequality.
\end{proof}

\begin{lemma}\label{lemmaA4}
    Suppose that Assumptions \ref{assum1}-\ref{assum2} hold for $k\in\{0\}\cup\mathcal{A}_\eta$, and $\log p+t\lesssim n_{\mathcal{A}_\eta}+n_0$ for some $t>0$. Let $\zeta_{t,\{0\}\cup\mathcal{A}_\eta}=t+\log(2p(|\mathcal{A}_\eta|+1))$ and $$\xi_{t,\mathcal{A}_\eta,k}=\sqrt{\left(\frac{n_k}{p}+\zeta_{t,\{0\}\cup\mathcal{A}_\eta}\right)\big(\log(e(|\mathcal{A}_\eta|+1))+t\big)}.$$ Under the events with $\mathcal{L}=\{0\}\cup\mathcal{A}_\eta$ in Lemma \ref{lemmaA6}, we have that, with probability at least $1-2e^{-t}$, 
    \begin{align*}
        \|\nabla\hat{L}(\boldsymbol{w}^{\mathcal{A}_\eta})\|_\infty\leq& C\left(\frac{\sum_{k\in\{0\}\cup\mathcal{A}_\eta}\big(\mathcal{V}_{n_k,p}(t,\{0\}\cup\mathcal{A}_\eta)\|\varphi^{(k)}\|_2+\xi_{t,\mathcal{A}_\eta,k}\big)}{n_{\mathcal{A}_\eta}+n_0}+\frac{\eta}{p}+\frac{\eta(|\mathcal{A}_\eta|+1)\zeta_{t,\{0\}\cup\mathcal{A}_\eta}}{n_{\mathcal{A}_\eta}+n_0}\right)\\
        &+C(1\vee \eta)\sqrt{\frac{\log p+t}{n_{\mathcal{A}_\eta}+n_0}}
    \end{align*}
    for some constant $C>0$.
\end{lemma}

\begin{proof}
    Note that 
\begin{align*}
        -\nabla \hat{L}(\boldsymbol{w}^{\mathcal{A}_\eta})=&\frac{1}{n_{\mathcal{A}_\eta}+n_0}\sum_{k\in\{0\}\cup\mathcal{A}_\eta}\hat{\boldsymbol{U}}_k^\top(\tilde{\boldsymbol{Y}}^{(k)}-\hat{\boldsymbol{U}}_k\boldsymbol{w}^{(k)})\\
        &+\frac{1}{n_{\mathcal{A}_\eta}+n_0}\sum_{k\in\{0\}\cup\mathcal{A}_\eta}\hat{\boldsymbol{U}}_k^\top\hat{\boldsymbol{U}}_k(\boldsymbol{w}^{(k)}-\boldsymbol{w}^{\mathcal{A}_\eta})
        :=E_{1}+E_{2}.
\end{align*}
For $E_1$, since $\hat{\boldsymbol{U}}_k^\top(\tilde{\boldsymbol{Y}}^{(k)}-\hat{\boldsymbol{U}}_k\boldsymbol{w}^{(k)})=\hat{\boldsymbol{U}}_k^\top(\boldsymbol{Y}^{(k)}-\boldsymbol{X}^{(k)}\boldsymbol{w}^{(k)})=\hat{\boldsymbol{U}}_k^\top\boldsymbol{\mathcal{E}}^{(k)}+\hat{\boldsymbol{U}}_k^\top\boldsymbol{F}_k\boldsymbol{\varphi}^{(k)}$, where 
$\|\hat{\boldsymbol{U}}_k^\top\boldsymbol{F}_k\boldsymbol{\varphi}^{(k)}\|_\infty\leq C\mathcal{V}_{n_k,p}(t)\|\boldsymbol{\varphi}^{(k)}\|_2$, and $$\|(\hat{\boldsymbol{U}}_k-\boldsymbol{U}_k)^\top\boldsymbol{\mathcal{E}}^{(k)}\|_\infty\leq C\sqrt{\left(\frac{n_k}{p}+\zeta_{t,\{0\}\cup\mathcal{A}_\eta}\right)\big(\log(e(|\mathcal{A}_\eta|+1))+t\big)}=C\xi_{t,\mathcal{A}_\eta,k}$$ hold with probability at least $1-2|\mathcal{L}|e^{-cn_{\min}}-|\mathcal{L}|e^{-cp}-5e^{-t}$ by Lemma \ref{lemmaA6}, we have 
\begin{align*}
    \|E_1\|_\infty\leq&\frac{1}{n_{\mathcal{A}_\eta}+n_0}\Bigg\{\sum_{k\in\{0\}\cup\mathcal{A}_\eta}\left(\|\hat{\boldsymbol{U}}_k^\top\boldsymbol{F}_k\boldsymbol{\varphi}^{(k)}\|_\infty+\|(\hat{\boldsymbol{U}}_k-\boldsymbol{U}_k)^\top\boldsymbol{\mathcal{E}}^{(k)}\|_\infty\right)\\
    &~~~~~~~~~~~~~~~~~~~~~~~~~~~~~~~~~~~~~~~~~~~~~~~~~~~~~~~~~~~~~~~~~~~~+\Bigg\|\sum_{k\in\{0\}\cup\mathcal{A}_\eta}\boldsymbol{U}_k^\top\boldsymbol{\mathcal{E}}^{(k)}\Bigg\|_\infty\Bigg\}\\
    \leq&C\left(\frac{\sum_{k\in\{0\}\cup\mathcal{A}_\eta}\big(\mathcal{V}_{n_k,p}(t,\{0\}\cup\mathcal{A}_\eta)\|\varphi^{(k)}\|_2+\xi_{t,\mathcal{A}_\eta,k}\big)}{n_{\mathcal{A}_\eta}+n_0}\right)+\frac{1}{n_{\mathcal{A}_\eta}+n_0}\left\|\sum_{k\in\{0\}\cup\mathcal{A}_\eta}\boldsymbol{U}_k^\top\boldsymbol{\mathcal{E}}^{(k)}\right\|_\infty.
\end{align*}
For the second term, by Bernstein's inequality \citep[Theorem 2.8.1]{vershynin2018high}, for some constant $C>0$, we have 
\begin{align*}
    \mathbb{P}\left\{\frac{1}{n_{\mathcal{A}_\eta}+n_0}\left(\sum_{k\in\{0\}\cup\mathcal{A}_\eta}\boldsymbol{U}_k^\top\boldsymbol{\mathcal{E}}^{(k)}\right)_j\leq C\left(\sqrt{\frac{u}{n_{\mathcal{A}_\eta}+n_0}}+\frac{u}{n_{\mathcal{A}_\eta}+n_0}\right)\right\}\geq 1-2e^{-u}.
\end{align*}
Taking $u=t+\log (2p)$ yields that 
\begin{align*}
    \mathbb{P}\left\{\frac{1}{n_{\mathcal{A}_\eta}+n_0}\left\|\sum_{k\in\{0\}\cup\mathcal{A}_\eta}\boldsymbol{U}_k^\top\boldsymbol{\mathcal{E}}^{(k)}\right\|_\infty\leq C\left(\sqrt{\frac{t+\log (2p)}{n_{\mathcal{A}_\eta}+n_0}}+\frac{t+\log (2p)}{n_{\mathcal{A}_\eta}+n_0}\right)\right\}\geq 1-e^{-t}.
\end{align*}
for some constant $C>0$.

For $E_2$, according to Lemma \ref{lemmaA1} and the definition of $\mathcal{A}_\eta$, we have $\|\boldsymbol{w}^{(k)}-\boldsymbol{w}^{\mathcal{A}_\eta}\|_1\leq\|\boldsymbol{w}^{(k)}-\boldsymbol{\beta}\|_1+\|\boldsymbol{\beta}-\boldsymbol{w}^{\mathcal{A}_\eta}\|_1\leq(C_1+1)\eta$. This equation, together with Lemma \ref{lemmaA6} implies that
\begin{align*}
    &\left\|\frac{1}{n_{\mathcal{A}_\eta}+n_0}\sum_{k\in\{0\}\cup\mathcal{A}_\eta}(\hat{\boldsymbol{U}}_k^\top\hat{\boldsymbol{U}}_k-\boldsymbol{U}_k^\top\boldsymbol{U}_k)(\boldsymbol{w}^{(k)}-\boldsymbol{w}^{\mathcal{A}_\eta})\right\|_\infty\\
    \leq&\frac{1}{(n_{\mathcal{A}_\eta}+n_0)}\sum_{k}\|\hat{\boldsymbol{U}}_k^\top\hat{\boldsymbol{U}}_k-\boldsymbol{U}_k^\top\boldsymbol{U}_k\|_{\max}\|\boldsymbol{w}^{(k)}-\boldsymbol{w}^{\mathcal{A}_\eta}\|_1\leq C\eta\left(\frac{1}{p}+\frac{(|\mathcal{A}_\eta|+1)\zeta_{t,\{0\}\cup\mathcal{A}_\eta}}{n_{\mathcal{A}_\eta}+n_0}\right)
\end{align*}
Moreover, by the definition of $\boldsymbol{w}^{\mathcal{A}_\eta}$, we have $\sum_{k\in\{0\}\cup\mathcal{A}_\eta}\boldsymbol{U}_k^\top\boldsymbol{U}_k(\boldsymbol{w}^{(k)}-\boldsymbol{w}^{\mathcal{A}_\eta})/(n_{\mathcal{A}_\eta}+n_0)=\sum_{k\in\{0\}\cup\mathcal{A}_\eta}\sum_{i\in[n_k]}(\boldsymbol{u}_i^{(k)}\boldsymbol{u}_i^{(k)\top}-\boldsymbol{\Sigma}_{\boldsymbol{u}}^{(k)})(\boldsymbol{w}^{(k)}-\boldsymbol{w}^{\mathcal{A}_\eta})/(n_{\mathcal{A}_\eta}+n_0)$ is the sum of independent mean-zero random vectors, and $$\|\boldsymbol{u}_i^{(k)}\boldsymbol{u}_i^{(k)\top}(\boldsymbol{w}^{(k)}-\boldsymbol{w}^{\mathcal{A}_\eta})\|_{\psi_1}\leq\|\boldsymbol{u}_i^{(k)}\|_{\psi_2}^2\|\boldsymbol{w}^{(k)}-\boldsymbol{w}^{\mathcal{A}_\eta}\|_{1}\lesssim\eta.$$
It follows from Bernstein's inequality that
$$\mathbb{P}\Bigg(\Bigg\|\frac{1}{n_{\mathcal{A}_\eta}+n_0}\sum_{k\in\{0\}\cup\mathcal{A}_\eta}\boldsymbol{U}_k^\top\boldsymbol{U}_k(\boldsymbol{w}^{(k)}-\boldsymbol{w}^{\mathcal{A}_\eta})\Bigg\|_\infty> C\eta\left(\sqrt{\frac{\log (2p)+t}{n_{\mathcal{A}_\eta}+n_0}}+\frac{\log (2p)+t}{n_{\mathcal{A}_\eta}+n_0}\right)\Bigg)\leq e^{-t}$$ for some constants $C>0$. Hence, when $\log p+t\lesssim n_{\mathcal{A}_\eta}+n_0$, we have 
\begin{align*}
    \|\nabla\hat{L}(\boldsymbol{w}^{\mathcal{A}_\eta})\|_\infty\leq& C\left(\frac{\sum_{k\in\{0\}\cup\mathcal{A}_\eta}\big(\mathcal{V}_{n_k,p}(t,\{0\}\cup\mathcal{A}_\eta)\|\varphi^{(k)}\|_2+\xi_{t,\mathcal{A}_\eta,k}\big)}{n_{\mathcal{A}_\eta}+n_0}+\frac{\eta}{p}+\frac{\eta(|\mathcal{A}_\eta|+1)\zeta_{t,\{0\}\cup\mathcal{A}_\eta}}{n_{\mathcal{A}_\eta}+n_0}\right)\\
    &+C(1\vee \eta)\sqrt{\frac{\log p+t}{n_{\mathcal{A}_\eta}+n_0}}
\end{align*}
holds with probability at least $1-2e^{-t}$.
\end{proof}

\begin{lemma}\label{lemmaA5}
    Suppose that Assumptions \ref{assum1}-\ref{assum2} hold for $k=0$, and $\log p+t\lesssim n_0$ for some $t>0$. Let $\zeta_{t,\{0\}}=t+\log(2p)$ and $$\xi_{t,\{0\}}=\sqrt{\left(\frac{n_0}{p}+t+\log(2p)\right)(t+1)}.$$ Under the events with $\mathcal{L}=\{0\}$ in Lemma \ref{lemmaA6}, we have that, with probability at least $1-e^{-t}$, 
    \begin{align*}
        \|\nabla\hat{L}^{(0)}(\boldsymbol{\beta})\|_\infty\leq& C\frac{\mathcal{V}_{n_0,p}(t,\{0\})\|\varphi^{(0)}\|_2+\xi_{t,\{0\}}}{n_0}+C\sqrt{\frac{\log p+t}{n_0}}
    \end{align*}
    for some constant $C>0$.
\end{lemma}
\begin{proof}
    Analogous to the proof of Lemma \ref{lemmaA4}, we have
    \begin{align*}
        \|\nabla\hat{L}^{(0)}(\boldsymbol{\beta})\|_\infty\leq&\frac{1}{n_0}\left\{\|\hat{\boldsymbol{U}}_0^\top\boldsymbol{F}_0\boldsymbol{\varphi}^{(0)}\|_\infty+\|(\hat{\boldsymbol{U}}_0-\boldsymbol{U}_0)^\top\boldsymbol{\mathcal{E}}^{(0)}\|_\infty+\left\|\boldsymbol{U}_0^\top\boldsymbol{\mathcal{E}}^{(0)}\right\|_\infty\right\}\\
        \leq& C\frac{\mathcal{V}_{n_0,p}(t,\{0\})\|\varphi^{(0)}\|_2+\xi_{t,\{0\}}}{n_0}+C\sqrt{\frac{\log p+t}{n_0}}
    \end{align*}
    with probability at least $1-e^{-t}$.
\end{proof}

\begin{lemma}\label{lemmaA6}
    Suppose that the conditions of Lemma \ref{lemma1} and Assumption \ref{assum2} hold for some nonempty $\mathcal{L}\subseteq\{0\}\cup[K]$. Define
    \begin{align*}
        \mathcal{V}_{n_k,p}(t,\mathcal{L})=\frac{n_k}{p}+\sqrt{\frac{\zeta_{t,\mathcal{L}}}{n_k}}+\sqrt{\frac{n_k\zeta_{t,\mathcal{L}}}{p}},~~\text{where}~~\zeta_{t,\mathcal{L}}=t+\log(2p|\mathcal{L}|).
    \end{align*}
    For every $t>0$, with probability at least $1-2|\mathcal{L}|e^{-cn_{\min}}-|\mathcal{L}|e^{-cp}-5e^{-t}$, simultaneously for $k\in\mathcal{L}$,
    \begin{align*}
        &\sup_{\boldsymbol{\phi}\in\mathbb{R}^{r_k}:\|\boldsymbol{\phi}\|_2=1}\left\|\hat{\boldsymbol{U}}_k^\top\boldsymbol{F}_k\boldsymbol{\phi}\right\|_\infty\leq C\mathcal{V}_{n_k,p}(t,\mathcal{L}),\\
        &\left\|(\hat{\boldsymbol{U}}_k-\boldsymbol{U}_k)^\top\boldsymbol{\mathcal{E}}^{(k)}\right\|_\infty\leq C\sqrt{\left(\frac{n_k}{p}+\zeta_{t,\mathcal{L}}\right)(\log(e|\mathcal{L}|)+t)},\\
        &\left\|\hat{\boldsymbol{U}}_k^\top\hat{\boldsymbol{U}}_k-\boldsymbol{U}_k^\top\boldsymbol{U}_k\right\|_{\max}\leq C\left(\frac{n_k}{p}+\zeta_{t,\mathcal{L}}\right).
    \end{align*}
\end{lemma}

\begin{proof}[of Lemma \ref{lemmaA6}]
We follow the decompositions in the proofs of Lemma C.3, Lemma C.4, and Lemma C.6 of \cite{fan2024latent}, replacing their single-domain stochastic bounds by the simultaneous concentration bounds below.

Write $\boldsymbol{R}_k=\hat{\boldsymbol{F}}_k-\boldsymbol{F}_k\boldsymbol{H}_k^\top$. We use the event established in Lemma \ref{lemma1} and its proof with $\mathcal{J}=[p]$.

\textbf{Step 1: Bounds for $\sup_{\|\boldsymbol{\phi}\|_2=1}\left\|\hat{\boldsymbol{U}}_k^\top\boldsymbol{F}_k\boldsymbol{\phi}\right\|_\infty$.} This corresponds to Lemma C.3 of \cite{fan2024latent}. First, by conditional sub-Gaussian concentration and Bernstein's inequality, we have
\begin{align*}
    &\max_{j\in[p]}\|\boldsymbol{e}_j^\top\boldsymbol{U}_k^\top\boldsymbol{F}_k\|_2\leq C\sqrt{n_k\zeta_{t,\mathcal{L}}},\\
    &\max_{j\in[p]}\|\boldsymbol{e}_j^\top\boldsymbol{U}_k^\top\boldsymbol{U}_k\boldsymbol{B}_k\|_2\leq C\left(n_k+\sqrt{n_kp\zeta_t}+\sqrt{p}\zeta_{t,\mathcal{L}}\right),\\
    &\max_{j\in[p]}\|\boldsymbol{e}_j^\top\boldsymbol{U}_k^\top\boldsymbol{U}_k\boldsymbol{U}_k^\top\boldsymbol{F}_k\|_2\leq
    C(n_k+p)(\sqrt{n_k}+\sqrt{\zeta_{t,\mathcal{L}}})(\sqrt{\zeta_{t,\mathcal{L}}}\wedge\sqrt{n_k}),
\end{align*}
simultaneously over all domains and masks, with additional failure probability at most $Ce^{-t}$. Specifically, the first inequality follows from the same argument used for $\mathcal{C}_3$ in the proof of Lemma \ref{lemma1}. For the second bound, by Assumption \ref{assum1}, we have
\begin{align*}
    &\left|\mathbb{E}\left[u_{ij}^{(k)}\boldsymbol{u}_i^{(k)\top}[\boldsymbol{B}_k]_{\cdot h}\right]\right|=\left|\boldsymbol{e}_j^\top\boldsymbol{\Sigma}_{\boldsymbol{u}}^{(k)}[\boldsymbol{B}_k]_{\cdot h}\right|\leq\|\boldsymbol{\Sigma}_{\boldsymbol{u}}^{(k)}\|_1\|\boldsymbol{B}_k\|_{\max}\leq C.
\end{align*}
\begin{align*}
    &\left\|u_{ij}^{(k)}\boldsymbol{u}_i^{(k)\top}[\boldsymbol{B}_k]_{\cdot h}\right\|_{\psi_1}\leq C\|u_{ij}^{(k)}\|_{\psi_2}\|\boldsymbol{u}_i^{(k)\top}[\boldsymbol{B}_k]_{\cdot h}\|_{\psi_2}\leq C\nu_0^2\sqrt p.
\end{align*}
Thus, Bernstein's inequality \citep[Theorem 2.8.1]{vershynin2018high} gives,
for any $u\geq 0$,
\begin{align*}
    &\mathbb{P}\Bigg\{\Bigg|\sum_{i\in[n_k]}\Big(u_{ij}^{(k)}\boldsymbol{u}_i^{(k)\top}[\boldsymbol{B}_k]_{\cdot h}-\mathbb{E}\left[u_{ij}^{(k)}\boldsymbol{u}_i^{(k)\top}[\boldsymbol{B}_k]_{\cdot h}\right]\Big)\Bigg|>C\nu_0^2\sqrt p(\sqrt{n_ku}+u)\Bigg\}\leq2e^{-u}.
\end{align*}
Taking $u=\zeta_{t,\mathcal{L}}$ and applying a union bound over $k\in\mathcal{L}$, $j\in[p]$ and $h\in[r_k]$, we obtain
\begin{align*}
    \max_{j\in[p]}\|[\boldsymbol{U}_k]_{\cdot j}^\top\boldsymbol{U}_k\boldsymbol{B}_k\|_2\leq C\left(n_k+\sqrt{n_kp\zeta_{t,\mathcal{L}}}+\sqrt p\,\zeta_{t,\mathcal{L}}\right)
\end{align*}
simultaneously for all $k\in\mathcal{L}$,
with probability at least $1-e^{-t}$. For the third bound, conditional on $\boldsymbol{U}_k$, for any $j\in[p]$ and $h\in[r_k]$, the sub-Gaussian concentration inequality gives
\begin{align*}
    &\mathbb{P}\Bigg\{\left|[\boldsymbol{U}_k]_{\cdot j}^\top\boldsymbol{U}_k\boldsymbol{U}_k^\top[\boldsymbol{F}_k]_{\cdot h}\right|>C\nu_0\sqrt u\left\|\boldsymbol{U}_k\boldsymbol{U}_k^\top[\boldsymbol{U}_k]_{\cdot j}\right\|_2\bigg|\boldsymbol{U}_k\Bigg\}\leq2e^{-u}.
\end{align*}
On the event established in the proof of Lemma \ref{lemma1}, we have
\begin{align*}
    \left\|\boldsymbol{U}_k\boldsymbol{U}_k^\top[\boldsymbol{U}_k]_{\cdot j}\right\|_2\leq\|\boldsymbol{U}_k\|_2^2\|[\boldsymbol{U}_k]_{\cdot j}\|_2\leq C(n_k+p)(\sqrt{n_k}+\sqrt{\zeta_{t,\mathcal{L}}}).
\end{align*}
Therefore, taking $u=\zeta_{t,\mathcal{L}}$ and applying a union bound gives
\begin{align*}
    \max_{j\in[p]}\|[\boldsymbol{U}_k]_{\cdot j}^\top\boldsymbol{U}_k\boldsymbol{U}_k^\top\boldsymbol{F}_k\|_2\leq C(n_k+p)(\sqrt{n_k}+\sqrt{\zeta_{t,\mathcal{L}}})\sqrt{\zeta_{t,\mathcal{L}}}
\end{align*}
on the event in the proof of Lemma \ref{lemma1}, except on an additional event of probability at most $e^{-t}$.

Following the decomposition of $\Delta^\circ$ in the proof of Lemma C.3 of \cite{fan2024latent}, since $\|\boldsymbol{B}_k^\top\boldsymbol{U}_k^\top\hat{\boldsymbol{F}}_k\|_2\leq\|\boldsymbol{U}_k\boldsymbol{B}_k\|_2\|\hat{\boldsymbol{F}}_k\|_2\leq Cn_k\sqrt p,$ on the events established in the proof of Lemma \ref{lemma1}, we obtain
\begin{align*}
    \max_{j\in[p]}\|\boldsymbol{e}_j^\top\boldsymbol{U}_k^\top\boldsymbol{R}_k\|_2\leq& \frac{1}{n_k}\max_{j\in[p]}\|\boldsymbol{e}_j^\top\boldsymbol{U}_k^\top\boldsymbol{F}_k\|_2\|\boldsymbol{B}_k^\top\boldsymbol{U}_k^\top\hat{\boldsymbol{F}}_k\|_2\|\boldsymbol{V}_k^{-1}\|_2\|\boldsymbol{H}^{-\top}\|_2\\
    &+\frac{1}{n_k}\max_{j\in[p]}\|\boldsymbol{e}_j^\top\boldsymbol{U}_k^\top\boldsymbol{U}_k\boldsymbol{B}_k\|_2\|\boldsymbol{F}_k^\top\hat{\boldsymbol{F}}_k\|_2\|\boldsymbol{V}_k^{-1}\|_2\|\boldsymbol{H}^{-\top}\|_2\\
    &+\frac{1}{n_k}\max_{j\in[p]}\|\boldsymbol{e}_j^\top\boldsymbol{U}_k^\top\boldsymbol{U}_k\boldsymbol{U}_k^\top\boldsymbol{F}_k\|_2\|\boldsymbol{V}_k^{-1}\|_2\\
    &+\frac{1}{n_k}\max_{j\in[p]}\|\boldsymbol{e}_j^\top\boldsymbol{U}_k^\top\|_2\|\boldsymbol{U}_k\|_2^2\|\boldsymbol{R}_k\|_2\|\boldsymbol{V}_k^{-1}\|_2\|\boldsymbol{H}^{-\top}\|_2\\
    \leq&C\Bigg\{\sqrt{\frac{n_k\zeta_{t,\mathcal{L}}}{p}}+\frac{n_k}{p}+\frac{\zeta_{t,\mathcal{L}}}{\sqrt p}+\left(1+\frac{n_k}{p}\right)\left(1+\sqrt{\frac{\zeta_{t,\mathcal{L}}}{n_k}}\right)\left(\sqrt{\frac{1}{p}}+\frac{1}{n_k}\right)\Bigg\}\\
    \leq&C\mathcal{V}_{n_k,p}(t,\mathcal{L}),
\end{align*}
where we use $\zeta_{t,\mathcal{L}}\leq n_{\min}$ in the last step. The four contributions correspond to
$\Delta_{1}^\circ$, $\Delta_{2}^\circ$,
$\Delta_{31}^\circ$, and $\Delta_{32}^\circ$ in the proof of Lemma C.3 of \cite{fan2024latent}. 

The remaining term, corresponding to $\Delta_\diamond$, satisfies
\begin{align*}
    \max_{j\leq p}\|\boldsymbol{e}_j^\top(\hat{\boldsymbol{U}}_k-\boldsymbol{U}_k)^\top\boldsymbol{R}_k\|_2\leq C\sqrt{\frac{n_k}{p}+\zeta_{t,\mathcal{L}}}\left(\sqrt{\frac{n_k}{p}}+\frac1{\sqrt{n_k}}\right)\leq C\mathcal{V}_{n_k,p}(t).
\end{align*}
Since $\hat{\boldsymbol{U}}_k^\top\hat{\boldsymbol{F}}_k=0$ and $\|\boldsymbol{H}_k^{-1}\|_2\leq C$,
\begin{align*}
    \hat{\boldsymbol{U}}_k^\top\boldsymbol{F}_k=-\hat{\boldsymbol{U}}_k^\top\boldsymbol{R}_k\boldsymbol{H}_k^{-\top}.
\end{align*}
If $\zeta_{t,\mathcal{L}}\geq n_k$, then
$\hat{\boldsymbol{U}}_k^\top\boldsymbol{F}_k
=-\hat{\boldsymbol{U}}_k^\top\boldsymbol{R}_k\boldsymbol{H}_k^{-\top}$.
By Lemma~\ref{lemma1},
$$\max_{j\in[p]}\|[\hat{\boldsymbol{U}}_k^\top\boldsymbol{F}_k]_{j\cdot}\|_2\lesssim \sqrt{\zeta_{t,\mathcal{L}}}\bigl(\sqrt{n_k/p}+n_k^{-1/2}\bigr)\leq \mathcal{V}_{n_k,p}(t,\mathcal{L}).$$
This proves the first assertion.

\textbf{Step 2: Bounds for $\|(\hat{\boldsymbol{U}}_k-\boldsymbol{U}_k)^\top\boldsymbol{\mathcal{E}}^{(k)}\|_\infty$.} This corresponds to Lemma C.4 of \cite{fan2024latent}. Since $\boldsymbol{\mathcal{E}}^{(k)}$ is independent
of $\boldsymbol{F}_k$ and $\boldsymbol{U}_k$, conditional on $\boldsymbol{F}_k$ and
$\boldsymbol{U}_k$, the sub-Gaussian concentration inequality and a union bound
yield that
\begin{align*}
    \|\hat{\boldsymbol{F}}_k^\top\boldsymbol{\mathcal{E}}^{(k)}\|_2\leq C\sqrt{n_k(t+\log(e|\mathcal{L}|))},\qquad\|\boldsymbol{R}_k^\top\boldsymbol{\mathcal{E}}^{(k)}\|_2\leq C\|\boldsymbol{R}_k\|_{\mathrm{F}}\sqrt{t+\log(e|\mathcal{L}|)}
\end{align*}
simultaneously over all $k\in\mathcal{L}$ on the event established in Lemma \ref{lemma1}, except on an additional event of probability at most $Ce^{-t}$. Here we used $\|\hat{\boldsymbol{F}}_k\|_2=\sqrt{n_k}$ and $r_k\leq r_*=O(1)$.

Recall that
\begin{align*}
    \hat{\boldsymbol{U}}_k-\boldsymbol{U}_k=-\boldsymbol{R}_k\boldsymbol{H}_k^{-\top}\boldsymbol{B}_k^\top-\hat{\boldsymbol{F}}_k(\hat{\boldsymbol{B}}_k-\boldsymbol{B}_k\boldsymbol{H}_k^{-1})^\top.
\end{align*}
Consequently,
\begin{align*}
    \|(\hat{\boldsymbol{U}}_k-\boldsymbol{U}_k)^\top\boldsymbol{\mathcal{E}}^{(k)}\|_\infty\leq\max_{j\in[p]}\|[\boldsymbol{B}_k]_{j\cdot}\|_2\|\boldsymbol{H}_k^{-1}\|_2\|\boldsymbol{R}_k^\top\boldsymbol{\mathcal{E}}^{(k)}\|_2+\max_{j\in[p]}\|[\hat{\boldsymbol{B}}_k-\boldsymbol{B}_k\boldsymbol{H}_k^{-1}]_{j\cdot}\|_2\|\hat{\boldsymbol{F}}_k^\top\boldsymbol{\mathcal{E}}^{(k)}\|_2.
\end{align*}
By Lemma \ref{lemma1}, we obtain
\begin{align*}
    \|(\hat{\boldsymbol{U}}_k-\boldsymbol{U}_k)^\top\boldsymbol{\mathcal{E}}^{(k)}\|_\infty\leq& C\left(\sqrt{\frac{n_k}{p}}+\frac1{\sqrt{n_k}}\right)\sqrt{t+\log(e|\mathcal{L}|)}+C\sqrt{\frac1p+\frac{\zeta_{t,\mathcal{L}}}{n_k}}\sqrt{n_k(t+\log(e|\mathcal{L}|))}\\
    \leq& C\sqrt{\left(\frac{n_k}{p}+\zeta_{t,\mathcal{L}}\right)(t+\log(e|\mathcal{L}|))},
\end{align*}
where the last inequality follows from
$\zeta_{t,\mathcal{L}}\geq1$.

\textbf{Step 3: Bounds for $\left\|\hat{\boldsymbol{U}}_k^\top\hat{\boldsymbol{U}}_k-\boldsymbol{U}_k^\top\boldsymbol{U}_k\right\|_{\max}$.} Following the proof of Lemma C.6 of \cite{fan2024latent},
\begin{align*}
    \|\hat{\boldsymbol{U}}_k^\top\hat{\boldsymbol{U}}_k-\boldsymbol{U}_k^\top\boldsymbol{U}_k\|_{\max}\leq&2\|\hat{\boldsymbol{U}}_k^\top(\hat{\boldsymbol{U}}_k-\boldsymbol{U}_k)\|_{\max}+\max_{j\in[p]}\|[\hat{\boldsymbol{U}}_k-\boldsymbol{U}_k]_{\cdot j}\|_2^2\\
    \leq&C\mathcal{V}_{n_k,p}(t)+C\left(\frac{n_k}{p}+\zeta_{t,\mathcal{L}}\right)\leq C\left(\frac{n_k}{p}+\zeta_{t,\mathcal{L}}\right).
\end{align*}
Here we used
\begin{align*}
    \hat{\boldsymbol{U}}_k^\top
       (\hat{\boldsymbol{U}}_k-\boldsymbol{U}_k)
    =\hat{\boldsymbol{U}}_k^\top
       \boldsymbol{F}_k\boldsymbol{B}_k^\top,
\end{align*}
Combining the events above completes the proof.
\end{proof}

\begin{lemma}\label{lemmaA7}
    Suppose that Assumptions \ref{assum1}-\ref{assum2} hold for $k\in\{0,1,\ldots,K\}$ and Assumption \ref{assum4} (a) holds. Under the events with $\mathcal{L}=\{0\}$ in Lemma \ref{lemmaA6}, for each $r\in\{1,2,3\}$ and any $\boldsymbol{w}_1,\boldsymbol{w}_2\in\mathbb{R}^{p}$, with probability at least $1-2p^{-100}$, we have 
    \begin{align*}
        &|\hat{L}_0^{[r]}(\boldsymbol{w}_1)-\hat{L}_0^{[r]}(\boldsymbol{w}_2)-(L_0^{[r]}(\boldsymbol{w}_1)-L_0^{[r]}(\boldsymbol{w}_2))|\\
        \lesssim&\sqrt{\frac{\log p}{n_0}}\|\boldsymbol{w}_1-\boldsymbol{w}_2\|_1+\sqrt{\frac{\log p}{n_0}}\|\boldsymbol{w}_1-\boldsymbol{w}_2\|_2\big(\|\boldsymbol{w}_1-\boldsymbol{w}_2\|_2+\|\boldsymbol{w}_2-\boldsymbol{\beta}\|_2\big)\\
        &+\Big(\frac{1}{p}+\frac{\log p}{n_0}\Big)\|\boldsymbol{w}_1-\boldsymbol{w}_2\|_1\big(\|\boldsymbol{w}_1-\boldsymbol{w}_2\|_1+\|\boldsymbol{w}_2-\boldsymbol{\beta}\|_1\big),
    \end{align*}
    where $$L_0^{[r]}(\boldsymbol{w})=\frac{1}{2}\mathbb{E}\left[\frac{1}{n_0/3}\|\mathrm{diag}(\{\mathds{1}(i\in R^{[r]}\}_{i=1}^{n_0})(\boldsymbol{Y}^{(0)}-\boldsymbol{U}_0\boldsymbol{w}-\boldsymbol{F}_0\boldsymbol{\gamma}^{(0)})\|_2^2\right]$$
\end{lemma}

\begin{proof}
    Firstly, we have 
    \begin{align*}
        &\hat{L}_0^{[r]}(\boldsymbol{w}_1)-\hat{L}_0^{[r]}(\boldsymbol{w}_2)=\nabla\hat{L}_0^{[r]}(\boldsymbol{w}_2)^\top(\boldsymbol{w}_1-\boldsymbol{w}_2)+\frac{1}{2}(\boldsymbol{w}_1-\boldsymbol{w}_2)^\top\frac{1}{n_0/3}\hat{\boldsymbol{U}}_{0}^{[r]\top}\hat{\boldsymbol{U}}_{0}^{[r]}(\boldsymbol{w}_1-\boldsymbol{w}_2)\\
        =&\nabla\hat{L}_0^{[r]}(\boldsymbol{\beta})^\top(\boldsymbol{w}_1-\boldsymbol{w}_2)+(\boldsymbol{w}_2-\boldsymbol{\beta})^\top\frac{1}{n_0/3}\hat{\boldsymbol{U}}_{0}^{[r]\top}\hat{\boldsymbol{U}}_{0}^{[r]}(\boldsymbol{w}_1-\boldsymbol{w}_2)\\
        &+\frac{1}{2}(\boldsymbol{w}_1-\boldsymbol{w}_2)^\top\frac{1}{n_0/3}\hat{\boldsymbol{U}}_{0}^{[r]\top}\hat{\boldsymbol{U}}_{0}^{[r]}(\boldsymbol{w}_1-\boldsymbol{w}_2)
    \end{align*}
    Similarly, apply this equation to $L_0^{[r]}(\boldsymbol{w}_1)-L_0^{[r]}(\boldsymbol{w}_2)$, together with the equation above yields that 
    \begin{align*}
        &\hat{L}_0^{[r]}(\boldsymbol{w}_1)-\hat{L}_0^{[r]}(\boldsymbol{w}_2)-(L_0^{[r]}(\boldsymbol{w}_1)-L_0^{[r]}(\boldsymbol{w}_2))\\
        =&(\nabla\hat{L}_0^{[r]}(\boldsymbol{\beta})-\nabla L_0^{[r]}(\boldsymbol{\beta}))^\top(\boldsymbol{w}_1-\boldsymbol{w}_2)+(\boldsymbol{w}_2-\boldsymbol{\beta})^\top\left(\frac{1}{n_0/3}\hat{\boldsymbol{U}}_{0}^{[r]\top}\hat{\boldsymbol{U}}_{0}^{[r]}-\boldsymbol{\Sigma}_{\boldsymbol{u}}^{(0)}\right)(\boldsymbol{w}_1-\boldsymbol{w}_2)\\
        &+\frac{1}{2}(\boldsymbol{w}_1-\boldsymbol{w}_2)^\top\left(\frac{1}{n_0/3}\hat{\boldsymbol{U}}_{0}^{[r]\top}\hat{\boldsymbol{U}}_{0}^{[r]}-\boldsymbol{\Sigma}_{\boldsymbol{u}}^{(0)}\right)(\boldsymbol{w}_1-\boldsymbol{w}_2):=E_1+E_2+E_3.
    \end{align*}
    By Lemmas \ref{lemmaA5}--\ref{lemmaA6} and Assumption 4,
    \begin{align*}
        \|\nabla\hat{L}_0^{[r]}(\boldsymbol{\beta})\|_\infty\leq C\sqrt{\frac{\log p}{n_0}},~~~\left\|\frac{\boldsymbol{U}_{0}^{[r]\top}\boldsymbol{U}_{0}^{[r]}}{n_0/3}-\frac{\hat{\boldsymbol{U}}_0^{[r]\top}\hat{\boldsymbol{U}}_0^{[r]}}{n_0/3}\right\|_{\max}\leq C\left(\frac{1}{p}+\frac{\log p}{n_0}\right)
    \end{align*}
    with probability at least $1-p^{-100}$, where we use $n_0=O(p\log p)$ and $\log p=O(n_0)$. Moreover, since $\boldsymbol{a}^\top\boldsymbol{u}^{(0)}\boldsymbol{u}^{(0)\top}\boldsymbol{b}$ is sub-exponential with parameter at most $\nu_0^2\|a\|_2\|b\|_2$ for any vectors $\boldsymbol{a}$ and $\boldsymbol{b}$, by the Bernstein's inequality \citep[Theorem 2.8.1]{vershynin2018high}, one attains that 
    \begin{align*}
        \mathbb{P}\left\{\left|\boldsymbol{a}^\top\left(\frac{1}{n_0/3}\boldsymbol{U}_{0}^{[r]\top}\boldsymbol{U}_{0}^{[r]}-\boldsymbol{\Sigma}_{\boldsymbol{u}}^{(0)}\right)\boldsymbol{b}\right|\leq C_0\nu_0^2\|\boldsymbol{a}\|_2\|\boldsymbol{b}\|_2\left(\sqrt{\frac{u}{n_0}}+\frac{u}{n_0}\right)\right\}\geq 1-2e^{-u}
    \end{align*}
    Consequently, choosing $u=100\log p+\log 2$ for any fixed $\boldsymbol{w}_1,\boldsymbol{w}_2\in\mathbb{R}^p$ and using $\log p\lesssim n_0$, under the events with $\mathcal{L}=\{0\}$ in Lemma \ref{lemmaA6}, with probability at least $1-2p^{-100}$, we have 
    \begin{align*}
    &\left|\hat{L}_0^{[r]}(\boldsymbol{w}_1)-\hat{L}_0^{[r]}(\boldsymbol{w}_2)-\big(L_0^{[r]}(\boldsymbol{w}_1)-L_0^{[r]}(\boldsymbol{w}_2)\big)\right|\\
    \lesssim&\sqrt{\frac{\log p}{n_0}}\|\boldsymbol{w}_1-\boldsymbol{w}_2\|_1+\sqrt{\frac{\log p}{n_0}}\|\boldsymbol{w}_1-\boldsymbol{w}_2\|_2\big(\|\boldsymbol{w}_1-\boldsymbol{w}_2\|_2+\|\boldsymbol{w}_2-\boldsymbol{\beta}\|_2\big)\\
    &+\Big(\frac{1}{p}+\frac{\log p}{n_0}\Big)\|\boldsymbol{w}_1-\boldsymbol{w}_2\|_1\big(\|\boldsymbol{w}_1-\boldsymbol{w}_2\|_1+\|\boldsymbol{w}_2-\boldsymbol{\beta}\|_1\big).
\end{align*}
\end{proof}

\begin{lemma}[Corollary 2.1 of \cite{chernozhukov2013gaussian}]\label{LemmaA8}
        Let $\boldsymbol{X}_1,\cdots,\boldsymbol{X}_n=(X_{i1},\cdots,X_{id})^\top$ be n i.i.d. $d$-dimensional random vectors with mean $\mathbf{0}_d$. Suppose that there are some constants $0<C_1<C_2$ such that $\mathbb{E}[X_{ij}^2]\geq C_1$  and $X_{ij}$ is sub-exponential with $\|X_{ij}\|_{\psi_1}\leq C_2$ for all $1\leq j\leq d$. If $\log^7 (dn)/n=o(1)$, then
        $$\lim_{n\to\infty}\sup_{t\in\mathbb{R}}\left|\mathbb{P}\left(\left\|\frac{1}{\sqrt{n}}\sum_{i=1}^{n}\boldsymbol{X}_i\right\|_\infty\leq t\right)-\mathbb{P}(\|\boldsymbol{N}\|_\infty\leq t)\right|=0,$$
        where $\boldsymbol{N}\sim \mathbb{N}(\mathbf{0}_d,\mathbb{E}[\boldsymbol{X}_i\boldsymbol{X}_i^\top])$.
    \end{lemma}

\begin{lemma}[Lemma 7 of \cite{cai2025statistical}]\label{LemmaA9}
    For any $\delta>0$ and sequences of random vectors $\{ X_n \}_{n \in \mathbb{N}}, \{ Y_n \}_{n \in \mathbb{N}}, \{ Z_n \}_{n \in \mathbb{N}}$ satisfying $$\sup_{t \in \mathbb{R}} \left| \mathbb{P}(\|X_n\|_\infty\leq t)-\mathbb{P}(\|Z_n\|_\infty\leq t) \right|=\Delta_n,$$ we have that
    $$\sup_{t\in\mathbb{R}}\left|\mathbb{P}(\|X_n+Y_n\|_\infty\leq t)-\mathbb{P}(\|Z_n\|_\infty\leq t)\right|\leq \mathbb{P}(\|Y_n\|_\infty>\delta)+\sup_{t\in\mathbb{R}}\mathbb{P}(|\|Z_n\|_\infty-t|<\delta)+\Delta_n.$$
    \end{lemma}

\section{Additional Simulation Results}\label{secC}
\subsection{Simulation Results for the Case Satisfying Assumption \ref{assum4} (c)}
We consider an additional setting motivated by Assumption 4. In each domain, the idiosyncratic covariates $\boldsymbol{u}^{(k)}$ are generated from $\mathbb{N}(\boldsymbol{0}_p,\boldsymbol{\Sigma})$ or multivariate t distribution $t_{10}(\boldsymbol{0},\boldsymbol{\Sigma})$ with common $\Sigma_{ij}=0.5^{|i-j|}$, and the errors are generated from standard normal distribution or $t_5$ distribution. For each informative source, we set $\boldsymbol{w}^{(k)}-\boldsymbol\beta=(\eta/p)\boldsymbol{R}^{(k)}$, where $\boldsymbol{R}^{(k)}$ has independent Rademacher entries and $\eta\in\{5,10\}/\sqrt{n_0}$. For each noninformative source, we perturb one coordinate chosen uniformly outside the support of $\boldsymbol\beta$ by $\pm\sqrt{\log p}$ with equal probability. We use the sample sizes, dimensions and estimation procedures of the main simulation and vary the number of informative sources from zero to ten. The results are shown in Figure \ref{figC1}:

\begin{figure}[!htbp]
    \begin{subfigure}{\textwidth}
        \centering
        \begin{subfigure}{0.47\textwidth}
            \centering
            \includegraphics[width=\textwidth]{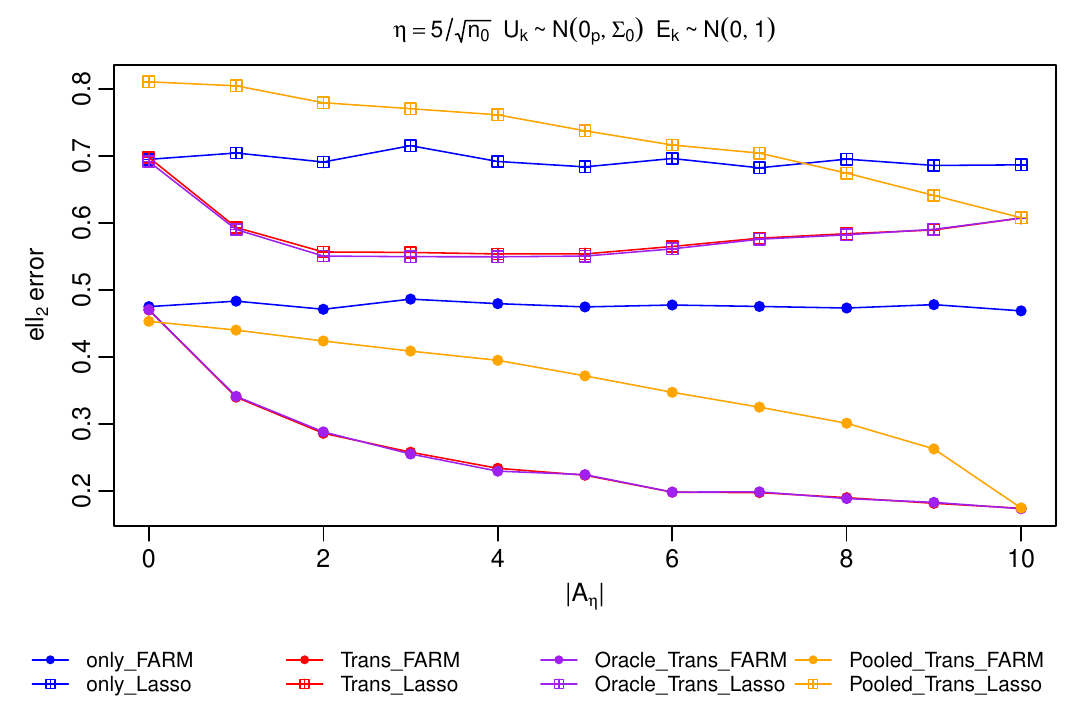}
        \end{subfigure}
        \hfill
        \begin{subfigure}{0.47\textwidth}
            \centering
            \includegraphics[width=\textwidth]{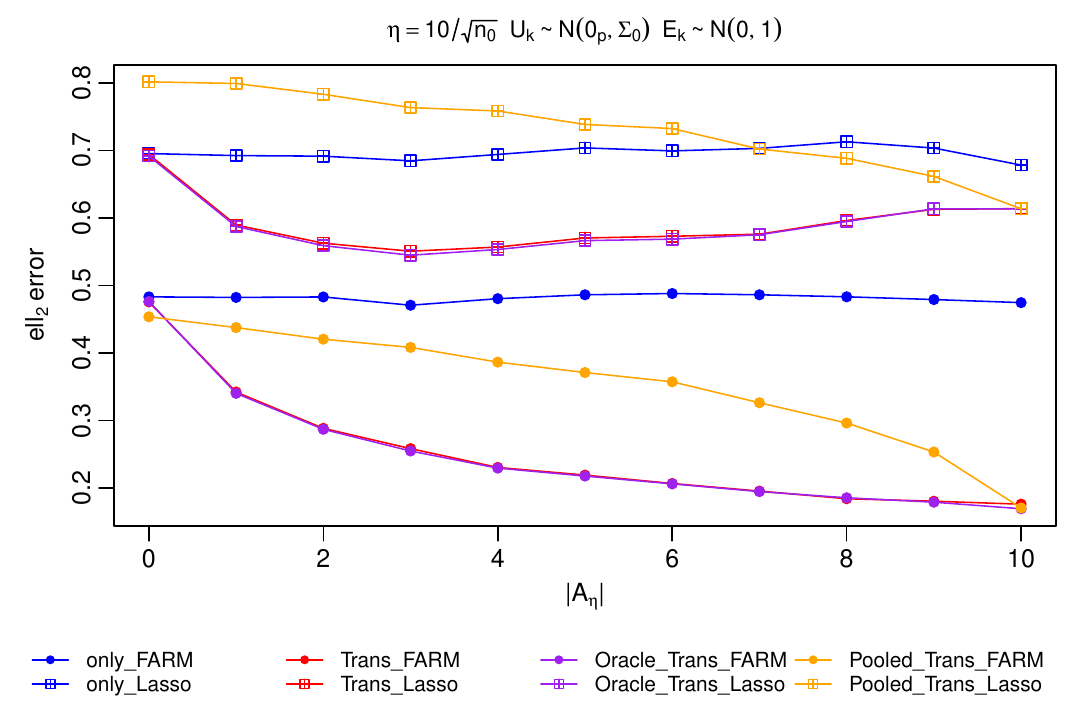}
        \end{subfigure}
    \end{subfigure}

    \begin{subfigure}{\textwidth}
        \centering
        \begin{subfigure}{0.47\textwidth}
            \centering
            \includegraphics[width=\textwidth]{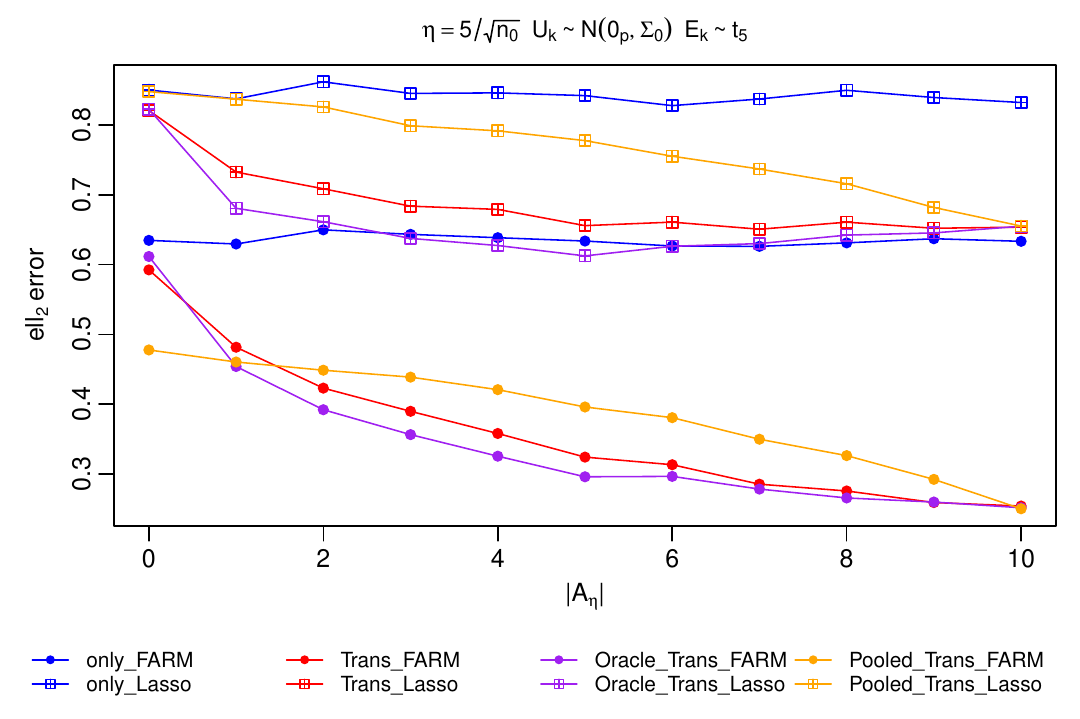}
        \end{subfigure}
        \hfill
        \begin{subfigure}{0.47\textwidth}
            \centering
            \includegraphics[width=\textwidth]{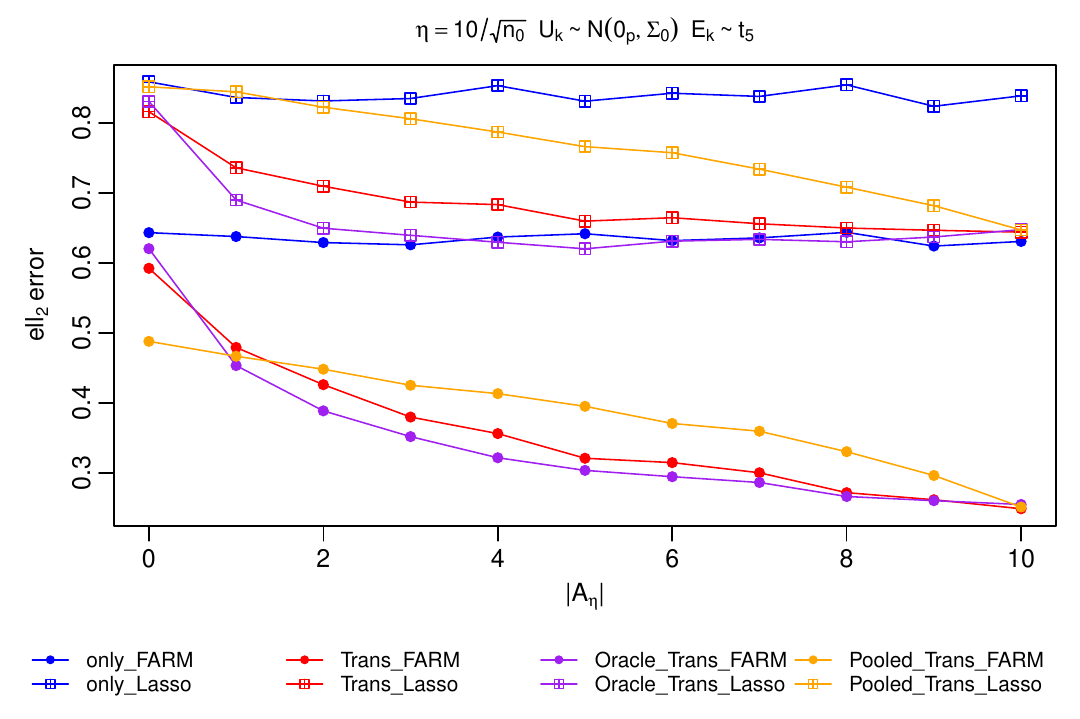}
        \end{subfigure}
    \end{subfigure}

    \begin{subfigure}{\textwidth}
        \centering
        \begin{subfigure}{0.47\textwidth}
            \centering
            \includegraphics[width=\textwidth]{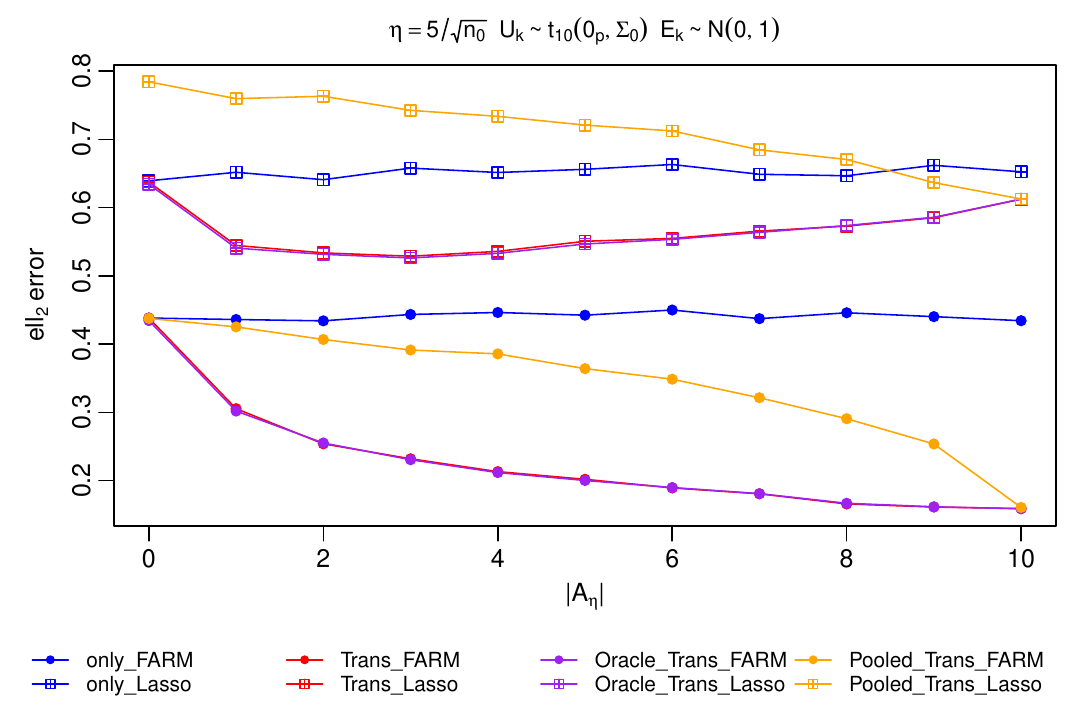}
        \end{subfigure}
        \hfill
        \begin{subfigure}{0.47\textwidth}
            \centering
            \includegraphics[width=\textwidth]{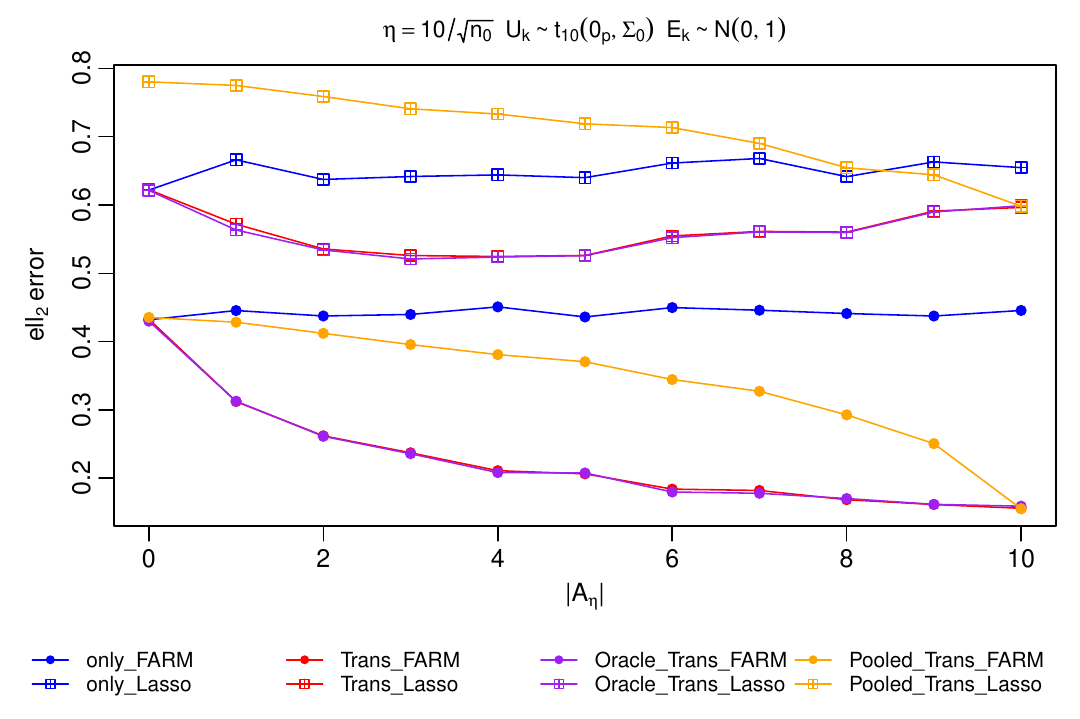}
        \end{subfigure}
    \end{subfigure}

    \begin{subfigure}{\textwidth}
        \centering
        \begin{subfigure}{0.47\textwidth}
            \centering
            \includegraphics[width=\textwidth]{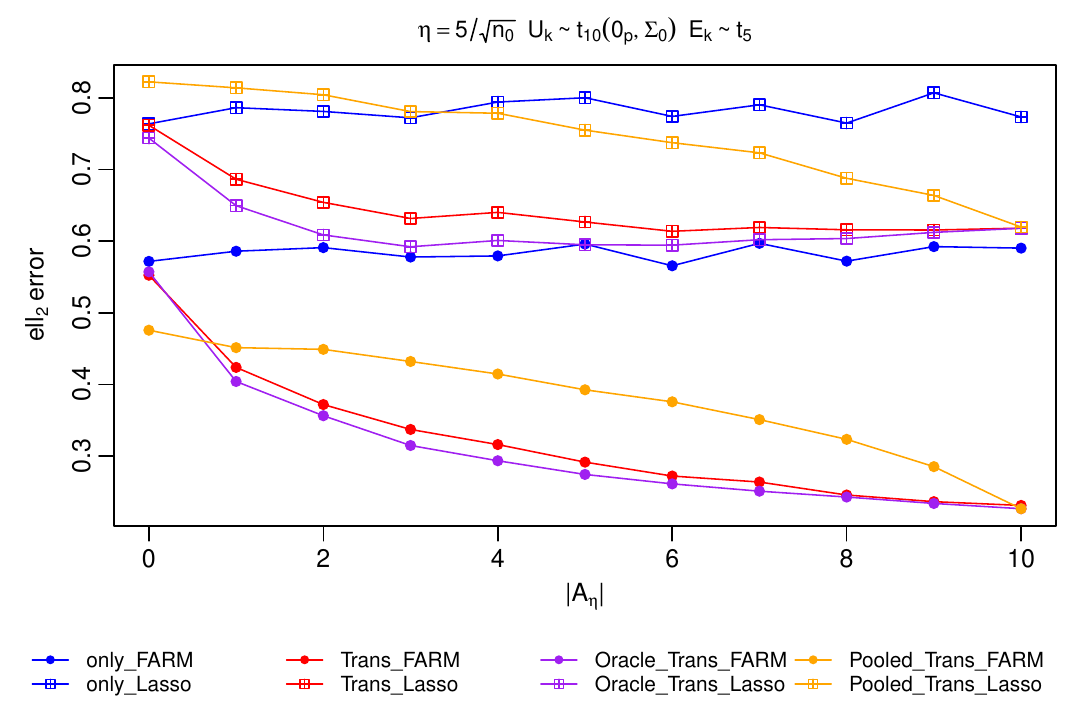}
        \end{subfigure}
        \hfill
        \begin{subfigure}{0.47\textwidth}
            \centering
            \includegraphics[width=\textwidth]{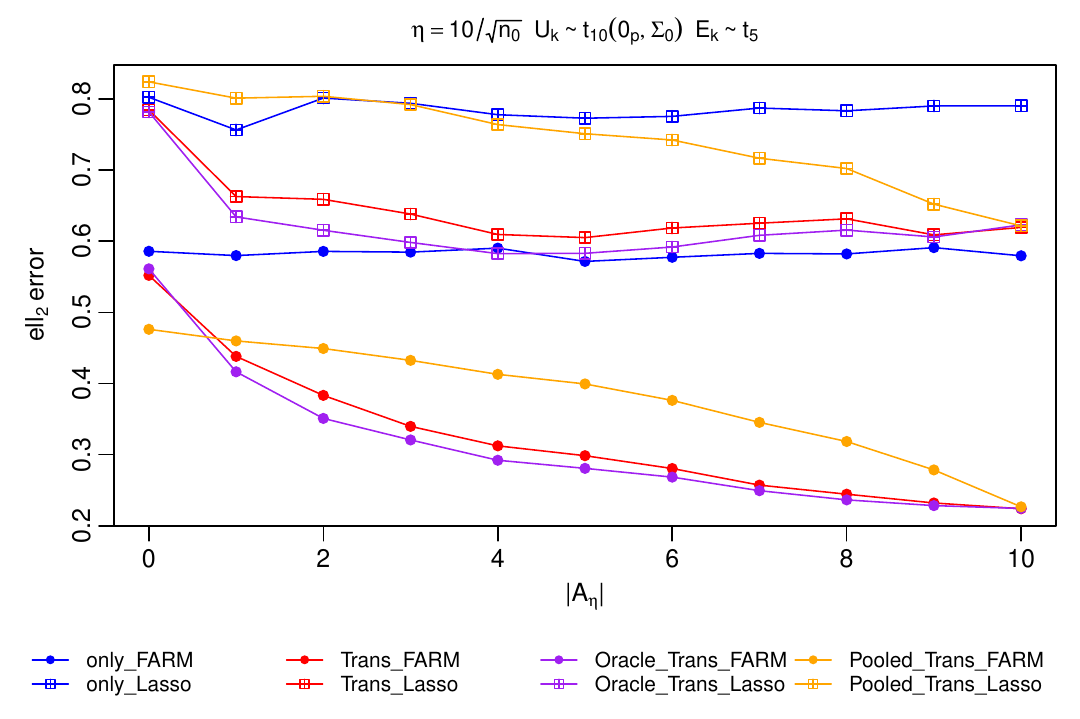}
        \end{subfigure}
    \end{subfigure}
    
    \captionsetup{font=footnotesize}
    \caption{Comparison of $\ell_2$ errors averaged over $200$ replications versus the number of informative sources ($|\mathcal{A}_\eta|$) across four cases with $\eta\in\{5/\sqrt{n_0},10/\sqrt{n_0}\}$, $K=10$, $n_0=n_1=\ldots=n_K=300$, $p=500$, and $s=20$.}
    \label{figC1}
\end{figure}
As in the main experiment, the estimation error of Trans-FARM generally decreases with the number of informative sources and remains close to that of Oracle-Trans-FARM. One difference is that, under $t_5$ errors, Pooled-Trans-FARM has smaller estimation error than Trans-FARM when no source is informative. This may be because each source differs from the target at only one randomly chosen inactive coordinate, so pooling dilutes these shifts and may offer greater gains in estimating the shared coefficients when the error variance is higher.

\subsection{Additional Results for Non-studentized Trans-FARM Intervals}\label{secC.2}
We supplement Section \ref{sec5.2} with results for the non-studentized Trans-FARM intervals established in Corollary \ref{cor1}, denoted by Trans-FARM2. The simulation settings are unchanged, and the results for Only-FARM and Trans-FARM are reproduced from Section \ref{sec5.2} for comparison.
\begin{table*}[htbp]
	\caption{Comparison of ACP and AL averaged over $500$ replications for $p,n_0\in\{200,500\}$.}
    \centering
	\resizebox{0.9\linewidth}{!}{
		\begin{tabular}{ccccccccccccc}
			\hline
			\multirow{2}{*}{Methods}&\multirow{2}{*}{Case}&\multicolumn{2}{c}{$p=200,n_0=200$}& &\multicolumn{2}{c}{$p=500,n_0=200$}& &\multicolumn{2}{c}{$p=200,n_0=500$}& &\multicolumn{2}{c}{$p=500,n_0=500$}\cr
            & &ACP&AL& &ACP&AL& &ACP&AL& &ACP&AL\cr
			\hline
			Only-FARM&\multirow{3}{*}{(a)}&0.736&0.554& &0.800&0.586& &0.946&1.794& &0.838&0.376\cr
			Trans-FARM& &0.942&0.548& &0.952&0.577& &0.938&2.070& &0.938&0.374\cr
            Trans-FARM2& &0.946&0.552& &0.976&0.585& &0.964&1.792& &0.960&0.375\cr
            Only-FARM&\multirow{3}{*}{(b)}&0.646&0.702& &0.680&0.749& &0.924&2.298& &0.788&0.481\cr
			Trans-FARM& &0.930&0.694& &0.958&0.737& &0.956&2.663& &0.950&0.479\cr
            Trans-FARM2& &0.946&0.705& &0.952&0.750& &0.952&2.288& &0.946&0.482\cr
            Only-FARM&\multirow{3}{*}{(c)}&0.728&0.503& &0.776&0.535& &0.946&1.820& &0.806&0.338\cr
			Trans-FARM& &0.952&0.494& &0.976&0.519& &0.968&2.084& &0.956&0.335\cr
            Trans-FARM2& &0.946&0.503& &0.966&0.533& &0.948&1.827& &0.960&0.338\cr
            Only-FARM&\multirow{3}{*}{(d)}&0.618&0.638& &0.696&0.686& &0.944&2.318& &0.700&0.433\cr
			Trans-FARM& &0.932&0.628& &0.944&0.665& &0.962&2.646& &0.936&0.431\cr
            Trans-FARM2& &0.932&0.638& &0.962&0.682& &0.972&2.322& &0.954&0.432\cr
			\hline
		\end{tabular}
	}
	\label{tab:nonstudentized}
\end{table*}

Table~\ref{tab:nonstudentized} shows that Trans-FARM2 generally achieves coverage closer to the nominal level than Only-FARM, with nearly identical ALs. Except at $p=200$ and $n_0=500$, Trans-FARM (studentized Trans-FARM) yields slightly shorter intervals, consistent with its adjustment for coefficient-specific scales. In that setting, Trans-FARM2 instead produces shorter intervals while maintaining near-nominal coverage. These results show that the inferential benefits of transfer learning also extend to the non-studentized procedure. With $10$-fold nodewise cross-validation, when $0.9n_0\geq p-\hat{r}_0$, the training folds may span the full row space of the PCA residual matrix, assuming its rank is $p-\hat{r}_0$, where $\hat{r}_0$ is the estimated number of factors for the target dataset. Exact training fits then also hold on validation folds, potentially favoring very small penalties and unstable scale estimates. This may explain the longer studentized intervals at $p=200$ and $n_0=500$. Our results therefore suggest considering non-studentized Trans-FARM in this regime when such instability occurs, while favoring studentized Trans-FARM for its shorter intervals in the other settings considered.

\end{appendices}

\end{document}